\documentclass[prx,aps,showpacs,twocolumn,superscriptaddress,nofootinbib,accepted=2019-10-24]{quantumarticle}

\usepackage[utf8]{inputenc}
\usepackage[english]{babel}
\usepackage{amsthm}
\newtheorem{theorem}{Theorem}[section]
\newtheorem{corollary}{Corollary}[theorem]
\newtheorem{lemma}[theorem]{Lemma}
\newtheorem{proposition}[theorem]{Proposition}

\newtheorem{example}[theorem]{Example}
\usepackage{xcolor}
\usepackage{setspace}
\usepackage{amsmath,bbm}
\usepackage{amsfonts}
\usepackage{mathtools}
\usepackage{amssymb}
\usepackage{amsthm}
\usepackage[bottom]{footmisc}
\usepackage{slashed}
\usepackage{graphicx}
\usepackage[utf8]{inputenc}
\usepackage[hidelinks]{hyperref}
\usepackage[T1]{fontenc}

\usepackage{booktabs}
\usepackage{subfigure}
\usepackage{verbatim}
\DeclareMathOperator{\Tr}{Tr}
\usepackage{graphicx}

\usepackage{bm}

\usepackage[normalem]{ulem}
\usepackage{dsfont}

\usepackage{hyperref}

\usepackage[numbers,sort&compress]{natbib}

\usepackage{natbib}

\def\>{\rangle}
\def\<{\langle}

\def\E{ {\mathcal E} }

\def\H{ {\mathcal H} }
\def\M{ {\mathcal M} }

\def\O{ {\mathcal O} }

\def\D{ {\mathcal D} }

\def\I{ \mathbbm{1} }
\def\tr{ \mbox{tr} }

\def\p{\boldsymbol{p}}
\def\q{\boldsymbol{q}}
\def\x{\boldsymbol{x}}

\begin{document}
\pagenumbering{arabic}

\title{Decomposable coherence and quantum fluctuation relations}
\author{Erick Hinds Mingo}
\affiliation{Controlled Quantum Dynamics Theory Group, Imperial College London, Prince Consort Road, London SW7 2BW, UK}
\author{David Jennings}
\affiliation{Controlled Quantum Dynamics Theory Group, Imperial College London, Prince Consort Road, London SW7 2BW, UK}
\affiliation{Department of Physics, University of Oxford, Oxford, OX1 3PU, UK}
\affiliation{School of Physics and Astronomy, University of Leeds, Leeds, LS2 9JT, UK.}
\date{\today}

\begin{abstract}

In Newtonian mechanics, any closed-system dynamics of a composite system in a microstate will leave all its individual subsystems in distinct microstates, however this fails dramatically in quantum mechanics due to the existence of quantum entanglement. Here we introduce the notion of a `coherent work process', and show that it is the direct extension of a work process in classical mechanics into quantum theory. This leads to the notion of `decomposable' and `non-decomposable' quantum coherence and gives a new perspective on recent results in the theory of asymmetry as well as early analysis in the theory of classical random variables. Within the context of recent fluctuation relations, originally framed in terms of quantum channels, we show that coherent work processes play the same role as their classical counterparts, and so provide a simple physical primitive for quantum coherence in such systems. We also introduce a pure state effective potential as a tool with which to analyze the coherent component of these fluctuation relations, and which leads to a notion of temperature-dependent mean coherence, provides connections with multi-partite entanglement, and gives a hierarchy of quantum corrections to the classical Crooks relation in powers of inverse temperature.
\end{abstract}

\maketitle

\hspace{1cm}
\section{Introduction}

The superposition principle is at the core of what makes quantum mechanics so special \cite{Dirac,Peres,TheoryExperiments,SchrodingerCat}. Classically a particle might be located in a  microstate at any one of a number of spatial sites with sharp momentum, however quantum mechanics allows a fundamentally new kind of state -- the particle being a superposition of multiple different locations $x_1, x_2, \dots , x_n$. For this, the particle is in a state $|\psi\> = \sum_k e^{i\theta_k} \sqrt{p_k} |x_k\>$ where $\theta_k$ are phase angles, and $p_k$ is the probability of a measurement of position returning the classical outcome $x_k$. While the measurement outcomes are random, the state $|\psi\>$ is a perfectly sharp state (a pure state, or state of maximal knowledge) and so should be viewed exactly on a par with a classical state of well-defined location $|x_i\>$, say. More precisely, the structure of the state space in any physical theory determines many of the distinct characteristics of the particular theory. State spaces are always convex sets determined entirely by the extremal points of the set -- the pure states of the physical theory. The admissible measurements that the theory allows are defined in relation to this state space (see \cite{ChiribellaSpekkens} for more details) and so are secondary theoretical ingredients. Therefore a comparison of classical and quantum mechanics at a fundamental level should equate sharp microstates with pure quantum states $|\psi\>$, as opposed to say comparing measurement statistics.

In classical mechanics a system with initial position and conjugate momentum $(x_0,p_0)$ evolves in time along a well-defined trajectory $(x(t), p(t))$ in phase space determined by the Hamiltonian of the system. This evolution extends into quantum mechanics where a quantum state evolves under a unitary transformation. The unitary dynamics of such a system $S$ can also be understood in terms of paths, however it is now described in terms of a path integral \cite{Feynman} where the transition amplitudes are given by integrating the functional $\exp[\frac{i}{\hbar}\int dt (p(t)\dot{x}(t)- H_S[x(t),p(t)])]$ over all paths $(x(t), p(t))$ in phase space consistent with the boundary conditions, and where $H_S [x(t),p(t)]$ is the classical Hamiltonian for the system $S$. At the operator level this dynamics is described a unitary transformation on states $|\psi\> \rightarrow U(t) |\psi\>$, for some unitary operator $U(t)$. Thus, in both classical and quantum mechanics the Hamiltonian plays a key role in the time evolution of a system. Moreover in quantum mechanics, if the system is energetically closed then we have $[U,H_S]=0$ for the unitary evolution of the system. In both classical and quantum mechanics an initial sharp (pure) state is evolved to a final sharp state.

The path integral perspective describes this unitary evolution as a sum over all consistent paths, and so receives contributions from trajectories that respect neither energy conservation nor the classical equations of motion. However, in the ``limit of $\hbar \rightarrow 0$'' a stationary phase argument tells us that the dominant contributions to the evolution come from those trajectories around the classical phase space trajectory $(x_{\rm class}(t), p_{\rm class}(t))$, namely the one that obeys the Hamiltonian equations of motion \cite{Goldstein} given by $\dot{p}= -\partial_x H_S, \dot{x} = \partial_p H_S$. Thus we recover classical mechanics in the limit.

In the case of two quantum systems $S$ and $A$, exactly the same formalism applies, and coherent dynamics of the composite quantum system can be analysed without any further fuss -- governed by a total Hamiltonian $H_{SA}$ for the joint system $SA$. However, we can now ask the following question: 
\begin{center}
\textit{What unitary transformations of a composite quantum system $SA$ in a state $|\psi\>_{SA}$ are possible that (a) obey energy conservation over the system $SA$, but (b) give rise to a transformation of the system $S$ from an initial pure state and into a final pure state?}
\end{center}
It is clear that in the classical limit the answer to this trivial with sharp microstates -- all classical dynamics that conserve the energy $H_{SA}$ give rise to deterministic transformations of $S$. However for coherent quantum systems, which can become entangled with each other, the answer is less obvious. Indeed, as we shall see, the above is directly linked to central results from the resource theory of asymmetry \cite{Harmonic,ReferenceFrame,GourSpekkens,Vaccaro,Cristina} and foundational results from the 1930s on classical random variables \cite{Lukacs,CharacteristicFunctions,PoissonDistribution,ArithmeticProbability}.

The physical motivation for this is given by considering the elementary notion of a mechanical work process on a classical system along some trajectory. For this, a mechanical system $S$ initially in some definite phase space state $(\x_0,\p_0)$ is evolved deterministically to some final state $(\x_1,\p_1)$ and the mechanical work $w$ for the process is computed from $w=\int_{\rm path} \boldsymbol{F}\cdot \boldsymbol{dx}$, where $\boldsymbol{F}(t)$ is the force exerted on the system along the particular path \cite{Goldstein}. However, the system in this case is not energetically isolated, and so the quantity of energy $w$ is only meaningful because it corresponds to an equal and opposite change in energy of some external classical degree of freedom $A$ that is used to induce the process. For example, we could imagine some external weight of mass $m$ that is initialised at some definite height $h_0=0$ and finishes deterministically in some final height $h_1$. In order to associate the height value $h_1$ for $A$ to $w$ the work done in the process on $S$ (via $w=mgh_1$) it is necessary that energy is exactly conserved between the two systems $S$ and $A$. However, if we now wish to consider \emph{superpositions} of work processes on $S$ then this account becomes non-trivial, and so provides the starting point for our analysis.

In this paper we shall address the above question and show that the classical notion of deterministic mechanical work can be extended into quantum mechanics in a way that exactly parallels the classical case, and gives rise to connections with topics in contemporary quantum physics. Firstly, it leads to a notion of decomposable and non-decomposable coherence and has an immediate description in terms of recent frameworks for quantum coherence \cite{Harmonic,ReferenceFrame,GourSpekkens,Vaccaro,WinterCoherence,BenCoherenceFluctuation}. We then extend these considerations in noisy quantum environments and show that this notion of a coherent work process connects naturally with fluctuation relations \cite{Crooks,Evans,Jarzynski2,ColloquiumFluctuation,Aberg,Alhambra,Zoe,HyukJoon}. The coherent fluctuation theorems we develop go beyond traditional relations, such as the Crooks relation, and allow the analysis of quantum coherence in a noisy thermal environment. Because quantum coherence is handled explicitly this also leads to connections to many-body entanglement and our results provide a clean explanatory framework for recent experimental proposals in trapped ion systems \cite{Zoe}.

\subsection{Summary of the core results}

This paper's aim is to extend classical notions of work to quantum systems and provide novel tools for the analysis of quantum thermodynamics and the role that coherence plays. To a large extent, quantum thermodynamics has taken inspiration from statistical mechanics when studying notions of `work' \cite{NoGo,Allahverdyan,AspectsWork,WorkObservable}. In the present analysis, we instead revisit the \emph{deterministic} Newtonian concept of work and propose a natural coherent extension into quantum mechanics. We then show how this coherent notion of work automatically arises in the \emph{non-deterministic} context of recent fully quantum fluctuation theorems \cite{Aberg} and provides a fully quantum-mechanical account of thermodynamic processes.

In Section \ref{sec: Coherent work} we begin by defining a coherent work process
\begin{equation}
|\psi_0\>_S \stackrel{\omega}{\longrightarrow} |\psi_1\>_S,
\end{equation}
which is a deterministic primitive to describe coherent energy exchanges. The evolution is constrained to conserve energy microscopically and preserve the statistical independence of the systems. In particular, the exclusion of entanglement generating processes in this definition is a choice required in order to be the quantum-mechanical equivalent of Newtonian work involving a deterministic transition of a `weight' system. Coherent work processes are found to separate into coherently trivial and coherently non-trivial types, with the distinction appearing as a consequence of the notion of `decomposability of random variables' in probability theory \cite{CharacteristicFunctions,ArithmeticProbability}. This connection enables us to show that coherent states form a closed subset under coherent work processes. As a result, we find the classical limit of this fragment of quantum theory is equivalent to Newtonian work processes under a conservative force. After this, we explicitly provide a complete characterization of coherent work processes for quantum systems with equal level spacings, and also prove that measurement statistics on such systems always become less noisy under coherent work processes.

In Section \ref{sec: fluctuation theorem} we move to a thermodynamic regime in which we now have a thermal mixed state described by an inverse temperature $\beta = (kT)^{-1}$. We derive a quantum fluctuation theorem using an \emph{inclusive} picture, and frame the result on the athermal system.
Under the assumptions of  time-reversal invariant dynamics and  microscopic energy conservation, we derive a fluctuation theorem in terms of cumulant generating functions that has the classical Crooks result as its high-temperature limit. However we find that the quantum fluctuation theorem admits an infinite series of corrections to the Crooks form in powers of $\beta$. The higher order terms account for the coherent properties of the initial pure states of the system and gives an operational perspective on related works \cite{Aberg,HyukJoon}. There exists another natural decomposition of the fluctuation theorem in terms of `average change in energy' and `average change in coherence' with the caveat that the average coherence must be carefully interpreted. We derive a closed expression for the `average coherence' using readily interpretable quantities. 

A key result of the work is in Section \ref{sec: coherent fluctuation} where we show that the quantum fluctuation theorem becomes a \emph{function of the coherent work state} in a precise sense, and so further justifies the choice of coherent work processes as a coherent primitive.

\subsection{Relation to no-go results}

In the literature, one finds a range of different notions of work. For example, given some closed system in a state $\rho$, described by a time-dependent Hamiltonian $H(t)$ and evolving under a unitary $U$, the average work is frequently expressed as $W = \tr(H(t_f) U \rho U^\dagger) - \tr(H(t_i) \rho)$. Two-point measurement schemes (TPM) are common ways of defining the work done in a process. In such a scheme, one might define
\begin{equation}\label{TPM scheme}
W = \sum_w  \sum_{\substack{i,j':\\ w= E_{j'} - E_i  }}  w \, p_{j'| i} \, p_i
\end{equation}
where $p_i$ is the probability of measuring $|i \>\< i|$ when the system is in the state $\rho$, and $p_{j'|i}$ is the transition probability of going from $|i \> \to |j' \>$ given $|j' \>$ is an eigenstate of $H(t_f)$. Within a TPM definition of work, one would wish to recover the traditional results from classical statistical mechanics when $\rho$ is diagonal (fully incoherent) state. However, a no-go result was obtained in \cite{NoGo} that shows that these two requirements are not compatible for general coherent states. In particular, given any tuple $(H(t_i),H(t_f),U)$, it is impossible to always assign POVM operators for the TPM scheme
\begin{align}
 M^{(w)}_{TPM} &= \sum_{\substack{i,j': \\ w = E_{j'} - E_i}} p_{j'|i} p_i |i \> \< i|
\end{align}
that provide an average work from the expectation value of $X = \sum_w w M^{(w)}_{TPM}$ that is equal for the expectation value of the operator $X = H(t_i) - U^\dagger H(t_f) U$.

Our results include the TPM scheme as a special case, but circumvent no-go results by distinguishing between a random variable obtained by a POVM measurement on a quantum state with coherence, and the quantum state itself as a complete description of a physical system\footnote{In the sense that quantum states in general do not admit satisfactory hidden variable theories.}. For us, `coherent work' is a fully quantum mechanical property and not a classical statistical feature. This stance is further supported by the fact that our coherent work processes play precisely the same role in the coherent fluctuation theorem, as do Newtonian work processes within the classical Crooks relation.

\section{Coherent superpositions of classical processes and decomposable quantum coherence.}\label{sec: Coherent work}

We now go into more detail on what is required in order to have superpositions of processes that overall conserve energy and give rise to deterministic pure state transformations on a subsystem. The criteria are the same as in the case of mechanical work discussed in the introduction. We assume a process in which the following hold:
\begin{enumerate}
\item (Sharp initial state) The quantum system $S$ begins in a definite initial pure state $|\psi_0\>$.
\item (Deterministic dynamics) The system undergoes a deterministic quantum evolution, given by some unitary $U(t)$.
\item (Conservation) Energy conservation holds microscopically and is accounted for with any auxiliary `weight' system $A$ initialised in a default energy eigenstate $|0\>$.
\item (Sharp final state) The system $S$ finishes in some final pure state $|\psi_1\>$.
\end{enumerate} 
Since the dynamics is unitary, and $S$ finishes in a pure state, this implies the system $A$ must also finish in some pure state $|\omega\>$. In the classical case this change in state of the `weight' exactly encodes the work done on the system $S$, and so by direct extension we refer to $|\omega\>$ as the \emph{coherent work output} of the process on $S$. For simplicity in what follows, we shall use the notation
\begin{equation}
|\psi_0\>_S \stackrel{\omega}{\longrightarrow} |\psi_1\>_S,
\end{equation}
to denote the coherent work process on a system $S$ in which $|\psi_0\>_S\otimes |0\>_{A} \rightarrow U(t) [|\psi_0\>_S\otimes |0\>_{A}]= |\psi_1\>_S \otimes |\omega\>_{A}$ under a energy conserving interaction. If separate coherent work processes $|\psi_0\>_S \stackrel{\omega}{\longrightarrow} |\psi_1\>_S$ and $|\psi_1\>_S \stackrel{\omega'}{\longrightarrow} |\psi_0\>_S$ are possible, the transformation is termed \emph{reversible} and denoted by $|\psi_0 \> \stackrel{\omega}\longleftrightarrow |\psi_1 \>$. We are able to label the two processes by $\omega$ instead of $(\omega,\omega')$ because the states $|\omega\>$ and $|\omega'\>$ are related in a very simple way, as we prove in Appendix \ref{thm:uniqueness}.

The coherent work output is in general a quantum state with coherences between energy eigenspaces, however its form is essentially unique for a given initial state $|\psi_0\>$ and given final state $|\psi_1\>$ for the system $S$. This is discussed more in Appendix \ref{thm:uniqueness}. Finally, we note the assumption that the output system is the same as the input system $S$ can be easily dropped, where we require that $SA \rightarrow S'A'$ under an isometry, with energy conservation defined over the composite input/output systems involved.

We can now provide some concrete examples of coherent work processes.

\begin{center}
\begin{figure}[t]
\includegraphics[width=8cm]{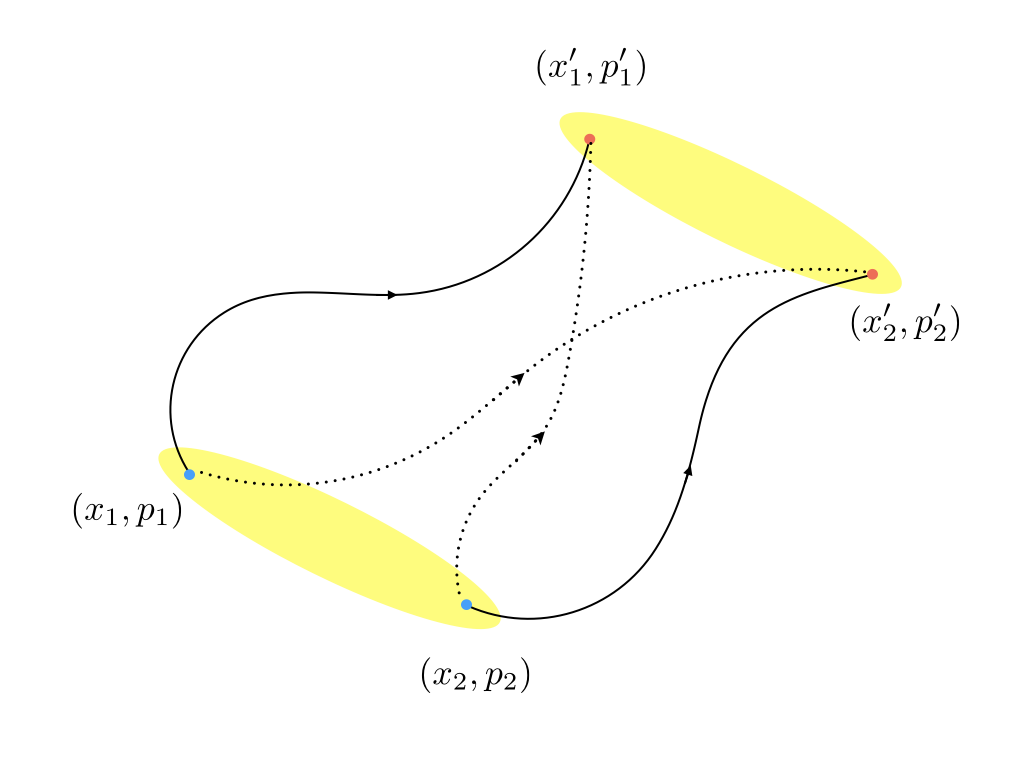}
\caption{\textbf{Coherent superpositions of mechanical processes.} What unitary transformations are possible on a joint system that conserves total energy, but gives rise to a deterministic transformation of a pure quantum state $|\psi_0\>$ into some final pure quantum state $|\psi_1\>$? Conservation laws place non-trivial obstacles on quantum systems, and not all such processes can be superposed. This corresponds to the existence of quantum states that have non-decomposable coherence. States in the set $\mathcal{C}$ of semi-classical states, defined in the main text, are infinitely divisible and in the $\hbar \rightarrow 0$ limit recover the classical regime. 
}
\end{figure}
\end{center}

\begin{example}
Consider some finite quantum system $S$ with Hamiltonian $H_S$. We assume for simplicity that the energy spectrum of $H_S$ is non-degenerate and given by $\{0,1,2, \dots \}$. We write $|k\>$ for the energy eigenstate of $H_S$ with energy $k$. 

Consider the following two incoherent transitions of a system $S$ between energy states.  The first is $|0\> \rightarrow |5\>$ and the second is $|1\> \rightarrow |6\>$. In both cases the process can be done by supplying $5$ units of energy to the system from an external source $A$, and thus $|\omega\> = |-5\>$. Moreover there is a unitary $V$ that can do both processes coherently in superposition:
\begin{equation}
V[( a |0\> +b |1\>) \otimes |0\>] = (a |5\> +b |6\>) \otimes |-5\>,\nonumber
\end{equation}
where $a,b$ are arbitrary complex amplitudes for the quantum state. Note also, that for the case of $|a|=|b|=1/\sqrt{2}$ this coherent transition could equally arise as a superposition of $|0\> \rightarrow |6\>$ and $|1\> \rightarrow |5\>$. For this realisation, a different energy cost occurs for each individual transition, however the net effect also gives rise to $|-5\>$ on the external source system. This process is reversible.
\end{example}

Beyond this simple example subtleties emerge and one can obtain non-trivial interference effects in the work processes, as the following example demonstrates.

\begin{example}
Consider again, a quantum system $S$ as in the previous example. One can coherently and deterministically do a work process where
\begin{equation}
\frac{1}{2} (|0\> + |1\> + |2\> + |3\>) \stackrel{\omega}{\longrightarrow} \frac{1}{\sqrt{2}} (|5\> + |6\>),
\end{equation}
which can be viewed as a merging of classical work trajectories. The coherent work output for this process is $ |\omega\>= \frac{1}{\sqrt{2}} (|-3 \> + | -5 \> )$. This process is irreversible.
\end{example}
The proof that this transition is possible deterministically and irreversibly follows directly from Theorem \ref{thm: Convolution} below.

However quantum mechanics also has prohibitions that give rise to highly non-trivial constraints that rule out certain processes, as the following illustrates.
\begin{example}\label{impossible}
Let $S$ be the same quantum system as in the above examples. It is \emph{impossible} to superpose the two transitions $|0\> \rightarrow |5\>$ and $|1\> \rightarrow |7\>$ in a coherent work process, even if we allow $A$ to finish in some superposed state $|\omega\>$ with amplitudes over different energies. 
\end{example}
This result again follows from Theorem \ref{thm: Convolution} below, but one way to see that this is impossible is that the swapping of the initial states $|0\>$ and $|1\>$ and the swapping of the two final states $|5\>$ and $|7\>$ do not correspond to the same change in energy, and so the ``which way'' information for the process cannot be erased in a way that does not leave an external energy signature (as occurs in the first example for $|a| = |b| = 1/\sqrt{2}$).

\subsection{``Which-way'' information for processes and energy conservation constraints}

We can make more precise what we mean by the energy conservation condition providing a restriction what classical work processes can be superposed when the quantum state has some permutation symmetry. For simplicity we can consider the same quantum system $S$ as in the above examples. Suppose it is in some initial state $|\psi\>_S$ that is invariant under swapping of energy eigenstates $|a\>$ and $|b\>$. We can let $F_{a,b} :=( |a\>\<b| + |b\>\<a| + \sum_{k \ne a,b} |k\>\<k|)$ be the operator that swaps these levels, and thus $F_{a,b}\otimes \I_A |\psi\>\otimes |0\> = |\psi\> \otimes |0\>$.  This must map to a corresponding operation on the output state $|\phi\>\otimes |\omega\>$, which we denote $\tilde{F}$, and which must obey $ V F_{a,b}\otimes \I_A = \tilde{F} V$, where $V$ is the interaction unitary that conserves energy and performs the transformation. Therefore we have that $\tilde{F} = V (F_{a,b}\otimes \I_A) V^\dagger$. However because of energy conservation we must also have that
\begin{align}
V |a\> \otimes |0\> &= |x \>\otimes |a-x\> \\
V |b\> \otimes |0\> &= |y\> \otimes |b-y\>,
\end{align}
for some integers $x,y$, and thus,
\begin{equation}
\begin{aligned}
\tilde{F} = |x,a-x\>\<y,b-y| &+ |y,b-y\>\<x,a-x|  \\  &+ \mbox{ other terms}.\nonumber
\end{aligned}
\end{equation}
However in order for such a which-way symmetry to relate solely to the system $S$, we must have that $\tilde{F}$ factors into a product operator of the form $X\otimes \I_A$ so that the invariance on the initial superposition is associated with a corresponding transformation on the output state. Restricting just to the space spanned by $|x,a-x\>$ and $|y,b-y\>$ we see that this implies $a-x = b-y$, and so $y-x= b-a$ and the invariance is under the swapping operation $X = F_{x,y}$. Thus, the invariance under permutations in a superposition together with energy conservation implies that a transformation such as the one in example \ref{impossible} is forbidden.

\subsection{Superposition of classical processes in a semi-classical regime}\label{C-states}
Having provided some initial examples for coherent work processes, we next link with the concept of work in classical mechanics through the following result. For this we consider a mechanical coordinate $x$ together with its conjugate momentum $p$, which are quantised in the usual way with commutation relation $[x,p]=i\hbar \I$. We define a coherent state for the system via $|\alpha\> = D(\alpha) |0\>$ for any $\alpha \in \mathbb{C}$ and where $a|0\> = 0$, with $a := (x + ip)/\sqrt{2}$ and with the displacement operator $D(\alpha) := \exp [\alpha a^\dagger - \alpha^* a]$. This defines the set of coherent states $\{ |\alpha\> = D(\alpha)|0\> : \alpha \in \mathbb{C}\}$ for the system. 

We now note that we always have the freedom to rigidly translate a coherent state in energy $|\psi\> \rightarrow \Delta^k |\psi\>$, where $\Delta := \sum_{n \ge 0} |n+1\>\<n| $ for some non-negative integer $k$ ($\Delta$ is occasionally referred to as the phase operator \cite{PhaseOperator}). We will also allow arbitrary phase shifts in energy $|n\> \rightarrow e^{i\theta_n}|n\>$. We now define the states
\begin{equation}
|\alpha , k \> := \Delta^k |\alpha \>,
\end{equation} 
where $|\alpha \>$ is a coherent state \cite{DeformedAnnihilation,PACSOriginal,PhotonSubtraction,PhotonSubtraction1,PhotonSubtraction2}.
We denote by $\mathcal{C}$ the set of quantum states obtained from any coherent state via rigid translations by a finite amount together with arbitrary phase shifts in energy. Note that arbitrary energy eigenstates are included in the set $\mathcal{C}$ as a rigid-translation of the vacuum $|\alpha = 0 \>$ state. For oscillator systems the energy eigenstates are not usually considered as classical states, however their inclusion in the set $\mathcal{C}$ is required for consistency in the $\hbar \rightarrow 0$ limit.

Explicitly, the states take the form
\begin{equation}
\mathcal{C} := \left\{ |\psi \> = e^{-i L t} |\alpha , k \>  :  k \in \mathbb{N}_0,  \alpha \in \mathbb{C}, t\in \mathbb{R}, [L, H ] =0  \right\} \nonumber
\end{equation}
where we allow the use of an arbitrary Hermitian observable $L$ that commutes with $H$ to generate arbitrary relative phase shifts in energy. 

We now consider the harmonic oscillator, but it is expected that a similar statement for the classical $\hbar \rightarrow 0$ limit can be made for more general coherent states.

\begin{theorem}[Semi-classical regime]\label{thm:Semi Classical} Let $S$ be a harmonic oscillator system, with Hamiltonian $H_S = h \nu a^\dagger a$. Let $\mathcal{C}$ be the set of quantum states for $S$ as defined above. Then:
\begin{enumerate}
\item The set of quantum states $\mathcal{C}$ is closed under all possible coherent work processes from $S$ to $S$ with fixed Hamiltonian $H_S$ for both the input and output.
\item Given any quantum state $|\psi\>\in \mathcal{C}$ there is a unique, canonical state $|\alpha, k\> \in \mathcal{C}$ such that we have a reversible transformation
\begin{equation}
|\psi\> \stackrel{\omega_c}{\longleftrightarrow} |\alpha,k\>,
\end{equation}
with $|\omega_c\> = |0\>$ and $\alpha = |\alpha|$.
Moreover, the only coherent work processes possible between canonical states are 
\begin{equation}
|\alpha, k\> \stackrel{\omega}{\longrightarrow} |\alpha',k\>,
\end{equation}
 such that $|\alpha'| \le |\alpha|$. Modulo phases, the coherent work output can be taken to be of the canonical form $|\omega\>_A=  |\lambda, n\>_A$, where $\lambda = \sqrt{|\alpha|^2 - |\alpha'|^2}$ and $n,k'$ are any integers that obey $n+k' = k$.
\item In the classical limit of large displacements $|\alpha| \gg 1$, from coherent processes on $\mathcal{C}$ we recover all classical work processes on the system $S$ under a conservative force.
\end{enumerate}
\end{theorem}\label{classical}
Note that we can weaken the assumption that the output system is $S$ and the Hamiltonian $H_S$ is the same at the start and at the end. It is readily seen that any admissible output system $S'$ must have within its energy spectrum infinitely many discrete energies with separations being multiples of $h \nu$, and the output state distribution is always Poissonian on a subset of these discrete eigenstates. A similar degree of freedom exists for the auxiliary system $A$, which also must have a Poissonian distribution on a subset of discrete energy levels. The output parameters $\alpha, \alpha', \lambda$ must obey the same conditions as in Theorem \ref{thm:Semi Classical}. 

The proof of the above theorem is given in Appendix \ref{Proof Semi Classical}, and provides a simple correspondence between the coherent work processes and work processes in classical Newtonian mechanics. The main aspect of this result that is non-trivial is to show that $\mathcal{C}$ is closed under work processes, but follows from a non-trivial result in probability theory.

Care must be taken when defining the classical limit as we do in Theorem \ref{thm:Semi Classical}. Conservative forces are characterised by their `path-independence', and so the energy change is dependent only upon the initial and final phase space co-ordinates $(\x_0,\p_0)$ and $(\x_1,\p_1)$ respectively. But the distinction between quantum and classical is that in classical physics we can assign a well-defined pair of co-ordinates to our system, labelling the position and momentum. In addition to this, the energy is also as sharp as instruments allow. We therefore use the fact that for $|\alpha|$ large, the standard deviation of the position, momentum and energy of our system grow increasingly negligible compared to their expected values. Furthermore, the energy change of the two phase space points entirely captures the work done in the process.

The coherent state result shows that in a semi-classical limit we recover the familiar work behaviour. However, the earlier examples show that once we deviate from the semi-classical regime the structure of coherence in such processes becomes highly non-trivial and certain transformations are impossible. This is captured by the notion of \emph{decomposable coherence} and \emph{non-decomposable coherence}, which we now describe.

\subsection{Decomposable and non-decomposable quantum coherence}

The topic of coherent work is directly related to whether a random variable in classical probability theory is `decomposable' or not. In classical theory a random variable $\hat{X}$ is said to be \emph{decomposable} if it can be written as $\hat{X}=\hat{Y}+\hat{Z}$ for two independent, non-constant random variables $\hat{Y}$ and $\hat{Z}$ \cite{Lukacs,ArithmeticProbability}. This turns out to provide a different characterization of coherent work processes. For compactness, given a system $X$ with Hamiltonian $H_X$, denote by $\hat{X}$ the random variable obtained from the measurement of $H_X$ in some state $|\psi\>$. The probability distribution is given by $p (E_k) = \<\psi|\Pi_k |\psi\>$, where $\Pi_k$ is the projector onto the energy eigenspace of $H_X$ with energy $E_k$. 

Within our quantum-mechanical setting we now say that a state $|\psi\>$ has \emph{non-decomposable coherence} if given a coherent work process $|\psi\> \stackrel{\omega} \longrightarrow |\phi\>$ this implies that either $|\phi \>$ or $|\omega\>$ are energy eigenstates of their respective Hamiltonians, and otherwise the state is said to have \emph{decomposable coherence}. We also refer to such non-decomposable transformations as ``trivial'' coherent processes. With this notion in place, we obtain the following result.

\begin{theorem}
\label{thm:Decomposable}Given a quantum system $X$ with Hamiltonian $H_X$ with a discrete spectrum, a state $|\psi\>$ admits a non-trivial coherent work process if and only if the associated classical random variable $\hat{X}$ is a decomposable random variable. Furthermore, for a coherent work process 
\begin{equation}
|\psi \>_X \stackrel{\omega}{\longrightarrow} |\phi \>_Y,
\end{equation}
the associated classical random variables are given by $\hat{X} =  \hat{Y} + \hat{W}$, where $\hat{Y}$ and $\hat{W}$ correspond to  the measurements of $H_Y$ and $H_{W}$ in $|\phi \>_Y$ and $|\omega \>_W$ respectively.
\end{theorem}

We see that the states in $\mathcal{C}$ have \emph{infinitely divisible} coherence, and thus admit an infinite sequence of non-trivial coherent work processes.  A random variable $\hat{Y}$ is infinitely divisible if for any positive integer $n$, we can find $n$ independent and identically distributed random variables $\hat{X}_n$ that sum to $\hat{Y}$ \cite{CharacteristicFunctions}. The proof of this is provided in Appendix \ref{Proof Decomposable}.

\begin{center}
\begin{figure}[t]
\includegraphics[width=8cm]{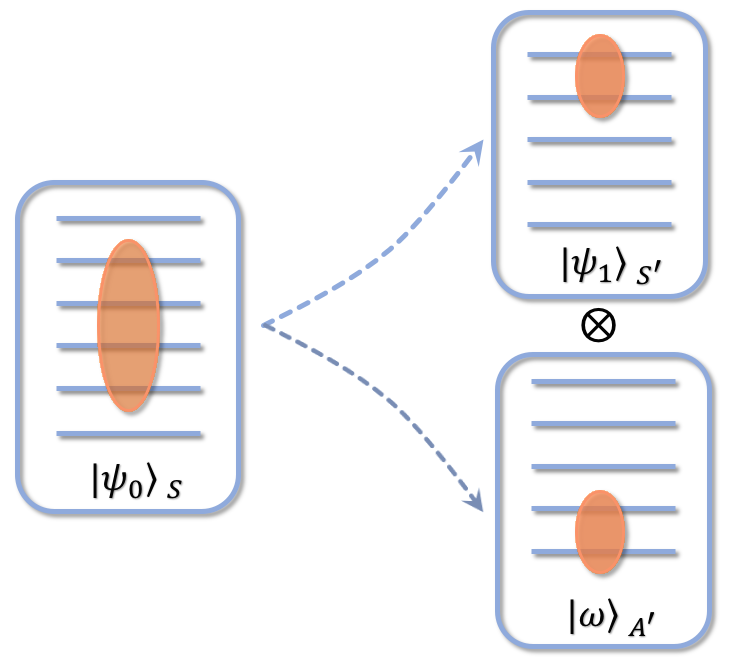}
\caption{\textbf{Decomposability of coherence.} Shown is a schematic of decomposable coherence. A state $|\psi_0\>$ for a quantum system $S$ can be split into pure state coherence $|\psi_1\>$ on $S'$ and a coherent work output $|\omega\>$ on $A$ under an energetically closed process. The probability distributions for the energy measurements on both $|\psi_1\>$ and $|\omega\>$ always majorize the distribution for $|\psi_0\>$ and so the energy measurement disorder on $S$ is always non-increasing.}  
\end{figure}
\end{center}

\subsection{General coherence decomposition and disorder in measurement statistics}
We now specify exactly what coherent work processes are possible for a quantum system $S$ with discrete, equally-spaced energy levels. To avoid boundary effects, we assume for mathematical convenience that both $S$ and $A$ are doubly infinite ladders and make use of core results from asymmetry theory \cite{GourSpekkens}. 

Note that for a finite $d$-dimensional system $S$ with constant energy gaps, essentially the same analysis applies since one can always embed the finite system into $d$ levels of such a ladder, apply the below theorem and ensure that the support of the output state does not stray outside the $d$ levels of the embedding. For systems with general energy spacings, again a similar analysis can be applied, since the required dynamics will not couple different Fourier modes in a quantum state and so will only transform coherences between subsets of levels with equal level spacings.

For the core model of a ladder system with equally spaced levels we have the following statement for coherent work processes, with the proof provided in Appendix \ref{Proof of ladder theorem}.
\begin{theorem}\label{thm: Convolution}
Let $S$ be a quantum system,  with energy eigenbasis $\{ |n\>: n \in \mathbb{Z}\}$ with fixed Hamiltonian $H_S = \sum_{n \in \mathbb{Z}} n |n\>\<n|$. Given an initial quantum state $|\psi\>_S = \sum_i \sqrt{p_i} |i\>_S$, a coherent work process
\begin{equation}
|\psi\>_S \stackrel{\omega}{\longrightarrow} |\phi\>_S ,
\end{equation}
is possible with
\begin{align}
|\phi \>_S &= \sum_j e^{i\theta_j} \sqrt{q_j} |j\>_S\\
|\omega \>_A &= \sum_k e^{i\varphi_k}\sqrt{r_k} |k\>_A
\end{align}
for arbitrary phases $\{\theta_j\}$ and $\{\varphi_k\}$, with output Hamiltonian $H_S$ as above and $H_A = \sum_{n \in \mathbb{Z}} n |n\>_A\<n|$,
if and only the distribution $(p_n)$ over $\mathbb{Z}$ can be written as
\begin{align}\label{trans-major}
p_n &= \sum_{j \in \mathbb{Z}} r_{j} (\Delta^{j}\mathbf{q})_n. \nonumber \\
&=  \sum_{j \in \mathbb{Z}} r_j q_{n-j}
\end{align}
for distributions $(q_m)$ and $(r_j)$ over $\mathbb{Z}$. 
\end{theorem}
Again, note that we may weaken the assumption that the Hamiltonian of $S$ is the same at the start and at the end with essentially the same conclusion. For example, one may introduce unoccupied energy levels in the output system $S$. Moreover if we drop the assumption of the precise form of the Hamiltonian $H_A$, it is readily checked that the output Hamiltonians of $S$ and $A$ could be shifted by equal and opposite constant amounts while respecting the conservation of energy condition. In each of these variants, however, the core structure of the output distributions is essentially the same.

Now the Birkhoff-von Neumann theorem states that every bistochastic matrix is a convex combination of permutations \cite{InfiniteBirkhoff,Convex}, and thus since equation (\ref{trans-major}) is a convex combinations of translations (which are themselves permutations) this means the two distributions are related by a bistochastic mapping $\q \rightarrow \p = A \q$, where $A$ is bistochastic. However this implies that $\q$ majorizes $\p$, written $\p \prec \q$, and which establishes a very useful relation between the input distribution on $S$ and output distributions on $S$ and $A$. Specifically it implies that in any such coherent work process we must have both $\p \prec \q$ and $\p \prec \mathbf{r}$. Now if $f$ is any concave real-valued function on probability distributions, then a standard theorem \cite{Convex} tells us that $f(\p) \ge f(\q)$ whenever $\p \prec \q$. This leads us to the following result.
\begin{corollary}[\textbf{Disorder in coherent processes never increases}] \label{corollary: disorder}
Let $f_{\rm \tiny dis}(\psi)$ be some real-valued concave function of the distribution $(p_k)$ over energy in $|\psi\>= \sum_k e^{i \theta_k }\sqrt{p_k} |k\>$ (such as the Shannon entropy function), then in any coherent work process $|\psi\> \stackrel{\omega}{\longrightarrow} |\phi\>$ we have that
\begin{align}
\max \{f_{\rm \tiny dis}(\phi), f_{\rm \tiny dis}(\omega) \}&\le f_{\rm \tiny dis}(\psi).
\end{align}
\end{corollary}
The significance of this is that every measure of ``disorder'' $f_{\rm \tiny dis}$ for the energy statistics in a pure quantum state $|\phi\>$ will be a convex function of this form. The corollary thus shows that the energy statistics of both the coherent work output and the final state of $S$ are always \emph{less disordered} than that of the initial state of $S$ with respect to all measures. This also provides an intuitive perspective on the decomposability of coherence in a such a process.

Under repeated coherent work processes on the final state of $S$ and the work output state, the disorder will become diluted but for a general scenario will stop when one reaches non-decomposable components. However for the class of states in $\mathcal{C}$ there is no such obstacle and the disorder can be separated indefinitely. Note also that if in addition the function $f_{\rm \tiny dis}$ is additive over quantum systems, namely $f_{\rm \tiny dis}(\phi \otimes \psi) = f_{\rm \tiny dis}(\phi) + f_{\rm \tiny dis}(\psi)$, then the total disorder over all systems as measured by this function will remain constant throughout.

\subsection{Multiple processes and coherently connected quantum states}
The framework we have laid out involves an auxiliary system initialised in an energy eigenstate. It is for this reason that the disorder dilutes among the bipartite system, rather than being able to increase in one subsystem. This has the advantage of isolating the coherent manipulations on the primary system. Due to this, in a sequence of coherent work processes $|\psi_1 \>_S \stackrel{\omega_1} \longrightarrow |\psi_2\>_S \cdots \stackrel{\omega_n}\longrightarrow |\psi_{n+1} \>_S$, there might only be a finite number of non-trivial processes possible before the system ends in a state with non-decomposable coherence. However, a notable counter-example is the semi-classical result for which this never happens.

The method we have chosen is the simplest conceptually, but it is not the most general procedure possible. One could imagine an auxiliary system prepared in an arbitrary state $|\phi \>_{A}$. In such a case, the approach involving random variables must be slightly modified but the same tools can be applied. The most general coherent work process possible, excluding entanglement generation, is a transformation of the form $\bigotimes_{i = 1}^n |\psi_i \> \longrightarrow \bigotimes_{i=1}^m |\psi_i'\>$, with the only two constraints that the dynamics conserve energy globally, and the subsystems are left in a pure product state. In this modified framework, results such as Corollary \ref{corollary: disorder} would no longer hold, since coherence can be concentrated from many subsystems to fewer subsystems. We leave this as an open question for later study.

A special case that is relevant here is if we have two quantum states $|\psi\>_S$ and $|\phi\>_S$ and the auxiliary system $A$ begins in a state $|\omega\>_A$ with coherence and terminates in some energy eigenstate $|0\>_A$. This can be denoted as
\begin{equation}\label{forward-process}
|\phi\>_S \stackrel{\omega}{\longrightarrow} |\psi\>_S,
\end{equation}
however, if we are concerned with transitioning from $|\psi\>_S$ to $|\phi\>_S$ we can use the fact that if $V$ is an energy conserving unitary realising the process (\ref{forward-process}) then its inverse $V^\dagger$ exists and is also energy conserving. This implies that we can write the inverse transformation as
\begin{equation}
 |\psi\>_S \stackrel{-\omega}{\longrightarrow} |\phi\>_S,
\end{equation}
which is realised through the unitary transformation
\begin{equation}
V^\dagger [ | \psi\>_S \otimes |\omega\>_A ] = |\phi\>_S \otimes |0\>_A.
\end{equation}
This extension of notation allows us to define the negative coherent work output $-\omega$, which physically means the coherent quantum state $|\omega\>_A$ \emph{required} to realise the transformation -- in other words $\omega$ is the coherent work \emph{input} to the process. If a coherent work process exists in either direction between $|\psi\>_S$ and $|\phi\>_S$ then we say the states $|\psi\>_S$ and $|\phi\>_S$ are \emph{coherently connected}. From this definition we see that the whole set $\mathcal{C}$ is a coherently connected set in the sense that any two states in $\mathcal{C}$ are coherently connected. Note however that coherent connectedness is not a transitive relation: if $(|\psi_1\>, |\psi_2\>)$ and $(|\psi_2\> ,|\psi_3\>)$ are each coherently connected pairs then it does not imply that $(|\psi_1\>, |\psi_3\>)$ are coherently connected.

\subsection{But shouldn't `work' be a number?}

We have used the term `coherent work output', however the issue of what `work' is in general, is subtle. Here we do not wish to use the term without justification, and so now expand on why this is not an unreasonable use of language. In particular, we argue that the approach is fully consistent with operational physics and is a natural generalisation of what happens in both classical mechanics and statistical mechanics. For ease of analysis we list the main components of the argument:
\begin{itemize}
\item Work is done on a system $S$ in a process when an auxiliary physical system $A$ is required to transform, under energy conservation from a default state, so as to realise the process on $S$.
\item Work on a system $S$ in Newtonian mechanics is a number $w\in \mathbb{R}$ that operationally can be read off from the transformed state of $S$ or equivalently from the transformed auxiliary system $A$'s state.
\item Work on a system $S$ in statistical mechanics \emph{is not a number}, but is described mathematically by a random variable $W: (d\mu(x),x)$ on the real line. Operationally, individual instances are read off from measurement outcomes on $S$, and by energy conservation are in a one-to-one correspondence with outcomes on $A$.
\item Quantum mechanics is a non-commutative extension of statistical mechanics, in which classical distributions $d \mu(x)$ are replaced with density operators $\rho$ on some Hilbert space $\H$.
\item Complementarity in quantum theory implies that general POVM measurements on a state $\rho$ are incompatible, and moreover ``unperformed measurements have no result'' \cite{Peres}. No underlying values are assigned to a pure state $|\omega\>$ in a superposition of different classical states.
\end{itemize}
The last two points are vital. One can of course consider random variables obtained from measurements on quantum systems, but this is not the same thing and cannot describe quantum physics in any sensible way. Trying to use classical random variables to fully describe quantum physics is essentially a ``hidden variable'' approach, however we know from many results that it is impossible to define a hidden variable theory for quantum mechanics unless one is willing to make highly problematic assumptions, such as violation of local causality \cite{Peres}. 

\begin{center}
\begin{figure}[t]
\includegraphics[width=8cm]{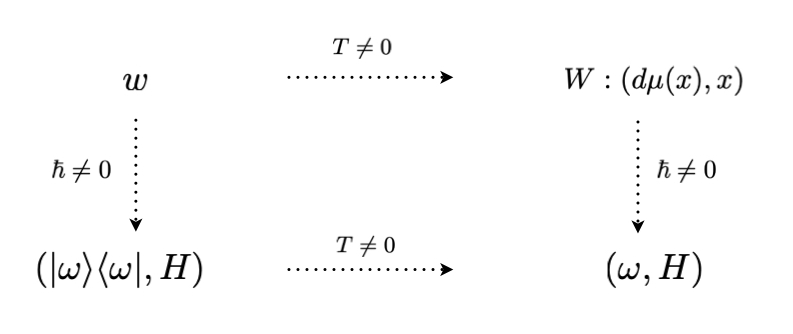}
\caption{\textbf{When does ``work'' make sense?} The two left-hand cases concern deterministic processes, whereas the two right-hand cases involve mixed states with unavoidable probabilistic elements. In Newtonian mechanics, work is a deterministic number $w$ and a largely unambiguous concept. Extensions into statistical mechanics mean that work is no longer described by a number, but instead is a classical random variable $W: (d \mu(x), x)$. Classical mechanical work occurs when this distribution is sharp, $\mu(x) = \delta (x-w)$. However one may extend the deterministic concept of work from Newtonian mechanics in a natural way into quantum mechanics. This leads to a coherent form of work, that coincides with a superposition of classical work processes and in the $\hbar \rightarrow 0$ limit recovers the classical notion. It also it arises in fluctuation relations in the same manner as its classical counterpart. }
\end{figure}
\end{center}
\vspace{-1cm}

In our approach, the coherent work process does \emph{not} associate a single, unique number to the coherent work output, instead we describe the output of the process in terms of a pair $(|\omega\>\<\omega|, H)$ that generalizes the classical stochastic pair $W:(d\mu(x) ,x)$ into a fully coherent setting. 

Obviously one could measure the state $|\omega\>$ in the energy eigenbasis to get measurement statistics $(p(E_1), p(E_1), p(E_2)\dots, )$, and associate energy scales to these observed outcomes, however this is a fundamentally different scenario -- the state $|\omega\>$ has no randomness in itself, instead it has complementarity in observables which is physically different. The measurement on $|\omega\>$, in contrast, introduces classical randomness and so steps out of the deterministic mechanical domain. Indeed, one can make the following comparison: \textit{the coherent work output $|\omega\>$ is to the measurement statistics $\p=(p(E_1), p(E_1), p(E_2)\dots, )$ what a laser is to incoherent thermal light}. We can obtain interference effects, entanglement and other non-classical phenomena from a quantum state $|\omega\>$, but not from a probability distribution $\p$. 

Note also, that essentially the same approach is taken for quantum entanglement. There is no measurement one can do corresponding to the ``amount of entanglement'' in a state, despite it being a crucial physical property of a state with dramatic effects. Instead, entanglement is quantified in terms of pure state `units', for example the Bell state $|\psi^-\>$. In both the coherent work case and the entanglement case, one ultimately concerned with empirical statistics and not abstract quantum states. The reconciliation of this is that quantification in terms of the abstract quantum states will determine the empirical statistics in a way that depends on the particular context involved.

In the context of the next section, by adopting this coherent description, we can evade recent no-go results \cite{NoGo} on fluctuation relations for thermodynamic systems and provide insight into recent developments in this field \cite{Aberg,Allahverdyan,Anders,Zoe,Autonomous} while also connecting naturally with well-established tools for the study of quantum coherence \cite{ReferenceFrame,Harmonic,MarvianThesis,QuantCoherence}. 

%

\section{Coherent Fluctuation Theorems}\label{sec: fluctuation theorem}

In this section we show that the concept of coherent work processes discussed in the previous section naturally occurs within a recent fluctuation theorem \cite{Aberg,Zoe,Alhambra} that explicitly handles all quantum coherences and which is experimentally accessible in trapped ion systems \cite{Zoe}. We first introduce the core physical assumptions involved and then explain how these naturally relate to the concept of coherent work processes. But first, we briefly re-cap on the classical Crooks fluctuation relation, where phase space trajectories are center stage.

\subsection{Classical, stochastic Tasaki-Crooks relation}

In the most basic setting one has a bath subsystem  $B$, initially in some thermal equilibrium state $\gamma_0 = \frac{e^{-\beta H_0}}{Z}$ with respect to an initial Hamiltonian $H_0$, which is then subject to some varying potential that changes its Hamiltonian $H_0 \rightarrow H(t) \rightarrow H_1$ and finishes in some final Hamiltonian $H_1$. Note that through an appropriate choice of units we can always arrange that a protocol starts at $t=0$ and finishes at time $t=1$. 

By doing sharp, projective measurements of the energy at the start and end of the protocol one obtains stochastic changes in energy on  $B$ for the protocol that lead to a stochastic definition of work (see the review \cite{ColloquiumFluctuation} and the recent paper \cite{AspectsWork} by Talkner and H\"{a}nggi for excellent discussions on the concept of work in statistical mechanics). In this setting one can readily derive a Tasaki-Crooks fluctuation relation that compares the work done during a forward protocol $\mathcal{P}: H_0 \rightarrow H(t) \rightarrow H_1$ with the work done in the time-reversed protocol $\mathcal{P}^*: H_1 \rightarrow H(t)_{\rm rev} \rightarrow H_0$, obtained by reversing the time-dependent variation of the system's Hamiltonian under $\mathcal{P}$ and beginning in equilibrium with respect to $H_1$. The Tasaki-Crooks relation compares the probability of doing work $W=w$ in the forward protocol with the probability of doing work $W=-w$ in the reverse protocol, and is given explicitly by
\begin{equation}
\frac{P[w | \gamma_0, \mathcal{P}]}{P[-w|\gamma_1, \mathcal{P}^*]} = e^{- \beta (\Delta F - w)}
\end{equation}
where $\Delta F := F_1 - F_0$ , and $F_k$ is the free energy of the equilibrium state $\gamma_k = \frac{e^{-\beta H_k}}{Z}$ for $k=0,1$. Note that there is no assumption that $B$ finishes in the equilibrium state $\gamma_1$, and there is no assumption that the protocol is in any way quasi-static. From this one can derive the Jarynski relation and the Clausius relation $\<W\> \ge \Delta F$ from traditional thermodynamics.

\begin{center}
\begin{figure}
\includegraphics[width=8cm]{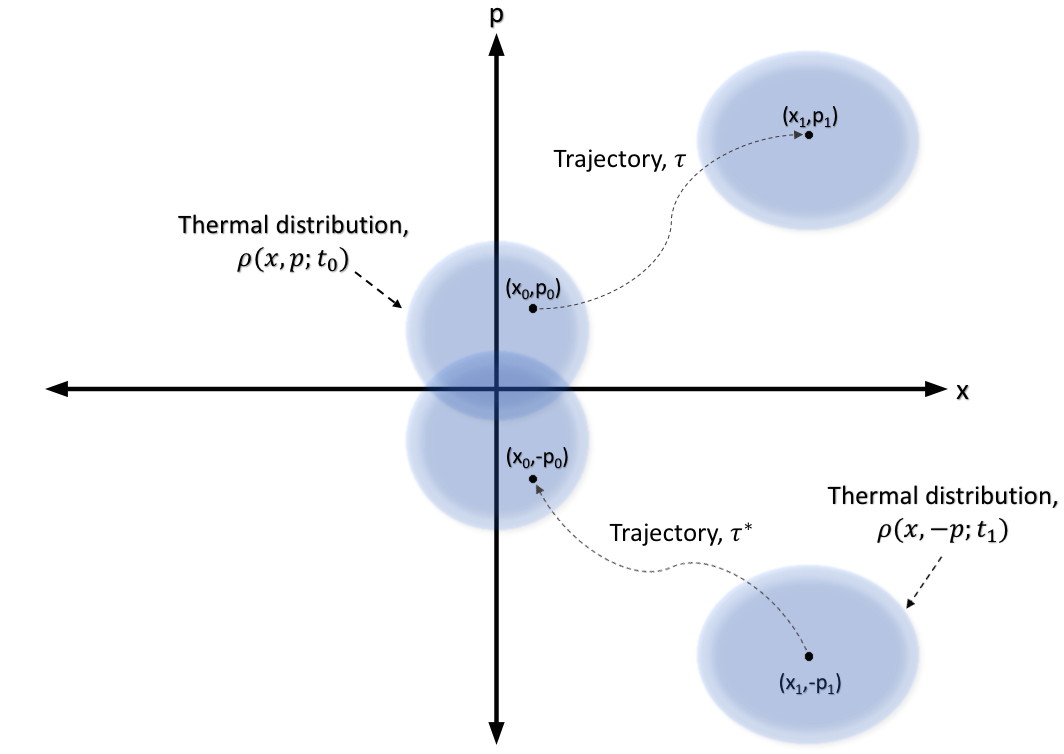}
\caption{\textbf{Classical fluctuation setting.} In  the  classical set-up,  any  particular  phase  space  point  ($\x,\p$)  is  taken  along  a unique  trajectory  ($\x(t),\p(t)$),  and  work  is  accumulated  in  the process.  Work  originates  as  a  change  in  Hamiltonian  from  a time $t_0$ to  a  time $t_1$.  This  change  is  induced  via  an  interaction with  an  implicit  energy  source,  in  which  the  total  system  can be  described  by  a  time-independent  Hamiltonian. Note that if one follows a classical Crooks derivation but imposes a restriction to coarse-grained resolutions of phase space (such as the shaded regions shown) then it is expected that one will obtained similar results to the coherent fluctuation relations for quantum theory.}
\end{figure}
\end{center}

\subsection{From stochastic relations to genuinely quantum fluctuation relations}
While the formalism of the two-point measurement scheme is applicable to both classical and quantum systems alike, it must be emphasized that any genuine coherent features are destroyed and the operational physics is entirely equivalent to classical stochastic dynamics on a classical system (see \cite{JenningsLeifer} for a discussion of why this stochasticity is essentially classical). Several works \cite{Allahverdyan,NoGo,HarryQuasiProb,Anders,NicoleJarzynski,Alhambra,AlbashFluctuation}  have proposed variants of the two-point scheme with the aim of providing fluctuation relations that capture genuinely quantum-mechanical features. Of note are the no-go results of \cite{Contextuality,NoGo} that prove that certain desirable features for a measurement scheme are incompatible. One model that attempts to circumvent these obstacles was proposed in \cite{Allahverdyan} where projective measurements were replaced by weak measurements, and in which quasi-probability distributions arise. A related analysis in terms of quasi-probabilities was provided in \cite{NicoleJarzynski}, while in \cite{Contextuality} it was shown that these ``negative probabilities'' are in fact a witness of contextuality -- a fundamental feature of quantum mechanics that includes quantum entanglement as a special case \cite{KochenSpecker,ContextualityThesis} and is conjectured to be the key ingredient providing speed-ups in quantum computers \cite{ContextualityMBQC,FrembsContextuality}.

However weak measurements are a very particular type of quantum measurement and so one can ask if an over-arching framework exists that includes the classical stochastic setting, weak measurements and fully general POVM measurements that lead to coherent fluctuation relations. In \citep{Aberg} a general framework was developed by \r{A}berg. The framework developed can be viewed as a fully coherent version of the ``inclusive'' description taken by Deffner and Jarzynski in \cite{DeffnerJarzynski} for the thermodynamics of classical systems, and leads to an extremely general fluctuation relation that handles coherence, and includes the classical fluctuation relations as special cases. However the technical level of the analysis is formidable and quite different from existing approaches. The physical assumptions of the model are natural, however it was not clear how one should interpret the broad, abstract results in more familiar terms.

\subsection{The coherent fluctuation framework}
Before proceeding, we make a brief comment on our notation. We shall use capital Roman letters ($A,B,C$) at the start of the alphabet to denote quantum systems that are initialized in quantum states with no coherence between energy eigenspaces. For example, for an auxiliary `weight' systems initialized in $|0\>$ we use $A$. In this analysis we are primarily concerned with coherent features and so we use the notation $S$ for a quantum system that initially has coherence between energy eigenspaces.

For fluctuation relation contexts it is typically the case that we have a large thermal bath that we do not control, together with another quantum system that we do control and which is initially in thermal equilibrium with the bath. Following the above convention we label these $B_1$ (e.g.\!\! for the uncontrolled system) and $B_2$ (e.g.\!\! for the controlled thermal system), and reserve $S$ for additional coherent degrees of freedom that we can also control and wish to study.

The approach taken in \cite{Aberg} can be viewed as a generalisation of \cite{DeffnerJarzynski} to a fully quantum setting in which all energies and all coherences in energy are accounted for explicitly. Moreover, it is exactly in the same spirit as the coherent work processes as introduced earlier -- quantum coherences in energy are always relational degrees of freedom and so an inclusive approach is natural. Specifically, in addition to the primary system $S$ one introduces an additional bath subsystem $B$ and the composite system $SB$ is \emph{energetically closed} in terms of energy flows and coherence flows, but not necessarily \emph{dynamically closed}\footnote{This does not mean that the system is ``autonomous'' -- an experimenter can still implement general time-dependent protocols by macroscopically varying internal parameters of the system.}. Therefore we make use of an inclusive ``microscopic'' description of the energetic degrees of freedom with the aim of arriving at a fluctuation relation that handles arbitrary coherence.

A key thing to highlight is that every fluctuation relation (classical or quantum) involves such an additional system $S$, however this system is generally left implicit as the external physical degrees of freedom that provide energy and coherence so as to change the Hamiltonian of $B$ in some time-dependent manner. Such time-dependent variations of a Hamiltonian interact non-trivially with coherent structures and so this is a primary reason why such an inclusive microscopic description is needed.

A crucial point is the following: since the system $SB$ is closed in terms of energy and coherence flows this means that all energy/coherence changes in $S$ correspond in a one-to-one fashion to the energy/coherence changes in $B$. 

The assumptions of the framework are as follows:
\begin{enumerate}
\item A microscopic, inclusive description is taken for a thermal bath and quantum system.
\item The microscopic description is energetically closed, but not dynamically closed.
\item Time-reversal symmetry holds for the microscopic dynamics of the composite system.
\item The thermal system $B$ is initially in some Gibbs state with respect to an initial Hamiltonian and $S$ is in some arbitrary quantum state.
\end{enumerate}

  The model therefore involves the initialisation of $SB$ in some joint product state $\rho \otimes \gamma_0$ where $B$ is in thermal equilibrium, as above, while $S$ is allowed to be in an arbitrary state $\rho$. The primary system is then subject to open system dynamics that transforms the state as $\rho \rightarrow \E(\rho)$ where
\begin{equation}
\E(\rho)= \tr_B V(\rho \otimes \gamma_0)V^\dagger,
\end{equation}
in terms of a microscopic unitary $V$ on $SB$ that may be partially controlled by macroscopic parameters by an experimenter. Note that any protocol that an experimenter implements will in general be through macroscopic parameters, however any such transformation will admit such a description via a Stinespring dilation \cite{StinespringDilation}. We refer the reader to \citep{Aberg} for a discussion of this point. 

In the incoherent regime, the fluctuation theorem compares transition probabilities of $\rho \to \sigma$ with transition probabilities of $\Theta \sigma \Theta^\dagger \to \Theta \rho \Theta^\dagger$ where $\Theta$ is a time-reversal operator. We shall find that when coherence is present, this must be generalised so that the core Crooks construction can be implemented.

\subsection{Time-dependent Hamiltonians within a microscopic, inclusive constraint--description}
Since we must first account for microscopic degrees of freedom, any protocol for the change in Hamiltonian $H_0 \rightarrow H(t) \rightarrow H_1$ should be handled with care. The physical reason for this is that a general time-dependent Hamiltonian will interact non-trivially with the quantum coherence between energies and so to properly describe the latter one must be careful with the former. At the microscopic level Hamiltonians are generators of time-translation, whereas any time-dependent Hamiltonian is always an effective description that arises through the interaction with some external system. 

Consider any time-dependent Hamiltonian $H(t)$ scenario that evolves from some initial $H_0 := H(0)$ to some final $H_1 := H(1)$. Since the `$t$' in $H(t)$ corresponds to a physically discernible value this implies that it can always be made explicit within a Hilbert space description via a degree of freedom of the composite system (otherwise the parameter is physically meaningless!).

Therefore, there always exists a Hilbert space description of a time-dependent Hamiltonian in terms of a composite $SB$ system with a Hilbert space that takes the form
\begin{equation}\label{split}
\H_{SB} =  \H_S \otimes (\H^0_B \oplus \H_B^1 \oplus \H_B^{\rm other} ),
\end{equation}
where $\H^0_B$ is the span of the eigenstates $\{|E_k^0\>\}$ of the initial Hamiltonian $H_0$, and similarly for $\H^1_B$ and $H_1$. The subspace $\H_B^{\rm other}$ corresponds to any other physical degrees of freedom that may be accessed at intermediate times $0 < t <1$ of the protocol\footnote{The description of a continuous infinity of `t' values in this vein has technical subtleties, namely the mathematical awkwardness of non-separable Hilbert spaces. Experimentally however, physics always involves a finite resolution scale and so this subtlety has no physical content here.}. Therefore any protocol $H_0 \rightarrow H(t) \rightarrow H_1$ on a quantum system can be understood at a microscopic level as the evolution of the bath subsystem $B$ from being constrained solely to the subspace $\H_B^0$, at the start of the protocol (t=0), to it being constrained to the subspace $\H_B^1$ at the end of the protocol (t=1). The underlying Hamiltonian for the composite system is therefore $H_{SB} = H_S \otimes \I_B + \I_S \otimes H_B$, with
\begin{equation}
H_B = H_0 \oplus H_1 \oplus H_{\rm other},
\end{equation}
where we simply combine all possible intermediate Hamiltonians $H(t)$ into the term $H_{\rm other}$ simply for compactness. This underlying description in terms of constraints, while appearing strange from the traditional $H(t)$ formulation, is entirely consistent with the usual story. More importantly, it turns out to provide a powerful perspective in the context of fluctuation relations with coherence.

\subsection{Coherent thermalisation of a system with respect to constraints}
In our analysis, the constraint--description for the time-dependent protocol turns out to connect with the traditional notion of thermodynamic constraints in phenomenological settings (see for example Callen \cite{Callen}). Specifically, we can talk of a ``coherent thermalisation'' of $B$ with respect to a constraint $\mathcal{C}$ at inverse temperature $\beta$, as we now describe through examples. 

Let $\mathcal{C}_0$ be the constraint ``$B$ is constrained to $\H_B^0$''. Mathematically this constraint is described by the projector $\Pi_0 = \sum_k |E_k^0\>\<E_k^0|$ onto the subspace $\H_B^0$. Experimentally this corresponds to a POVM measurement on $SB$ given by $\{\I_S \otimes \Pi_0, \I_S \otimes \Pi_1, \I_S \otimes \Pi_{\rm other} \}$ that asks: does  $B$ have a Hamiltonian $H_0$, the Hamiltonian $H_1$, or some intermediate Hamiltonian $H(t)$? In actual experiments it simply amounts to the experimenter looking at what the classical dials of the apparatus are set to, or if the time-varying protocol is fixed, looking at a clock to determine the time $t$. In the event of the first outcome, the system  $B$ is updated via the projector $\Pi_0$ that ensures it is entirely constrained to $\H_B^0$ and so $\mathcal{C}_0$ simply amounts to the statement: \textit{the system $B$ has Hamiltonian $H_0$}. 

Thermalisation with respect to the constraint $\mathcal{C}_0$ represented by a projector $\Pi_0$ is now defined \cite{Aberg} by the transformation
\begin{align}
\Pi_0 \rightarrow \Gamma(\Pi_0) &= \frac{e^{-\beta H_B/2} \Pi_0 e^{-\beta H_B/2}}{\tr (e^{-\beta H_B} \Pi_0)}\\
&= \frac{e^{-\beta H_0}}{Z} = \gamma_0,
\end{align}
and thus $\Gamma(\Pi_0)$ represents the statement that  $B$ has Hamiltonian $H_0$ and is thermalised with respect to it, at inverse temperature $\beta$. The mapping $\Gamma$ should be viewed as transforming a Hamiltonian constraint into a thermodynamic constraint. In exactly the same way, we also have that $\Gamma(\Pi_1) = \gamma_1$.

The particular form of the coherent thermalisation transformation, which we shall simply call Gibbs rescaling, is required for various reasons, however certain important cases should first be highlighted. If the constraint $\mathcal{C}$ is ``The system $B$ is in the energy eigenstate $|E_k\>$ of $H_0$'' then this is a much stronger constraint. Again, this is represented via a projector $|E^0_k\>\<E^0_k|$, however now we find that 
\begin{equation}
\Gamma(|E^0_k\>\<E^0_k|) = \frac{e^{-\beta H_B/2} |E^0_k\>\<E^0_k|e^{-\beta H_B/2}}{\tr (e^{-\beta H_B}|E^0_k\>\<E^0_k|)} = |E_k^0\>\<E_k^0|.
\end{equation}
Thus if $B$ is constrained to be exactly in the sharp eigenstate $|E_k^0\>$ then there are no remaining degrees of freedom to thermalise, and the state remains the same.

A more interesting case is if the only thing we know is that $B$ is in pure state, with a uniform superposition over all the eigenstates of $H_0$. This condition is described by the projector $|\I_0\>\<\I_0|$ where
\begin{equation}
|\I_0\> := \frac{1}{\sqrt{d_0}} \sum_k |E_k^0\>, 
\end{equation}
where $d_0$ is the dimension of $H_0$. It is readily checked that for this case
\begin{equation}
\Gamma(|\I_0\>\<\I_0|) = |\gamma_0\>\<\gamma_0|,
\end{equation}
where $|\gamma_0\> = \frac{1}{\sqrt{Z_0}} \sum_k e^{-\beta E_k^0/2} |E_k^0\> $ is the coherent Gibbs state with respect to the initial Hamiltonian $H_0$. Or, in terms of entanglement we can consider the maximally entangled state $|\phi^+\>_{AB} =\frac{1}{\sqrt{d}} \sum_{k=0}^{d-1} |E_k\>_A\otimes |E_k\>_B $ on two quantum systems of dimension $d$. In this case, coherent thermalization of the maximally entangled state with respect to $A$ leads to
\begin{align}
\Gamma(|\phi^+\>\<\phi^+| )& = |\tilde{\phi}\>\<\tilde{\phi}| \\
|\tilde{\phi}\>&= \frac{1}{\sqrt{Z}} \sum_k \sqrt{e^{-\beta E_k}} |E_k\>_A \otimes |E_k\>_B,
\end{align}
which is the thermofield double state \cite{ThermoFieldDouble} in high energy physics and condensed matter.

These examples justify $X\mapsto \Gamma (X)$ as describing a formal coherent thermalisation: if no coherences are present $\Gamma$ provides the thermal Gibbs state; pure states are always sent to pure states; and coherent superpositions over energy are re-weighted with Gibbs factors at inverse temperature $\beta$. 
However, it is important to emphasise that $\Gamma$ is not a physical transformation but an invertible mapping on states that pairs $(\rho,\Gamma(\rho))$, as being distinguished --  in the same way as time-reversal pairs $(\rho,\Theta \rho \Theta^\dagger)$.

The time-reversal operator pairs a forward temporal-direction with a backward temporal-direction, while $\Gamma$ pairs a Hamiltonian constraint with a thermodynamic constraint. It might seem strange that we pair a mechanical constraint with one that has an explicit thermodynamic feature (namely a temperature), but the need for this arises from not necessarily having sharp microstate properties in quantum superpositions and it is this indeterminacy that requires the Gibbs-rescaling. Such a scenario would also arise in a classical example where instead of doing a Crooks relation on microstates $(\x,\p)$ one follows the same construction on distinguished distributions $\{q_k(\x,\p)\}$, e.g. if we took coarse-grained resolutions of finite sized support on phase space. The $q_k(x,p)$ would also have to be Gibbs-rescaled in order to obtain a coarse-grained Crooks relation (this is the classical analog of the above example of projector $\Pi$ that is not of rank-1). 

For the case where the constraint is the above $\Pi_0$, the pairing $(\Pi_0, \Gamma(\Pi_0))$ could potentially be viewed as a statement of the equivalence of the microcanonical and canonical ensembles in a similar vein to recent classical results \cite{MicrocanonicalMap}. This perspective can be formulated in terms of dualities between constrained and unconstrained optimization problems, and we conjecture that the present Gibbs-rescaling could be naturally described in such terms. We leave this to future analysis.

Beyond its role with time-reversal in identifying the correct trajectories to pair, we note that the transformation $\Gamma$ is not a linear physical map but instead amounts to an change in our \emph{description} of the system. In particular, one that is covariant under energy conserving unitaries $V$, with $V\Gamma(\rho)V^\dagger = \Gamma(V\rho V^\dagger)$. This structure arises naturally in fluctuation settings and scenarios in which one reverses a general quantum operation. In \cite{CrooksTimeReversal} Crooks considered the reversal of a Markov process, and found the resulting transformation generates factors of $e^{-\beta E_k/2}/\sqrt{Z}$ that Gibbs re-scale all probability distributions. The above mapping should be viewed as the extension of this to a fully coherent setting. The transformation $\Gamma(X)$ also arises in the context of the Petz recovery map \cite{Petz1986} for quantum channels, and of course in the paper \cite{Aberg} by \AA berg, where it is discussed in more detail. Its form also arises in quantum-mechanical settings for the transition between microscopic and macroscopic descriptions of thermodynamic systems \cite{FaistThesis}, which is perhaps the most appropriate perspective in light of the above constraint discussion. It would be of interest to obtain an independent operational analysis of this transformation (perhaps as a coherent form of a maximum entropy principle). We do not expand any more on this here, but instead refer the reader to \cite{HyukJoon} and \citep{Aberg} for more discussion.

\subsection{A key relation between forwards and reverse protocols for quantum systems}

Given the preceding discussion we can derive the coherent fluctuation relation, however this follows from a core structure that is describable on a single quantum system. This was discussed in \cite{Aberg}, but here we provide a slightly modified approach that will prove useful for what follows. Let $S_{\mbox{\tiny tot}}$ be any quantum system, with Hamiltonian $H_{\mbox{\tiny tot}}$. We consider the following sequence that abstracts the classical notion of a `trajectory' to a coherent form for $S_{\mbox{\tiny tot}}$:
\begin{enumerate}
\item A measurement is done on $S_{\mbox{\tiny tot}}$  with outcome given by a (not necessarily rank-1) projector $\Pi_0$.
\item Coherent thermalisation of $\Pi_0$ with respect to $H_{\mbox{\tiny tot}}$ and at inverse temperature $\beta$ updates the system to the state $\Gamma (\Pi_0)$.
\item The system $S_{\mbox{\tiny tot}}$  evolves unitarily under a unitary $V$ that commutes with $H_{\mbox{\tiny tot}}$.
\item A final measurement is done on $S_{\mbox{\tiny tot}}$  with projective outcome $\Pi_1$ (again not necessarily rank-1).
\end{enumerate}
We denote the probability of this sequence as $P[\Pi_1 | \Pi_0]$, and is given explicitly as
\begin{equation}\label{basic-trajectory}
P[\Pi_1 | \tilde{\Pi}_0] = \tr [\Pi_1 V (\Gamma(\Pi_0)) V^\dagger].
\end{equation}
The expression can be viewed as the probability of a particular `trajectory' under the above protocol. 

The fluctuation setting also requires a notion of a time-reversal of a trajectory. In quantum mechanics, time-reversal is an anti-unitary transformation at the level of the Hilbert space. A state transforms as $|\psi\> \rightarrow \Theta |\psi\>$ where $\Theta$ is both anti-linear (e.g. $\Theta (\alpha |\psi\>) = \alpha^* \Theta|\psi\>$ for any $\alpha \in \mathbb{C}$) and $\Theta^\dagger \Theta = \Theta \Theta^\dagger = \I$. However any anti-unitary $\Theta $ can be written as $KU$ where $U$ is some unitary and $K$ is complex conjugation in a preferred basis. At the level of Hermitian operators (e.g. projectors, observables, quantum states...) for which we have that $X = X^\dagger$, this complex conjugation $X \rightarrow K^\dagger X K = X^* $ is equivalent to simply taking a transpose of the operator $X \rightarrow X^T$. For the case of time-reversal, we can identify the time-reversal of a state $\rho$ as $\rho \mapsto \rho^* = \rho^T$. An inspection of the off-diagonal components of $\rho$ in the energy eigenbasis (which is the relevant basis for the harmonic oscillator system below) makes this clearer: a typical component $\rho_{ij}$ will oscillate in time as $e^{-i \omega_{ij}t}$ and thus by taking the complex conjugation one obtains an oscillation of $e^{+i \omega_{ij}t}$ for the matrix component -- in other words one swaps the positive and negative modes in the quantum state.

Since $\tr[ X] = \tr [X^T]$ for any Hilbert space operator $X$, and $(XY)^T = Y^T X^T$, we see that Equation (\ref{basic-trajectory}) can be re-written as
\begin{align}
P[\Pi_1 | \tilde{\Pi}_0)] &=\tr [(\Pi_1 V \Gamma(\Pi_0) V^\dagger)^T] \\
&=\tr [( V^\dagger)^T\Gamma(\Pi_0)^T V^T \Pi_1^T].
\end{align}
where we have introduced the short-hand notation $\tilde{X} := \Gamma(X)$ for any operator $X$.
We now make the final assumption that $V$, in addition to $[V,H_{\mbox{\tiny tot}}] =0$, is also invariant under the above time reversal transformation -- namely $V=V^T$. Informally, this condition can be interpreted as assuming that the unitary $V$ does not inject in any microscopic time-asymmetry into the system, and thus any differences in probabilities between forward and reverse trajectories are purely due to the thermodynamic structure.

Using that $V=V^T$, and writing $X^* = X^T$ for any Hermitian $X$ we see that
\begin{align}
P[\Pi_1 | \tilde{\Pi}_0 ] &= \tr [( V^\dagger)\Gamma(\Pi_0)^T V \Pi_1^*] \nonumber \\
&= \frac{1}{\tr (e^{-\beta H_{\rm tot}} \Pi_0)}\tr [e^{-\beta H_{\rm tot}/2}\Pi_0^* e^{-\beta H_{\rm tot}/2}V \Pi_1^* V^\dagger]\nonumber \\
&= \frac{\tr (e^{-\beta H_{\rm tot}} \Pi_1^*)}{\tr (e^{-\beta H_{\rm tot}} \Pi_0)}\tr [\Pi_0^*V\Gamma(\Pi_1^* )V^\dagger]\nonumber \\
&= \frac{\tr (e^{-\beta H_{\rm tot}} \Pi_1)}{\tr (e^{-\beta H_{\rm tot}} \Pi_0)}P[\Pi_0^* | \tilde{\Pi}_1^*], 
\end{align}
where we used the fact that $H_{\mbox{\tiny tot}}^T = H_{\mbox{\tiny tot}}$ to write $\tr (e^{-\beta H_{\mbox{\tiny tot}}} \Pi_1^*)=\tr (e^{-\beta H_{\mbox{\tiny tot}}} \Pi_1)$. We can therefore express the ratio of an abstract `forward trajectory' to its `reversed trajectory' as
\begin{equation}\label{key}
\frac{P[\Pi_1 | \tilde{\Pi}_0]}{P[\Pi_0^* | \tilde{\Pi}_1^*]} = \frac{\tr (e^{-\beta H_{\mbox{\tiny tot}}} \Pi_1)}{\tr (e^{-\beta H_{\mbox{\tiny tot}}} \Pi_0)}.
\end{equation}
This is the key relation that we will use to obtain the fluctuation relations.

\subsection{Coherent Tasaki-Crooks relation.}

As already mentioned, we adopt an inclusive description in which the total system $S_{\mbox{\tiny tot}}$ is actually composed of two subsystems $S$ and $B$, and the total system $SB$ is energetically closed, but not dynamically closed. As discussed earlier, one should view the initially incoherent $B$ system as sub-divided into two components $B=B_1B_2$ where $B_1$ is a large, uncontrolled thermal bath while $B_2$ corresponds to a quantum system for which we can control its degrees of freedom -- for example we can vary its Hamiltonian in some manner. The total Hamiltonian for $SB$ is as before and has eigenspaces given by the decomposition in Equation (\ref{split}).

We now consider an initial constraint $\mathcal{C}_0$ for $t=0$ and an final constraint $\mathcal{C}_1$ for $t=1$, given by
\begin{align}
\mathcal{C}_0 &: \mbox{   $B$ has Hamiltonian $H_0$ and  $S$ is in a state $|\psi_0\>$.} \nonumber \\
\mathcal{C}_1 &: \mbox{  $B$ has Hamiltonian $H_1$ and $S$ is in a state $|\psi_1\>$.} \nonumber 
\end{align}
Mathematically, these are represented by $\Pi_{SB,0}= |\psi_0\>\<\psi_0| \otimes \Pi_0 $ and $\Pi_{SB,1}= |\psi_1\>\<\psi_1| \otimes \Pi_1 $ respectively. where $\Pi_k$ is now the projector onto $\H_B^k$ for the system $B$.

The coherent thermalisation for the initial set-up gives $\Gamma(\Pi_{SB,0}) =  |\tilde{\psi}_0\>\<\tilde{\psi}_0| \otimes \gamma_0 $ where $|\tilde{\psi}_0\>\<\tilde{\psi}_0| := \Gamma_S (|\psi_0\>\<\psi_0|)$ where $\Gamma_S$ is Gibbs-rescaling purely with respect to the Hamiltonian $H_S$. If we now substitute these components into Equation (\ref{key}) we obtain

\begin{center}
\begin{figure}[t]
\includegraphics[width=8cm]{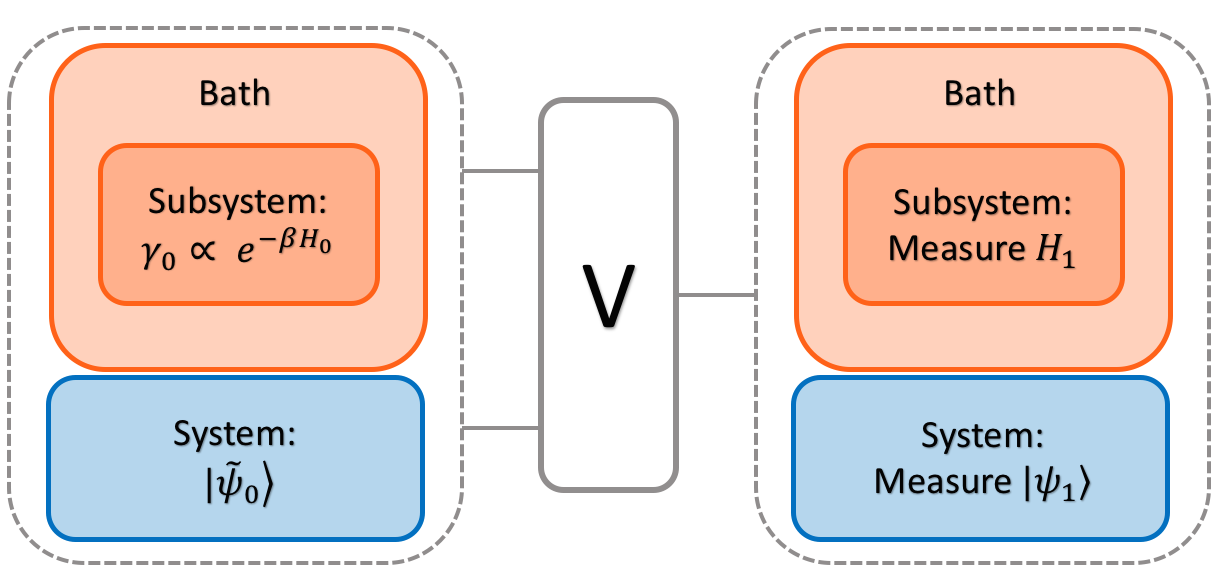}
\caption{\textbf{Global system schematic.} A system $S$, initially uncorrelated with a thermal system $B$, evolve under a microscopically energy conserving unitary $V$, to produce open-system dynamics on $S$. }  
\label{fig:FTSchematic}
\end{figure}
\end{center}
\vspace{-10pt}

\begin{align}
\frac{P[\Pi_{SB,1} | \tilde{\Pi}_{SB,0}]}{P[\Pi_{SB,0}^* | \tilde{\Pi}_{SB,1}^*]} &= \frac{\tr (e^{-\beta H_{SB}} |\psi_1\>\<\psi_1| \otimes \Pi_1 )}{\tr (e^{-\beta H_{SB}} |\psi_0\>\<\psi_0| \otimes \Pi_0)} \nonumber \\
&= \frac{\tr (e^{-\beta H_1})}{\tr(e^{-\beta H_0})} \frac{\tr (e^{-\beta H_S} |\psi_1\>\<\psi_1|)}{\tr (e^{-\beta H_S} |\psi_0\>\<\psi_0|)}  \nonumber\\
&=e^{-\beta \Delta F - \Delta \Lambda},
\end{align}
where $\beta F_k = -\log \tr e^{-\beta H_k}$ is the free energy of $B$ at time $t=k$, $\Delta F = F_1 - F_0$, and $\Delta \Lambda := \Lambda(\beta, \psi_1) - \Lambda(\beta, \psi_0)$, where we define $\Lambda(\beta, \rho)$ the \emph{effective potential} of a state $\rho$ at inverse temperature $\beta$ as
\begin{equation}
\Lambda(\beta, \rho) := - \log \tr (e^{-\beta H_S} \rho).
\end{equation}

With the understanding that the Hamiltonian of $B$ changes from $H_0$ to $H_1$ we can simply write
\begin{equation}
\frac{P[\psi_1 | \tilde{\psi}_0]}{P[\psi_0^* | \tilde{\psi}_1^*]} =  ^{-\beta \Delta F - \Delta \Lambda},
\end{equation}
as the final coherent Crooks relation for the ratio of the probabilities of a given forward coherent trajectory to the probability of its reverse trajectory.

It is important to note that if one ranges over all possible POVMs on $S$ then the above relation is equivalent to the abstract channel relation that was first derived by \r{A}berg in \cite{Aberg}. The difference here is that we focus on projective measurements that admit a simple constraint interpretation and introduce the effective potential. Both of these turn out to be key in our physical analysis of the relation. We now state the core quantum fluctuation theorem in terms of the effective potential of the quantum states. 
\begin{theorem}[\textbf{Effective potential form}]\label{thm:Fluctuation}
Let $H = H_S \otimes \I_B + \I_S \otimes H_B$ be a microscopic, time-reversal invariant Hamiltonian and assume $B$ begins in a thermal state at inverse temperature $\beta = 1/(kT)$. Additionally, assume the dynamics admit a microscopic, time-reversal invariant unitary $V$ such that $[V,H]=0$. Then with $P[\psi_1|\tilde{\psi}_0]$ and $P[\psi_0^*|\tilde{\psi}_1^*]$ defined as above, the  transition probabilities satisfy:
\begin{equation}
\frac{P[\psi_1 | \tilde{\psi}_0]}{P[\psi_0^* | \tilde{\psi}_1^*]} = e^{-\beta\Delta F - (\Lambda (\beta, \psi_1)- \Lambda (\beta, \psi_0))},
\end{equation}
where  $\Delta F = (F_1 - F_0)$ is the difference in free energies with respect to the final and initial Hamiltonians of $B$, and $\Lambda(\beta,\rho)$ is the effective potential for any state $\rho$ on $S$ at inverse temperature $\beta$.
\end{theorem}
The effective potential $\Lambda(\beta, \psi)$ can be viewed as a logarithm of the Laplace transform for the quantum state $|\psi\>$ of $S$ with respect to $H_S$, and as such it corresponds to the cumulant generating function \cite{Feller} for the measurement statistics of energy in the quantum state $|\psi\>$. Thus by expansion in terms of the cumulants, the fluctuation relation can be re-expressed as
\begin{equation}
\frac{P[\psi_1 | \tilde{\psi}_0]}{P[\psi_0^* | \tilde{\psi}_1^*]}  = e^{-\beta \Delta F - \sum_{n \ge 1}(-1)^n\frac{\beta^n }{n!} \Delta \kappa_n},
\end{equation}
where $\Delta \kappa_n = \kappa_n(\p_1) - \kappa_n(\p_0)$, and $\kappa_n(\p_k)$ is the $n^{\rm th}$ cumulant for the random variable obtained if one measured $H_S$ in the state $|\psi_k\>$, and which has probability distribution $\p_k$ over energy. Under the assumption that all cumulants are finite, one immediately sees that in the high temperature regime one recovers a classical form for the coherent Crooks relation in which the first order $\beta$ term dominates.
\begin{corollary}
In the high temperature limit,
\begin{equation}
 \frac{P[\psi_1|\tilde{\psi}_0]}{P[\psi_0^*|\tilde{\psi}_1^*]}  = e^{\beta (W_B - \Delta F) + \mathcal{O}(\beta^2)}
\end{equation}
where $ W_B := \<H_B\>_{t=1} - \<H_B\>_{t=0} = - [\< H_S \>_{\psi_1} - \< H_S \>_{\psi_0}] $.
\end{corollary}
Note that since the composite system $SB$ is energetically closed a change in energy in $B$, as measured by the first moment $\<H_B\>$, is identified with a corresponding change in energy in $S$.

A more interesting question is how to interpret the higher order corrections  $\frac{(-1)^n\beta^n \Delta \kappa_n}{n!}$ for $n>1$. These vanish if the state of $S$ is an energy eigenstate, and so arise from the non-trivial coherent structure of the input and output pure states on $S$. One might suspect that they are measures of quantum coherence, but this turns out to not be the case, for example $\kappa_3$ can both increase and decrease under incoherent quantum operations and therefore cannot be a genuine measure of coherence \cite{MarvianSpekkens,MarvianThesis}.

Despite the fact that a general cumulant $\kappa_n$ is not a measure of coherence it turns out that $\kappa_2$ is distinguished and admits such an interpretation. The second cumulant $\kappa_2$ is the variance of energy in a quantum state $|\psi\>$, and it has been shown that if one restricts to pure quantum states then in the asymptotic regime of many copies of a state $|\psi\>$ there is an essentially unique way to quantify coherence \cite{ReferenceFrame} between different eigenspaces, and is given by $\chi(\psi):=4 \pi \kappa_2(\psi)$. With this in mind, we next provide a decomposition of the effective potential $\Lambda(\beta, \psi)$ of a pure quantum state into energy and coherence contributions.

\subsection{Separation of the effective potential into energetic and coherent contributions}

We can show that the general fluctuation relation takes a particularly simple form by exploiting how $\Lambda(\beta, \psi)$ varies as a function of the inverse temperature $\beta$. We assume that $\Lambda(\beta, \rho)$ is differentiable to second-order with respect to $\beta$. Firstly, it is clear that $\Lambda(0,\rho) = 0$ for any quantum state $\rho$, and moreover that
\begin{equation}
\partial_\beta \Lambda(\beta, \rho)|_{\beta =0} = \<H\>_\rho.
\end{equation}
Looking at the second derivative $\partial_\beta^2 \Lambda(\beta, \rho)$ we find that
\begin{equation}
-\partial_\beta^2 \Lambda(\beta, \rho) = \frac{\tr [ \rho H^2 e^{-\beta H}]}{\tr ( e^{-\beta H}\rho)} - \left ( \frac{\tr [\rho H e^{-\beta H}]}{ \tr ( e^{-\beta H}\rho)}\right)^2.
\end{equation}
The right-hand side of this only depends on the distribution $\p = (p_k)$ over energy eigenstates for $\rho$. We therefore define $p_k = \tr [\Pi_k \rho]$ where $\Pi_k$ is the projector onto the energy eigenspace with energy $E_k$, and so
\begin{equation}
- \partial_\beta^2 \Lambda(\beta, \rho) = \frac{\sum_k[ p_k E_k^2 e^{-\beta E_k}]}{\sum_j( e^{-\beta E_j}p_j)} - \left ( \frac{\sum_k[p_k E_k e^{-\beta E_k}]}{ \sum_j ( e^{-\beta E_j}p_j)}\right)^2.
\end{equation}
Now define the distribution
\begin{equation}
\tilde{p}_k = \frac{e^{-\beta E_k}}{ \sum_j  e^{-\beta E_j}p_j } p_k \mbox{ for all } k.
\end{equation}
This is a Gibbs re-scaling of the distribution $p_k$ at inverse temperature $\beta$. In particular, we see that $\partial_\beta^2 \Lambda(\beta, \rho) = - \kappa_2(\tilde{\p})$, the variance of energy under the re-scaled distribution $\tilde{p}_k$.

We can now apply the second--order Mean Value Theorem to $\Lambda(\beta, \rho)$ to deduce that for some $\beta_m \in [0,\beta]$ we have
\begin{align}
\Lambda(\beta,\rho) &= \Lambda(0,\rho) + \beta \partial_\beta\Lambda(\beta, \rho)|_{\beta =0} \\ &{\ \ } + \frac{1}{2} \beta^2 \left [ \partial_\beta^2\Lambda(\beta, \rho)|_{\beta = \beta_m} \right] \nonumber \\
 &=\beta \<H\>_\rho + \frac{1}{2} \beta^2 \left [ \partial_\beta^2 \Lambda(\beta, \rho)|_{\beta = \beta_m} \right ],
\end{align}
and thus we have that in general
\begin{equation}
\Lambda(\beta,\rho) =\beta \<H\>_\rho - \frac{1}{8\pi} \beta^2\chi_m(\tilde{\rho}),
\end{equation}
where we have used that $\chi(\rho) = 4\pi \kappa_2(\p)$, with the short-hand notation $\tilde{\rho} = \Gamma(\rho)$.
Here $\chi_m (\tilde{\rho})$ is non-trivial and should be interpreted with care. It involves computing $\chi$ on a Gibbs re-scaling of the quantum state, and then evaluated at an effective inverse temperature $\beta_m \le \beta$ that is determined by the statistics of the original state. 

The utility of this is that for a pure state $|\psi\>$, the re-scaled state $|\tilde{\psi}\>$ is also pure, and thus $\chi$, the variance in energy, is always a genuine measure of coherence when restricted to pure states \cite{GourSpekkens}. Therefore, the higher order cumulants for $n>1$ arise from the coherent structure of $|\psi\>$ and while individual cumulants are not coherence measures, one can still deduce that the sum of all higher order terms can be reduced to a single coherence measure of a rescaled pure quantum state. We refer to $\chi_m(\tilde{\psi})$ as the \emph{mean coherence at inverse temperature $\beta$} of the pure quantum state $|\psi\>$, and it allows the effective potential to split in a remarkably simple way as
\begin{equation}
\Lambda(\beta, \psi) = \frac{\mbox{ (mean energy)}}{kT} - \frac{\mbox{(mean coherence)}}{8\pi(kT)^2},
\end{equation}
however the way in which $\chi_m(\tilde{\psi})$ corresponds to ``mean coherence'' is subtle. When $\beta_m = \beta$ then $\chi_m$ is precisely the amount of coherence in the physically prepared initial state $\Gamma(\psi) = \tilde{\psi}$ after Gibbs re-scaling. However more generally, $\chi_m$ is the amount of coherence in the pure state $\propto \exp[-(\beta_m - \beta)H_S] |\tilde{\psi}\>$ with $\beta_m$ determined by the energy statistics of the quantum state. We go into more detail on $\chi_m$ in the next section.

In terms of the mean coherence the core fluctuation relation can be re-stated as follows.
\begin{theorem}[\textbf{Mean coherence \& mean energy decomposition}]
Let the assumptions Theorem \ref{thm:Fluctuation} hold. Then with $\chi_m (\tilde{\rho})$ defined as before at an inverse temperature $\beta$, for any two pure states $|\psi_0 \>$ and $|\psi_1\>$, the fluctuation theorem takes the form
\begin{equation}
\frac{P[\psi_1 | \tilde{\psi}_0]}{P[\psi_0^* | \tilde{\psi}_1^*]} = e^{- \beta  \Delta F - \beta W_S + \frac{\beta^2}{8\pi} \Delta \chi },
\end{equation}
where $W_S = \< H_S \>_{\psi_1} - \< H_S \>_{\psi_0}$ and $\Delta \chi = \chi_m (\tilde{\psi}_1) - \chi_{m'} (\tilde{\psi_0})$. 
\end{theorem}
This fully separates the change in the first moment $\<H_S\>$ from the higher order corrections. Note that $\chi_m (\tilde{\psi}_1)$ and $\chi_{m'} (\tilde{\psi}_0)$ are in general evaluated at two different temperatures $\beta_m$ and $\beta_{m'}$ respectively. Such a difference depends on $\beta$ and also the original statistics of each state $|\psi_0\>$ and $|\psi_1\>$.

\subsection{An explicit form for the mean coherence at inverse temperature $\beta$}
The above form of $\chi_m$, and its dependence on temperature is slightly opaque, however one can obtain a simple expression for computing it that makes clear the different contributions. To this end we can make use of the following result.

\begin{lemma}\label{free-energy-like} Given a quantum system $S$ with Hamiltonian $H_S$ and $\{\Pi_k\}$ the projectors onto the energy eigenspaces of $H_S$, let us denote the de-phased form of state in the energy basis  by $\D(\rho) = \sum_k \Pi_k \rho \Pi_k$. Then, for any quantum state $\rho$, the effective potential $\Lambda(\beta, \rho)$ is given by
\begin{equation}
\Lambda(\beta, \rho) = \min_\sigma \{   \beta \<H_S\>_\sigma + S(\sigma||\D(\rho)) \},
\end{equation}
where the minimization is taken over all quantum states $\sigma$ of $S$, and $S(\sigma||\rho) = \tr [ \sigma \log \sigma - \sigma \log \rho]$ is the relative entropy function. Moreover, the minimization is attained for the state $\sigma = \Gamma(\D(\rho))$. 
\end{lemma}

The proof of this is provided in the Appendix \ref{Proof of Lemma III.3}. If the probability distribution over energy of $\rho$ is $\p = (p_k)$ then this is unaffected by the dephasing $\rho \mapsto \D(\rho)$. Denoting by $\tilde{\p}$ the Gibbs rescaling of $\p$, we therefore have that
\begin{equation}
  \Lambda(\beta, \rho) = \beta \<H_S\>_{\tilde{\p}} + S(\tilde{\p} || \p),
\end{equation}
where the relative entropy term is now the classical relative entropy for distributions and $\<H_S\>_{\tilde{\p}} := \sum_k E_k \tilde{p}_k$ is the expectation value of energy for the classical distribution $\tilde{\p}$. We thus have the following.
\begin{theorem}\label{thm:mean-coherence} Given a quantum system $S$ with Hamiltonian $H_S$, the mean coherence at inverse temperature $\beta$ of a quantum state $|\psi\>$ is given by
\begin{equation}\label{mean-coherence}
\frac{\beta^2}{8 \pi}\chi_m (\tilde{\psi}) =  \beta(\<H_S\>_{\p} - \<H_S\>_{\tilde{\p}}) - S(\tilde{\p} ||\p).
\end{equation}
where $\p=(p_k)$ is the distribution over energy of $\psi$ and $\tilde{\p}$ is the distribution over energy for $|\tilde{\psi}\>\<\tilde{\psi}| = \Gamma(|\psi\>\<\psi|)$.
\end{theorem}
Everything on the right-hand side of Equation (\ref{mean-coherence}) is easily computable, and its form sheds light on what the mean coherence at inverse temperature $\beta$ actually depends on. The first term quantifies the degree to which the coherent thermalisation raises or lowers the expected energy of $|\psi\>$ in units of $\beta$, while the second term quantifies the degree to which the coherent thermalisation makes the state $|\tilde{\psi}\>$ distinguishable from the state $|\psi\>$ (in the sense of hypothesis testing optimized over all possible quantum measurements \cite{QuantumSteinsLemma}). Since the left-hand side is always non-negative, we also see that energy term is never smaller than the relative entropy term.

\begin{example}
Consider the system $S$ in the uniform superposition over $d$ energy eigenstates $|\psi \> = \frac{1}{\sqrt{d}} \sum_{n =0}^{d-1} |E_n\>$. Then, as discussed earlier, this transforms under $\Gamma$ to the coherent Gibbs state $|\tilde{\psi} \> = |\gamma\>$. Then by direct calculation with equation (\ref{mean-coherence}), we find that
\begin{equation}
\frac{\beta^2}{8 \pi} \chi_m (|\gamma\> \< \gamma|) = \beta [F \left(\frac{1}{d}\mathds{1}_d \right) - F \left(\gamma\right)],
\end{equation}
where $F(\rho) := \< H \>_\rho - k_B T S(\rho)$ is the free-energy of the state $\rho$  and $\gamma$ is the Gibbs state. Thus the mean coherence is proportional to the difference between the free energy of $S$ with a trivial Hamiltonian $H_S=0$ and the free energy of $S$ with non-trivial Hamiltonian $H_S \ne 0$ at inverse temperature $\beta$. 
\end{example}

Finally, note that while $\chi$ is a genuine measure of coherence for any pure quantum state, this is different from $\chi_m$ being a monotone under the general dynamics of the setting. While this is not immediately obvious from the discussion, we highlight that the reduced dynamics for any state $\rho$ of $S$ is of the form
\begin{equation}
\E(\rho) = \tr_B V(\rho \otimes \gamma_0 )V^\dagger,
\end{equation}
and it is readily seen that this quantum operation is time-translation covariant, and so the total coherence between energy eigenspaces of $S$ can never increase \cite{Harmonic,CoherentNature,QuantCoherence}. In addition, given any elements $M_0, M_1$ that occur in some POVM $\M = \{M_k\}$, one can always include them in a covariant measurement on the system, by simply considering the orbit of each under the adjoint action $X\rightarrow U(t) X U(t)^\dagger$ for all $t$ and if these are not in the measurement set, expanding the set to include them. Thus, a POVM $\M = \{M_k \}$ exists that contains both $M_0$ and $M_1$ and the total non-selective evolution is time-translation covariant. Therefore there is no coherence being exploited to perform the measurement \cite{WayTheorem} and so over the whole procedure on $S$ we have that the coherence in $S$ can never increase.

\section{Coherent work processes, quantum fluctuation relations and the semi-classical limit.}\label{sec: coherent fluctuation}

We can now show that the earlier notion of a coherent work process fits neatly into the inclusive quantum fluctuation setting. Recall that the classical Crooks relation takes the form
\begin{equation}\label{classical fluctuations}
\frac{P[w | \gamma_0, \mathcal{P}]}{P[-w|\gamma_1, \mathcal{P}^*]} = e^{-\beta (\Delta F - w)}
\end{equation}
that relates the probability of an amount of work $w$ for a forward trajectory with $-w$ for the reversed trajectory. As already discussed this quantity $w$ is always associated with some auxiliary system within an inclusive description. It is readily seen that the previous quantum fluctuation relations can be re-cast in the following form that provides a natural extension of the classical result.

\begin{theorem}[\textbf{Fluctuations and coherent work processes}]\label{thm:CoherentWorkFT}
Let $S$ be a quantum system with Hamiltonian $H_S$ and let $|\psi_0\>_S$ and $|\psi_1\>_S$ be two pure states of $S$ that have energy statistics with finite cumulants of all orders. Then the following are equivalent:
\begin{enumerate}
\item The pure states $|\psi_0\>_S$ and $|\psi_1\>_S$ are coherently connected with coherent work output/input $\omega$ on an auxiliary system $A$.
\item Within the fluctuation relation context with a thermal system $B$ and quantum system $S$ we have that
\begin{equation}
\frac{P[\psi_1 | \tilde{\psi}_0]}{P[\psi^*_0 | \tilde{\psi}^*_1]}   = e^{-\beta(\Delta F \pm kT\Lambda(\beta,\omega))},
\end{equation} 
for all inverse temperatures $\beta \ge 0$, for some state $|\omega\>_A$ with finite cumulants on an auxiliary system $A$ and some choice of sign before $\Lambda(\beta, \omega)$.
\end{enumerate}
\end{theorem}

\begin{figure}[t]
\begin{center}
\includegraphics[width=6cm]{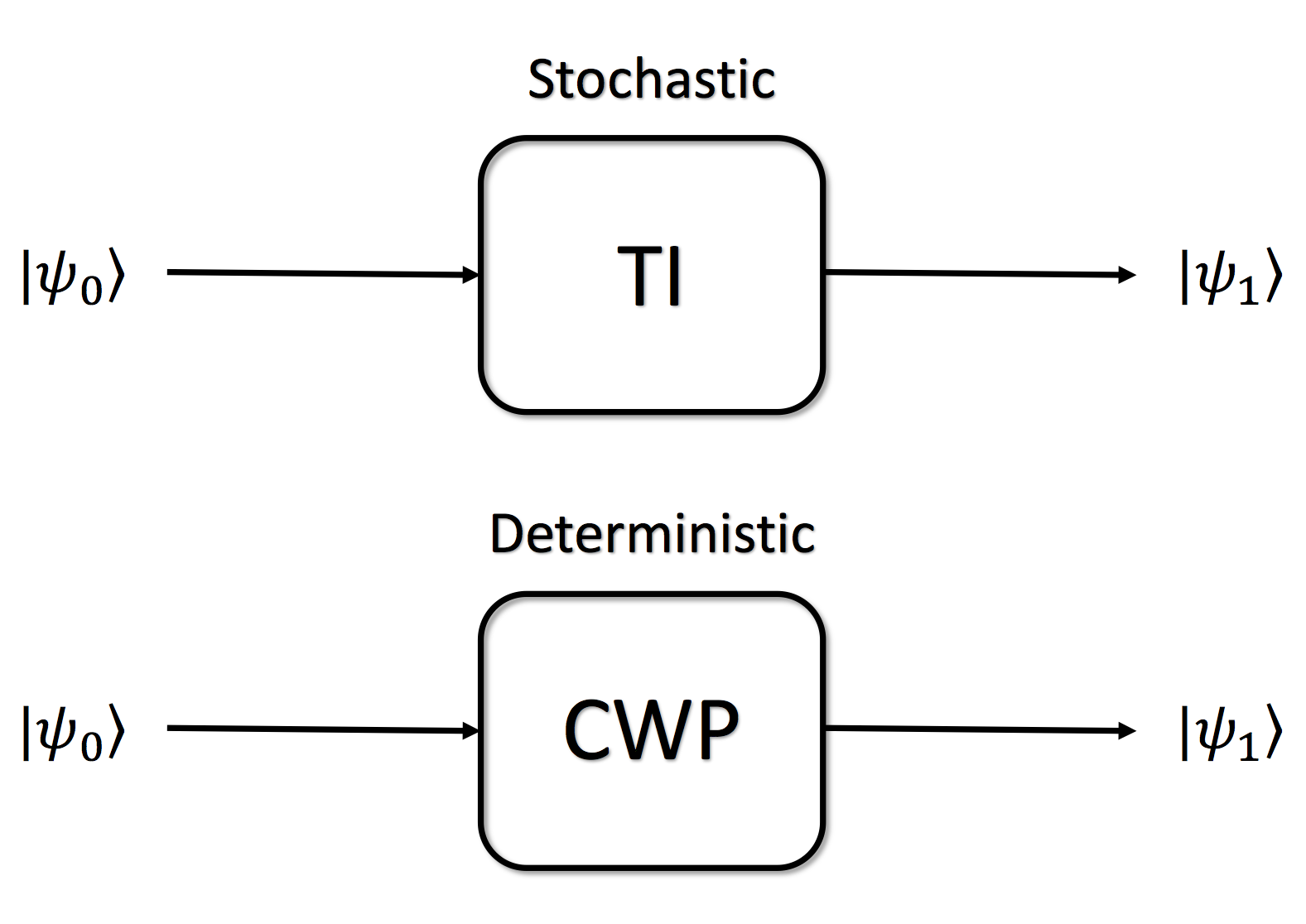}
\end{center}
\caption{\textbf{Coherent work processes (CWP) and thermal interactions (TI).} We introduced coherent work processes which are deterministic processes that transition a pure state to another pure state. Thermal interactions are stochastic, hence the output is a probabilistic mixture of pure states. In the context of the fluctuation relation, we post-select a particular pure state. The coherent work output/input $|\omega\>$ associated to the initial and final quantum states contributes $\Lambda(\beta,\omega)$ to the fluctuation relation and generalises the classical term $\beta w$ from the standard Crooks relation.  
}
\end{figure}

We note that for $|\psi_0\>_S$ and $|\psi_1\>_S$ being energy eigenstates we have that $|\omega \>_A$ is an energy eigenstate $|w \>_A $, for which $\Lambda(\beta, \omega) = \beta w$ for some sharp value $w \in \mathbb{R}$, and so the above expression reduces to a classical Crooks relation form. The proof of the theorem is provided in Appendix \ref{app:coherent work}. 

\subsection{Physical interpretation of the result}
Interpreting this theorem requires some care. Statement 1 is simply a deterministic transformation between pure states with coherence on a quantum system $S$. 

In contrast, statement 2 involves a stochastic thermal process in which initial coherence in $S$ interacts with a thermal environment $B$ and projective measurements at the start and at the end is assumed. Moreover this coherence is not outputted to some ordered degree of freedom but is dissipated into the large thermal environment $B$. The appearance of an auxiliary system $A$ in statement 2 is independent of the actual thermodynamic context and takes no part in the process.

This shows that coherent work processes and coherently connectedness relate to the coherent fluctuation in a simple way, via the coherent work output/input $\omega_A$. Moreover, the numerical contribution of $\omega_A$ to the fluctuation relation exponent is given simply by the effective potential of $\omega_A$. From a structural perspective we may therefore summarise this conclusion as follows:

\begin{center}
\emph{Coherent work processes are to coherent fluctuation relations what deterministic Newtonian work processes are to classical fluctuation relations.}
\end{center}
Moreover in the semi-classical regime, where coherence becomes negligible, the two fluctuation relations converge, and therefore we have that the coherent framework is a faithful and consistent extension of the classical work relations into quantum mechanics.

\subsection{Mixed state coherent work output/input} 
The connection with coherent work processes can be made more explicit by observing that the Gibbs state $\gamma_0$ is a probabilistic mixture of energy eigenstates and therefore the unitary part of the dynamics can be written as
\begin{align}
V(\psi_0 \otimes \gamma_0) V^\dagger & =\sum_k \frac{e^{-\beta E_k^0}}{Z_0} V (|\psi_0\>\<\psi_0| \otimes |E^0_k\>\<E_k^0|)V^\dagger.
\end{align}
In general $|\varphi_k\>_{SB} =V [ |\psi_0\>\otimes |E^0_k\>]$ is not a product state over $S$ and $B$, and so does not describe a coherent work process in itself, however the projective measurement on $S$ will collapse $|\varphi_k\>$ to a product state $|\psi_1\>\otimes |\phi^1_k\>$ and so the overall effect is the probabilistic transition $|\psi_0\>\otimes |E_k^0\> \rightarrow |\psi_1 \> \otimes |\phi^1_k\>$. Thus the particular energy eigenstate $|E^0_k\>$ in the Gibbs ensemble acts as a reference energy level and $|E^0_k\> \rightarrow |\phi^1_k\>$ can be viewed as a coherent work output conditioned on $|E_k^0\>$.

We can now link with more abstract formalism by considering the evolution $\E$ of $S$ and looking at what is called the \emph{complementary channel} \cite{WatrousQI} $\bar{\E}$ from $S$ into $B$ given by
\begin{equation}
\bar{\E} (\rho) := \tr_S V(\rho \otimes \gamma_0)V^\dagger.
\end{equation}
This provides a stochastic mixture of coherent work outputs, from the system $S$ into the bath system $B$. 

In fact this perspective can be taken as a more general formulation of coherent work processes in which a general state $\rho$ of a system $S$ transforms through some energy-conserving interaction with a second system $A$ in a default energy eigenstate $|0\>\<0|$ as $\rho \rightarrow \E(\rho) = \tr_A V (\rho \otimes |0\>\<0|) V^\dagger$. Then we define a general coherent work process as
\begin{equation}
\rho \stackrel{\omega}{\longrightarrow} \E(\rho).
\end{equation}
where $\omega = \bar{\E}(\rho)$ is now the mixed state coherent output for the process. 

The set of such quantum transformations are simply the time-translation covariant quantum operations \cite{Harmonic,QuantCoherence,CoherentNature}, and so in the same way as before we can say that two quantum states $\rho$ and $\sigma$ are coherently connected if there exists a time-translation covariant channel $\E$ from one state to the other. The coherent work output/input is given via the complementary channel $\bar{\E}$ acting on the input system. We leave the study of such features to future work, where coherent work processes could prove useful tools in studying asymmetry theory.

In the next section, we return to more concrete systems, and briefly outline how coherent work processes are experimentally accessible in existing trapped ion proposals.

\subsection{Coherent fluctuation relations in the experimental semi-classical regime}

We have shown that our analysis of coherent work processes, in which classical deterministic transitions are superposed together while respecting energy conservation, appears naturally within the fluctuation relation. We can now re-visit the case that the quantum system $S$ is a harmonic oscillator system, and restrict to states in $\mathcal{C}$ from Section \ref{C-states}, which are closed under coherent work processes. This is particularly of interest because a fluctuation relation for this system can be demonstrated within existing trapped ion systems \cite{Zoe}, and therefore shows that the theoretical analysis presented here is of relevance to existing experimental work.

We can without loss of generality restrict to the canonical states $|\alpha,k\>$ in the set $\mathcal{C}$, and to simplify things further we consider the subset formed of pure coherent states (with $k = 0$), although the fully general case can also be easily computed. This gives rise to the same fluctuation relation as in \cite{Zoe}, which followed from the analysis of \cite{Aberg}. For a coherent state $|\alpha\>$ of a quantum harmonic oscillator with Hamiltonian $H_S = h \nu (a^\dagger a + \frac{1}{2})$, the effective potential is readily computed and takes the form
\begin{equation}
\Lambda (\beta , \alpha) = \frac{1}{2}\beta h \nu + \frac{1}{h\nu}\left (\<H_S\>_\alpha - \frac{1}{2}h\nu\right) (1 - e^{-\beta h \nu}),
\end{equation}
where $\<H_S \>_\alpha := \<\alpha | H_S |\alpha\> = h\nu (|\alpha|^2 + 1/2)$.

A straightforward calculation leads to the following quantum fluctuation relation restricted to coherent states of an oscillator system $S$.
\begin{theorem}[\textbf{Semi-classical relation }\cite{Zoe}]\label{thm:FT forms}
Let the assumptions of Theorem \ref{thm:Fluctuation} hold, and let $S$ is a harmonic oscillator with Hamiltonian $H_S = h\nu (a^\dagger a + \frac{1}{2})$. Then for two coherent states of the system $|\alpha_0\>$ and $|\alpha_1\>$, the following holds
\begin{align}
\frac{P[\alpha_1|\tilde{\alpha}_0]}{P[\alpha_0^*|\tilde{\alpha}_1^*]}   &=  \exp \left [ -\frac{\Delta F }{k T}+  \frac{\bar{W}_B}{h \nu_{\rm \tiny{th}}} \right],
\end{align}
where  $h \nu_{th} = \< H_S \>_\gamma$ is the average energy of a Gibbs state $\gamma$ of a quantum harmonic oscillator, related to the thermal de Broglie wavelength $\lambda_{\rm \tiny{dB}} (T)$ via
\begin{equation}
h \nu_{\rm \tiny{th}} = \frac{h^2}{m \lambda_{\rm \tiny{dB}}(T)^2} + \frac{1}{2} h\nu,
\end{equation}
and $\bar{W}_B := - \frac{1}{2}(W_S + \tilde{W}_S)$, $W_S = \<H_S\>_{\alpha_1} - \<H_S\>_{\alpha_0}$, $\tilde{W}_S = \<H_S\>_{\tilde{\alpha}_1} - \<H_S\>_{\tilde{\alpha}_0} $.
\end{theorem}
The interpretation of this is quite natural, with only the term $\bar{W}$ needing care. We first note that $\bar{W}$ is the average of two changes in pure state energy: the energy change in $|\alpha_0\> \rightarrow |\alpha_1\>$ and its Gibbs rescaled version $|\tilde{\alpha}_0\> \rightarrow |\tilde{\alpha}_1\>$. Since Gibbs rescaling leaves energy eigenstates unchanged this implies that $\bar{W}$ reduces to the classical sharp energy transition in the case of zero coherences, otherwise the coherences provide a non-trivial distortion for a semi-classical `work' term\footnote{We do not refer to this as ``the work'' since it is purely restricted to the semi-classical case of coherent states, and there is no reason to believe that a physically sensible and unique ``work'' quantity should exist for general quantum systems.}. By comparing with the previous analysis one sees that $\bar{W}$ is not simply a change in first moments of $H$, and thus incorporates part of the mean coherence contribution from $\chi_m$. 

The reason that this form of the fluctuation relation is natural is that coherent states are the ``most classical'' states for the quantum system (they saturate the Heisenberg bound), and so gives rise to a fluctuation relation that is as close to the classical one as possible while still having coherent structure. This is reinforced by the fact that all the terms on the right-hand side, except for the `work' term $\tilde{W}$ are \emph{simple equilibrium properties} of a quantum system -- as in the standard Crooks relation. The principle difference is that the equilibrium temperature term $kT$ for the work is replaced by $h\nu_{\rm \tiny th}$ which is a quantum mechanical equilibrium property, related to the de Broglie thermal wavelength. This shows that in the high temperature regime where $kT \gg h \nu$, the equilibrium thermal fluctuations dominate the de Broglie thermal wavelength and we recover the classical regime $h \nu_{\rm \tiny th} \rightarrow kT$. In an intermediate temperature regime we get a smooth temperature-dependent distortion of the standard Crooks relation, while for very low temperatures we have $h \nu_{\rm \tiny th} \rightarrow 1/2 h\nu$ in a high coherence regime -- namely it is now the \emph{vacuum fluctuations} that dominate in the fluctuation relation.

One can also readily compute the mean coherence $\chi_m$ for the oscillator system, and find that for any coherent state $|\alpha\>$ it takes the form 
\begin{equation}
\chi_m =8 \pi |\alpha|^2 (kT)^2 [ \beta h\nu + e^{-\beta h \nu} -1 ].
\end{equation}
Expansion of the exponential gives that $\chi_m = 4 \pi |\alpha|^2 [ (h\nu)^2 - \frac{1}{3}\beta (h\nu)^3 + \cdots ]$ and so we see the temperature dependence only begins at cubic order in $h \nu$. We also see that $\lim_{h \rightarrow 0} \chi_m = 0$, for all temperatures $T$, and so the mean coherence $\chi_m$ is a purely quantum-mechanical feature that disappears in the classical limit. Perhaps more interestingly, we see that $\chi_m$ vanishes in the $T\rightarrow 0$ limit also.  Therefore, since $\chi_m$ only contributes when we simultaneously have both a non-zero Planck's constant $h$ and non-zero temperature $T$, the mean coherence is a genuinely quantum-thermodynamic property of the system.

\subsection{The macroscopic regime and multipartite entanglement.}

One method of characterising the classical regime is by letting thermal fluctuations dominate quantum fluctuations, as we showed. But one expects a notion `classicality' to also arise in the large system limit, where the state of the system is given in terms of a small set of intrinsic variables. In this section we show that if the system $S$ is composed of many independent, identically distributed (IID) systems $S_1,...,S_n$, all higher order corrections bar the variance vanish in the limit $n \to \infty$, under an appropriate scaling of energy units.

Given a macroscopic regime in which a quantum system with no correlations is well-described by a bounded list of intrinsic parameters, the total quantum state can be written as $\rho_{\rm tot} \approx \rho^{\otimes n}$, for some effective subsystem state $\rho$ that encodes the intrinsic parameters. We can now consider the macroscopic states $|\psi_n\> := |\psi\>^{\otimes n}$ and $|\phi_n\> := |\phi\>^{\otimes n}$ on the reference system, for two pure non-energy eigenstates $|\psi\>$ and $|\phi\>$ and see how the fluctuation relation behaves.

In particular, we consider a state $|\psi \> = \sum_{m = 0}^{d-1} \sqrt{p_m} |m \>$ of a $d$-dimensional system, where $\{|m\> \}$ are the eigenstates of the  Hamiltonian $H_1 = \epsilon \sum_{m=0}^{d-1} m |m \> \< m|$. Looking at $n$ copies of the system,
\begin{equation}
|\psi \>^{\otimes n} = \sum_{j = 0}^{n(d-1)} \sqrt{c_j} | e_j \>_n ,
\end{equation}
where we have decomposed $|\psi\>^{\otimes n}$ into a superposition of energy eigenstates $\{ |e_j\>_n \}$ of $n$ copies of the system \cite{ReferenceFrame}, where $|e_j\>_n$ has total energy $j$ and it is assumed $H_{\rm tot}$ is the sum of the non-interacting individual Hamiltonians. The coefficients $\{ c_j \}$ are then described by multinomial coefficients
\begin{equation}
c_j = { n \choose k_0 ... k_{d-1} } p_0^{k_0} p_1^{k_1} ... p_{d-1}^{k_{d-1}}
\end{equation}
where $j = \sum_i i \, k_i$ and $n = \sum_i k_i$. Furthermore, if the spectrum $\{ p_j \}$ is gapless, with $p_j \neq 0$ for $0 \le j \le d-1$, then in the asymptotic limit the distribution $c_j$ over energies becomes \cite{ReferenceFrame}
\begin{equation}
c_j = \frac{1}{\sqrt{2\pi \sigma_n^2 }} e^{- \frac{(\epsilon j-\mu_n(\psi))^2}{2 \sigma_n^2(\psi)} } + \mathcal{O} \left(n^{-1} \right),
\end{equation}
where $\mu_n (\psi) = \< \psi_n| H_{\rm tot} |\psi_n \> = n \< \psi | H_1 |\psi \> = n \mu$ and $\sigma_n^2(\psi) = n \sigma^2$ is the variance of $H_{\rm tot}$ evaluated on $|\psi_n \>$. 

In the large $n$ limit we know from the central limit theorem that the distributions over energy will tend to a Gaussian in the neighborhood around the mean energy. As we show in the Supplementary Material, the fluctuation theorem takes the following form: for any $\epsilon >0$ there is an $M$ such that for all $n >M$ we have
\begin{align}
\left | \frac{P[\psi_n | \tilde{\phi}_n]}{P[\phi^*_n | \tilde{\psi}^*_n]}   - e^{-n (\beta \Delta f +\beta \Delta \mu - \frac{1}{2}\beta^2\Delta \sigma^2)} \right | \le \epsilon,
\end{align}
however with the proviso that the scaling $\beta = \frac{\beta_0}{\sqrt{n} \lambda}$ is adopted, where $\beta_0$ is a dimensionless unit and $\lambda$ a characteristic energy scale -- which is technically required in order to apply the central limit theorem rigorously, and amounts to a choice of units for temperature for a fix scale $n$. This captures the informal statement that for the IID case the states ``become more like Gaussians'' and
\begin{equation}
 \frac{P[\psi_n | \tilde{\phi}_n]}{P[\phi^*_n | \tilde{\psi}^*_n]}  \stackrel{n \rightarrow \infty}\sim e^{-n (\beta \Delta f +\beta \Delta \mu - \frac{1}{2}\beta^2\Delta \sigma^2)}
\end{equation}

Thus in the macroscopic limit, the central limit theorem can be used to truncate the exponent in the fluctuation relation to the second-order cumulants, with the usual provisos in this statement.

As in the asymptotic IID regime one reduces to only the first two cumulants, we can use the second cumulant to place a bound on many-body entangled systems. Suppose we start with a state $|\phi\>^{\otimes n}$ and perform any unitary $U$ that is block diagonal in the total Hamiltonian $H = \sum_i H_i$ then in general the state $U |\phi\>^{\otimes n}$ will be entangled between the $n$ subsystems. But it is readily seen that $\Lambda(\beta, U\phi_nU^\dagger) = \Lambda(\beta, \phi_n)$ and thus (a) in the large $n$ limit will have a Gaussian profile in energy and (b) will generically be entangled between subsystems. 

When the total system $S$ is a large multi-partite system composed of $n$ spin--1/2 particles one can exploit the fact that the variance $Var(\psi,\sigma_i)$ of a Pauli operator $\sigma_i$ in a pure quantum state $|\psi\>$ is upper bounded when limited to $k$--producible states. A state $|\psi \>$ is called \emph{$k$--producible}  if it can be written in the form:
\begin{equation}
|\psi \> = |\phi_1 \> \otimes |\phi_2 \> \otimes ... \otimes |\phi_m \>,
\end{equation}
where $m \geq n/k$ and each pure state $|\phi_j \>$ involves at most $k$ particles \cite{kProducible} (which captures the notion of multipartite entanglement). It is clear from the definition that 1--producible states are product states with $m = n$, containing no entanglement. A pure state $|\psi \>$ has genuine $k$--partite entanglement if  $|\psi \>$ is $k$--producible but not $(k-1)$--producible.
Now for such a state, the variance of a  many-body Pauli spin operator $\Sigma_i^{(n)} = \sigma_i^{(1)} + \dots + \sigma^{(n)}_i$, where $i=x,y,z$ and the superscript denotes the qubit label, obeys 
\begin{equation}
4 \text{Var} (\psi,\Sigma_i^{(n)}) \le s k^2 + (n-sk)^2
\end{equation}
on the set of $k$--producible states \cite{Multipartite,Multipartite2}. Here $s$ is the integer part of $n/k$ and the bound ranges from $n$ up to $n^2$. 
\\
\indent
By using the above bound we see that the largest possible changes in variance will be $\pm [s k^2 + (n-sk)^2]$. Therefore for a system $S$ composed of $n$ particles, with each subsystem having a Hamiltonian $H_i = \epsilon \sigma_z$ giving a total Hamiltonian $H_S = \epsilon \, \Sigma_z^{(n)}$, 
\begin{equation}
 \frac{P [\psi | \tilde{\phi}_n]}{P[\phi_n^* | \tilde{\psi}^*]}   \leq  e^{\beta (W_{B} - \Delta F )} e^{ \frac{(\beta \epsilon)^2 }{8}(s k^2 + (n-sk)^2) + \O(1/\sqrt{n}) }.
\end{equation}
where $W_B = - [\< H_S \>_\psi - \< H_S \>_{\phi_n}]$ is the change in energy of the bath subsystem.
This holds for any large $n$ IID product state $|\phi_n\>$ and any $k$--producible state $|\psi \>$ with genuine $k$--partite entanglement. While this bound (which turns out to be tight \cite{Multipartite}) grows weaker and weaker as $n$ increases, it is nonetheless striking that one can apply the fluctuation relation to provide non-trivial statements on transitions between quantum states with different multipartite entanglement and it would be of interest to further develop this connection further.

%

\section{Outlook}
We started by extending the core defining characteristics of deterministic work in classical mechanics into quantum mechanics, which lead to the notion of coherent work processes. It turns out that the set of coherent processes is non-trivial for quantum systems, and it would be of interest to extend this analysis to determine, for example, what quantum systems admit infinitely decomposable quantum coherence. Because of Theorem \ref{thm:Decomposable} this can be determined from the large literature on the decomposability of classical random variables.

Another open question is whether a \emph{mixed state} version of coherent work processes leads to interesting physics, or could be related to on-going research in thermodynamics. The asymmetry theory perspective on coherent work processes means that this is expected to have non-trivial structure, and so applying existing asymmetry results within a thermodynamic context would be of interest.

The fluctuation relations and the effective potentials discussed here suggest that a wide range of tools (e.g. from many-body physics \cite{GooldManyBody}) could be brought to bear on coherent fluctuation relations. It is essential to develop greater connection with existing fluctuation relations, but it is conjectured that essentially any of the existing extensions of the classical stochastic Crooks relation could be faithfully represented in this setting.

\section{Acknowledgements}

We would like to thank Johan \r{A}berg, Zo{\"e} Holmes, Hyukjoon Kwon, Myungshik Kim, Jack Clarke and Doug Plato for useful discussions. EHM is funded by the EPSRC Centre for Doctoral Training in Controlled Quantum Dynamics. DJ is supported by the Royal Society and also a University Academic Fellowship.

\bibliographystyle{unsrtnat}
\bibliography{GrandCanPaper}

\clearpage

\appendix

\section{Coherent work processes}\label{thm:uniqueness}
Here we gather the proofs of the statements for coherent work processes. First we show that for auxiliary `weight' systems with non-degenerate spectra, the coherent work output is effectively unique.

\begin{proposition}
Let $|\omega \>_A = \sum_k e^{i\theta_k } \sqrt{p_k} |k \>_A$ be the coherent work output from a coherent work process $|\psi \>_S \stackrel{\omega}{\longrightarrow} | \phi\>_S$ where $\{ |k\>_A \}$ is the eigenbasis of the non-degenerate Hamiltonian $H_A$. Then the distribution $\{ p_k \}$ is uniquely determined, while the phases $\{ \theta_k\}$ are arbitrary.
\end{proposition}
\begin{proof}
Assume $|\omega \>_A$ is the coherent work output of the process $|\psi_S\> \stackrel{\omega}{\longrightarrow} |\phi_S\>$, enacted by the unitary $U$. Define the characteristic function $\varphi_{X,\eta} (t) := \tr ( e^{i H_X t} |\eta \> \< \eta |_X)$. Then under a coherent work process,
\begin{align}
\varphi_{S,\psi} (t) \varphi_{A, 0 } (t) &= \tr (e^{i (H_S + H_A) t} |\psi\> \< \psi |_S \otimes |0 \> \< 0|_A)  \\
    &= \tr (e^{i (H_S + H_A) t} U^\dagger U |\psi\> \< \psi |_S \otimes |0 \> \< 0 |_A ) \\
    &= \varphi_{S,\phi} (t) \varphi_{A,\omega} (t),
\end{align}
where $U$ is a unitary and $[U,H_S + H_A] = 0$ by the construction of coherent work processes. Due to the uniqueness property of characteristic functions \cite{Lukacs}, $\varphi_{A,\omega} (t) = \varphi_{A,\omega'} (t)$ for all $t \in \mathbb{R}$ if and only if the distribution for $\omega$ and $\omega'$ are equal. Therefore with $\psi$ and $\phi$ pre-determined, the distribution of $\omega$ is uniquely determined. 

Since the only constraint on $U$ is that it is energy conserving, then for any Hermitian operator $L_A$ such that $[L_A,H_A] = 0$, we can define a unitary $V = (\mathds{1}_S \otimes e^{-i L_A \theta} ) U$ that also conserves energy, with $[V,H_S + H_A] = 0$. Furthermore, with appropriate choice of $L_A$, any set of phases $\{\theta_k \}$ can be selected.
\end{proof}

It is necessary to constrain the spectrum of $A$ to be non-degenerate, as given a degenerate Hamiltonian $H_{A'}$ for a system $A'$, one could enact a unitary $T$ that locally conserves energy ($[H_{A'},T]=0$) and can arbitrarily change the state $|\omega\>_{A}$ in the degenerate subspace.

Though we describe the coherent work output as effectively unique, one is always free to enact evolution locally on the auxiliary system that will induce a set of relative phases. In this way, any phases gained from the joint evolution between $S$ and $A$ can be arbitrarily changed with local evolution. Since relative phases are a signature of time-evolution, the result distinguishes between global and local effects of evolution on states.

We can also prove that all reversible transformations involve trivial coherent work processes.

\begin{proposition}
A  coherent work process is trivial if and only if it is reversible. 
\end{proposition}
\begin{proof}
Let $V|\psi_0 \>_S \otimes |E_0\>_A = |\psi_1\>_{S'} \otimes |\omega\>_{A}$ be a coherent work process. In terms of random variables with respect to Hamiltonians $H_S$ and $H_A$, an equivalent condition for the existence of this process is $\hat{S} = \hat{S}' + \hat{A}$, where $\hat{S}, \,  \hat{S'}, \, \hat{A}$ are the random variables for $|\psi_0\>_S, \, |\psi_1\>_{S'}, \, |\omega\>_{A}$ respectively. If it is a trivial coherent work process, then at least one of $\hat{S}'$ or $\hat{A}$ is a constant random variable. Without loss of generality, we assume $\hat{A} = a$ is constant and so we must have that $\hat{S}' = \hat{S} - a$ and $\hat{S} = \hat{S'} + a$. This gives $|\psi_0 \>_S \stackrel{a}\longrightarrow |\psi_1\>_{S'}$ and $|\psi_1\>_{S'} \stackrel{-a}{\longrightarrow} |\psi_0 \>_S$.

Conversely, assume we have a reversible coherent work process \begin{align}
V|\psi_0 \>_S \otimes |E_0 \>_{A_1} &= |\psi_1\>_{S'} \otimes |\omega_1 \>_{A_1}, \\  
U|\psi_1 \>_{S'} \otimes |E_0' \>_{A_2} &= |\psi_0\>_{S} \otimes |\omega_2 \>_{A_2},
\end{align}
where we introduce a different auxiliary system for each process.
 Then $\hat{S} = \hat{S}' + \hat{A}_1$ and $\hat{S}' = \hat{S} + \hat{A}_2$. As the random variables are all independent, we have both $ \text{Var}(\hat{S}) = \text{Var}(\hat{S}') + \text{Var}(\hat{A}_1)$ and $ \text{Var}(\hat{S}') = \text{Var}(\hat{S}) + \text{Var}(\hat{A}_2)$. Therefore
 \begin{equation}
\mbox{Var}(\hat{A}_1) +\mbox{Var}(\hat{A}_1)  = 0.
\end{equation}
 Due to the non-negativity of the variance for real-valued random variables this implies that $\text{Var}(\hat{A}_1) = \text{Var}(\hat{A}_2) = 0$. This implies that both $\hat{A}_1$ $\hat{A}_2$ are constant random variables and so both processes are trivial as required.
\end{proof}

The consequence of this result is that any reversible coherent work process can simply be labeled by the constant value of $\hat{A}_1$. For example, if we have $|\psi_0\>_S \stackrel{\omega}\longrightarrow |\psi_1\>_{S'}$ and $|\psi_1\>_{S'} \stackrel{\omega'}\longrightarrow |\psi_0\>_S$ then this implies $|\omega\> = |E_k\>$ and $|\omega'\> = |-E_k\>$ for some energy $E_k$ (if $E_0 = 0$), then we label the reversible process simply as $|\psi_0\>_S \stackrel{E_k} \longleftrightarrow |\psi_1 \>_{S'}$.

\subsection{Decomposition of Poisson distributions}\label{Poisson-decomposition}

The main theorem used in this section is Theorem \ref{Raikov's Theorem} due to Raikov and later generalized by others, but we first re-cap on some background details \cite{Lukacs}.

The cumulative distribution function (CDF) of a random variable $\hat{X}$ on the real line is defined as $F(x):= Prob[ \hat{X} \le x]$. If $\hat{X}$ has CDF $F_1$ and $\hat{Y}$ has CDF $F_2$, and $\hat{X}$ and $\hat{Y}$ are independent, then the CDF of the random variable $\hat{Z} := \hat{X} + \hat{Y}$ has CDF given by the convolution
\begin{align}
F_{\hat{Z}}(x)  &= \int_{-\infty}^{\infty}F_1(x-y) dF_2(y) \nonumber \\ 
&= \int_{-\infty}^{\infty}F_1(x-y) \frac{dF_2}{dy}(y) dy\nonumber \\
&= \int_{-\infty}^{\infty}F_1(x-y) f_2(y) dy,
\end{align}
where $f_2(y)$ is the \emph{probability density} for $\hat{Y}$. For completeness of notation, we also define the step function $\epsilon(x) =1$ if $ x\ge 0$ and zero otherwise.

With this notation in place, we can now state the theorem by Raikov.

\begin{theorem}[Raikov, 1938]\label{Raikov's Theorem} The Poisson law has cumulative distribution function
\begin{equation}
F \left(\frac{x-\alpha}{\sigma}, \lambda \right) , \sigma >0, \alpha \in \mathbb{R}, \lambda \ge 0,
\end{equation}
where
\begin{equation}
F(x; \lambda) = e^{-\lambda} \sum_{ k \le x} \frac{\lambda^k}{k!}.
\end{equation}
(for $x<0$ the sum is empty and $F(x;\lambda)=0$). $F(x;\lambda)$ corresponds to the Poisson cumulative distribution for a random variable that takes only integer values at non-negative values $k$ with probability $\frac{\lambda^k}{k!}e^{-\lambda}$. When $\lambda=0$ we obtain the singular distribution law $\epsilon (x - \alpha)$.

The integral equation
\begin{equation}
\int_{-\infty}^\infty F_1(x-y) dF_2(y) = F \left(\frac{x-\alpha}{\sigma}, \lambda \right),
\end{equation}
where $F_1(x)$ and $F_2(x)$ are cumulative distribution functions only has solutions 
\begin{align}
F_1(x) &= F \left(\frac{x-\alpha_1}{\sigma}, \mu \right)\\
F_2(x) &= F \left(\frac{x-\alpha_2}{\sigma}, \nu \right),
\end{align}
where $\mu \ge 0$, $\nu \ge 0$, $\mu + \nu = \lambda$ and $\alpha_1 + \alpha_2 = \alpha$, with $\alpha_1,\alpha_2 \in \mathbb{R}$. 
\end{theorem}
We refer the reader to \cite{PoissonDistribution} for the proof. 

We note that the original statement of Raikov's theorem found in \cite{PoissonDistribution} has a typographical error to do with the shifting parameters $\alpha$ and $\beta$. Namely, $F_1(x) = F(\frac{x - \alpha - \beta}{\sigma},\mu)$ and $F_2(x) = F(\frac{x - \alpha + \beta}{\sigma},\nu)$, with $\beta \in \mathbb{R}$. This error can be seen in the following manner. Suppose we have a Poisson distributed random variable $\hat{Y}$ shifted by $a \in \mathbb{R}$, such that $\hat{Y}' = (\hat{Y} -a) \sim \text{Pois}(\mu)$, where $\text{Pois}(\mu)$ denotes a Poisson distribution on the non-negative integers. Similarly, let us define another Poisson distributed random variable $\hat{Z}$ shifted by $b \in \mathbb{R}$, such that $\hat{Z}' = (\hat{Z} - b) \sim \text{Pois}(\nu)$. If we assume $\hat{Y}$ and $\hat{Z}$ are independent, then the sum of the random variables is
\begin{align}
\hat{X} &= \hat{Y} + \hat{Z} \\
&= (\hat{Y} - a) + (\hat{Z} - b) + (a+b).
\end{align}
Since $a$ and $b$ are constants, we are able to re-arrange this to find
\begin{equation}
\hat{X} - (a+b) = (\hat{Y}-a) + (\hat{Z} - b).
\end{equation}
However, the right hand side of this equation corresponds to the sum of two  unshifted Poisson distributed random variables and thus $\hat{X} - (a+b) = \hat{Y}' + \hat{Z}' \sim \text{Pois}(\mu + \nu)$, where we deduce this from Raikov's theorem. From this we see that $\hat{X}$ is a Poisson distributed random variable shifted by an amount $(a + b)$. Now let us choose $a = \alpha + \beta$ and $b = \alpha - \beta$. Thus $a + b = 2\alpha$ and  we expect the sum of the shifted random variables to have a shifting of $2\alpha$, rather than the $\alpha$ that is claimed.

Another consistency check is to consider the limit of a Poisson cumulative distribution function. We have that
\begin{align}
\lim_{\lambda \to 0} F \left(\frac{x-\alpha}{\sigma};\lambda \right) &= \lim_{\lambda \to 0} \sum_{k \leq \frac{x-\alpha}{\sigma}} e^{-\lambda} \frac{\lambda^k}{k!} \\
&= \begin{cases}
1,& \text{if } x \geq \alpha \\
0,& \text{otherwise}
\end{cases}
\end{align}
where we make use of the limit of a Poisson probability mass function \cite{LimitPoisson}. Now suppose $\hat{Y}' \sim \text{Pois}(\mu)$ and $\hat{Y} = \hat{Y}' + a$. Then $\lim_{\mu \to 0} \hat{Y} = \lim_{\mu \to 0} \hat{Y}' + a = a$. If we similarly define $\hat{Z}' \sim \text{Pois}(\nu)$, with $\hat{Z} = \hat{Z}' + b$, then $\lim_{\nu \to 0} \hat{Z} = b$. Thus when we sum the two shifted random variables, $\lim_{\mu,\nu \to 0} \hat{Y} + \hat{Z} = a+b$, as expected.

\subsection{Proof for decomposability theorem}\label{Proof Decomposable}

\begin{theorem}
Given a quantum system $X$ with Hamiltonian $H_X$ with a discrete spectrum, a state $|\psi\>_X$ admits a non-trivial coherent work process if and only if the associated classical random variable $\hat{X}$ is a decomposable random variable. Furthermore, for a coherent work process 
\begin{equation}
|\psi \>_X \stackrel{\omega}{\longrightarrow} |\phi \>_Z,
\end{equation}
the associated classical random variables are given by $\hat{X} =  \hat{Z} + \hat{W}$, where $\hat{Z}$ and $\hat{W}$ correspond to  the measurements of $H_Z$ and $H_{W}$ in $|\phi \>_Z$ and $|\omega \>_W$ respectively.
\end{theorem}

\begin{proof}
Given quantum systems $X,W,Y,Z$ with corresponding Hamiltonians, suppose first the state $|\psi\>_X$ has decomposable coherence. This implies that there exists an auxiliary system $Y$ with zero energy eigenstate $|0\>_Y$ and an isometry $V$ from $\H_X \otimes \H_Y$ into $\H_Z \otimes \H_W$ with
\begin{equation}\label{conserve}
V (H_X\otimes \I_Y +\I_X\otimes H_Y) = (H_Z \otimes \I_W + \I_Z \otimes H_W)V^\dagger,
\end{equation}
such that 
\begin{equation}
V [ |\psi\>_X \otimes |0\>_Y] = |\phi\>_Z\otimes |\omega\>_W,
\end{equation}
and where neither $|\phi\>_Z$ nor $|\omega\>_W$ are energy eigenstates of their respective Hamiltonians. The associated classical random variables for the input systems are $\hat{X}$ and $\hat{Y}$, while the associated classical random variables on the output systems are $\hat{Z}$ and $\hat{W}$.

Now since $\hat{Y}=0$ the distribution for $\hat{X}$ coincides with the distribution $\hat{X} + \hat{Y}$, which is simply the total energy of the input systems. However equation (\ref{conserve}) implies that $V$ is block diagonal in the total energy and thus preserves the measurement statistics for the total energy. Therefore $ \hat{X}=\hat{Z} + \hat{W}$. Since $|\phi\>_Z$ and $|\omega\>_W$ are not energy eigenstates this means that both $\hat{Z}$ and $\hat{W}$ have non-trivial distributions and are independent since the composite state on the output is a product state, and thus $\hat{X}$ is a decomposable random variable. 

Conversely, given an initial state $|\psi\>_X$ with associated classical random variable $\hat{X}$, now suppose that $\hat{X}$ is decomposable. This means there exist classical random variables $\hat{Z}$ and $\hat{W}$, both with independent, non-trivial distributions, such that $\hat{X} = \hat{Z}+\hat{W}$. Since $H_X$ is assumed to have a discrete spectrum the classical random variable $\hat{X}$ only has support on a discrete set of values, and thus $\hat{Z}$ and $\hat{W}$ only take on discrete values too. Let $E_1, E_2, \dots$ denote the values that the random variables $\hat{X}, \hat{Z}, \hat{W}$ take on, and define the set $\Omega = \{E_1, \dots, E_d, \dots\}$. Denote by $(r_k)$ and $(s_k)$ respectively the distributions of $\hat{Z}$ and $\hat{W}$ over $\Omega$.

Given these distributions, we first introduce two quantum systems  $Z,W$, each of countable dimension $|\Omega|$ with Hamiltonians $ H_Z, H_W$ and each with an energy spectrum lying in $\Omega$. Define the states
\begin{align}
|\phi\>_Z &:= \sum_{k=1}^{|\Omega|} \sqrt{r_k} |E_k\>_Z \\
|\omega\>_W &:=  \sum_{k=1}^{|\Omega|} \sqrt{s_k} |E_k\>_W,
\end{align}
on quantum systems $Z$ and $W$ where we label the eigenstates of each Hamiltonian with their corresponding eigenvalues. We also introduce a third system auxiliary $Y$ of arbitrary dimension, but with a zero energy eigenstate $|0\>_Y$ for some Hamiltonian $H_Y$. We now construct an isometry $V_{\mbox{\tiny tot}}$ from $\H_X \otimes \H_Y$ to $\H_Z \otimes \H_W$ that sends $|\psi\>_X \otimes |0\>_Y$ to $|\phi\>_Z \otimes |\omega\>_W$. 

First observe that
\begin{equation}
|\phi\>_Z \otimes |\omega\>_W = \sum_{z,w = 1}^{|\Omega|} \sqrt{r_z s_w}  |E_z\>_Z \otimes |E_w\>_W.
\end{equation}
However $\hat{X} = \hat{Z} + \hat{W}$, with $\hat{Z}$ and $\hat{W}$ independent random variables. This implies that the distribution $(p_x)$ of $\hat{X}$ over $\Omega$ is given by the convolution

\begin{equation}\label{gen convolution}
p_x = \sum_{\substack{z,w: \\ E_x = E_z + E_w }} q_z r_w.
\end{equation}
Therefore we instead write the above expression for the output states as
\begin{align}
|\phi\>_Z \otimes |\omega\>_W &= \sum_{z,w = 1}^{|\Omega|} \sqrt{r_z s_w}  |E_z\>_Z \otimes |E_w\>_W \nonumber \\
&= \sum_{E_x} \hspace{-0.2cm} \sum_{\substack{z,w: \\ E_x = E_z + E_w }} \sqrt{r_z s_w}  |E_z\>_Z \otimes |E_w\>_W.
\end{align}
We let $\Pi_x$ denote the projector onto the $E_x$ eigenspace of $\H_X$, and so for those $x$ for which $p_x \ne 0$ we have
\begin{equation}
\Pi_x |\psi\>_X = \sqrt{p_x} |\psi_x\>_X,
\end{equation}
which defines a normalised pure state $|\psi_x\>_X$ within the $E_x$ eigenspace of $\H_X$. Thus 
\begin{equation}
|\psi\>_X = \sum_{x=1}^{|\Omega|} \sqrt{p_x} |\psi_x\>_X.
\end{equation}

We next define a linear mapping $V$ from $\H_{\mbox{\tiny sup}}:=\mbox{span}\{|\psi_x\>_X \otimes |0\>_Y\}_x$ into $\H_Z\otimes \H_W$ for which
\begin{equation}\label{initial V}
V \left [ |\psi_x\>_X \otimes |0\>_Y \right ]= \frac{1}{ \sqrt{p_x}}\hspace{-0.3cm} \sum_{\substack{z,w: \\ E_x = E_z + E_w }} \sqrt{r_z s_w}  |E_z\>_Z \otimes |E_w\>_W,
\end{equation}
for all $x =1, \dots |\Omega|$ for which $p_x \ne 0$. If $p_x \ne 0$ for some $x$ then equation (\ref{gen convolution}) implies that there is at least one pair $(z,w)$ such that $q_z$ and $r_w$ are non-zero and so for such an $x$ the image under $V$ is a non-zero vector.

\emph{V is an isometry}: It is readily checked, using the orthogonality of the eigenstates of the output systems and equation (\ref{gen convolution}) that
\begin{equation}
\<\psi_x |_X\otimes \<0|_Y V^\dagger V |\psi_x\>_X \otimes |0\>_Y = 1
\end{equation}
for all $x$ for which $p_x \ne 0$. Note that since different energy eigenspaces are orthogonal, we have both that the input states $\{|\psi_x\>_X \otimes |0\>_Y\}_x$ are an orthonormal basis for $\H_{\mbox{\tiny sup}}$, and also the output states are orthonormal
\begin{equation}
\<\psi_x |_X\otimes \<0|_Y V^\dagger V |\psi_y\>_X \otimes |0\>_Y =0,
\end{equation}
for any $x \ne y$.
The above provide the matrix components for $V$ and imply that $V$ is an isometry from $\H_{\mbox{\tiny sup}}$ into $\H_Z \otimes \H_W$. 

\emph{V is energy conserving:} Let $H_{XY}:= H_X \otimes \I_Y + \I_X \otimes H_Y$ be the initial Hamiltonian. Now since $\H_X \otimes \H_Y = \H_{\mbox{\tiny sup}} \oplus \H_{\mbox{\tiny sup}}^\perp$ and since $|\psi_x\>_X\otimes |0\>_Y$ are eigenstates of $H_{XY}$  we can write $H_{XY} = H_{XY,s} \oplus H_{XY,p}$, where $H_{XY,s}$ is simply the restriction of the total Hamiltonian to the subspace $\H_{\mbox{\tiny sup}}$ and $H_{XY,p}$ is the component on $\H_{\mbox{\tiny sup}}^\perp$. Now each pair $(z,w)$ in equation (\ref{initial V}) contributes a vector $\sqrt{r_z}|E_z\>_Z \otimes \sqrt{s_w}|E_w\>_W$ that is an $E_x$ energy eigenvector of the total Hamiltonian $H_{ZW}:=H_Z \otimes \I_W + \I_Z \otimes H_W$, and so $V|\psi_x\>_X \otimes |0\>_Y$ is an $E_x$ energy eigenstate of $H_{ZW}$ for any $x$. Thus $V$ sends $E_x$ eigenstates of $H_{XY,s}$ to $E_x$ eigenstates of $H_{ZW}$ as required by energy conservation.

Finally, by rearranging equation (\ref{initial V}) and summing over $x$ we see that $V[| \psi\>_X\otimes |0\>_Y] = |\phi\>_Z \otimes |\omega\>_W$. To get an isometry on the full space, we again write $\H_X \otimes \H_Y = \H_{\mbox{\tiny sup}} \oplus \H_{\mbox{\tiny sup}}^\perp$, and extend the systems $Z$ and $W$ and their Hamiltonians so that now 
\begin{equation}
\H_Z \otimes \H_W = \mbox{span}\{|E_z\>_Z \otimes |E_w\>_W\}_{z,w} \oplus \H_{\mbox{\tiny other}}
\end{equation}
with $\mbox{dim}( \H_{\mbox{\tiny other}}) \ge \mbox{dim} (\H_{\mbox{\tiny sup}}^\perp)$. The final isometry on the input system is then $V_{\mbox{\tiny tot}} = V \oplus V_{\mbox{\tiny other}}$, where $V_{\mbox{\tiny other}}$ is an arbitrary energy conserving isometry from $ \H_{\mbox{\tiny sup}}^\perp$ into $ \H_{\mbox{\tiny other}}$ . This constructs a globally energy conserving $V_{\mbox{\tiny tot}}$ such that
\begin{equation}
V_{\mbox{\tiny tot}} |\psi\>_X \otimes |0\>_Y = |\phi\>_Z \otimes |\omega\>_W,
\end{equation}
and since both of $\hat{Z}$ and $\hat{W}$ are non-trivial random variables we have that neither of the output states $|\omega\>_W, |\phi\>_Z $ is an energy eigenstate of its respective Hamiltonian. Thus the state $|\psi\>_X$ has decomposable coherence as required.

\end{proof}

\subsection{Proof of semi-classical coherent work processes theorem}\label{Proof Semi Classical}

The proof of the semi-classical theorem relies on Raikov's theorem, which is a result for classical probability distributions. The core idea behind the proof is outlined in Fig. \ref{fig: Raikov Diagram}. For the non-trivial direction, one starts with a tuple $(S,|\psi\>_S,H_S)$ with $S$ a quantum system with Hamiltonian $H_S$, and $|\psi\>_S$ a state of the system that admits a non-trivial coherent work process. Any state $|\psi\>_S$ uniquely determines a classical random variable $\hat{S}$, although the converse is not true since we have freedom in, for example, choosing phases or the dimension of $S$. If $\hat{S}$ is decomposable, then Raikov's theorem tells us the relationship between the components of the decomposition $\hat{S} = \hat{S}' + \hat{A}$. From this we can construct (non-unique) tuples $(S',|\phi\>_{S'},H_{S'})$ and $(A,|\omega\>_A,H_A)$ associated to the random variables $\hat{S}'$ and $\hat{A}$ respectively. On these quantum systems we then show that a coherent work process can be constructed. The precise statement and proof are as follows.

\begin{figure}[h]
\begin{center}
\includegraphics[width=6cm]{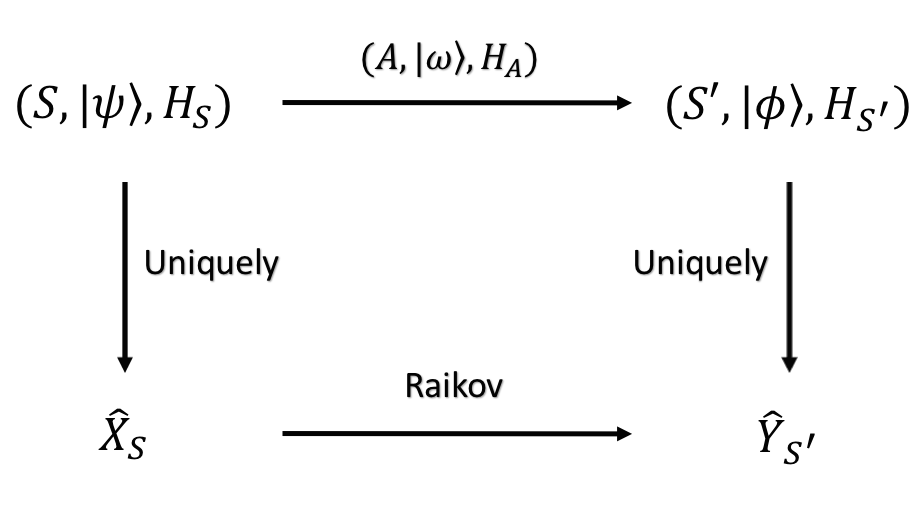}
\end{center}
\caption{\textbf{Raikov's theorem and coherent work processes.} Given any quantum system $S$ with Hamiltonian $H_S$ in a state $|\psi\>_X$ there is a unique classical random variable $\hat{X}_S$ corresponding to the energy measurement on $S$. Conversely, given a random variable $\hat{X}_S$ we may construct a (non-unique) tuple $(S,|\psi\>_S, H_S)$ associated to this random variable.
  \label{fig: Raikov Diagram}
}
\end{figure}

\begin{theorem} Let $S$ be a harmonic oscillator system, with Hamiltonian $H_S = h \nu a^\dagger a$. Let $\mathcal{C}$ be the set of quantum states for $S$ as defined above. Then:
\begin{enumerate}
\item The set of quantum states $\mathcal{C}$ is closed under all possible coherent work processes from $S$ to $S$ with fixed Hamiltonian $H_S$ for both the input and output.
\item Given any quantum state $|\psi\>_S\in \mathcal{C}$ there is a unique, canonical state $|\alpha, k\>_S \in \mathcal{C}$ such that we have a reversible transformation
\begin{equation}
|\psi\>_S \stackrel{\omega_c}{\longleftrightarrow} |\alpha,k\>_S,
\end{equation}
with $|\omega_c\>_A = |0\>_A$ and $\alpha = |\alpha|$.
Moreover, the only coherent work processes possible between canonical states are 
\begin{equation}
|\alpha, k\>_S \stackrel{\omega}{\longrightarrow} |\alpha',k\>_S,
\end{equation}
 such that $|\alpha'| \le |\alpha|$. Modulo phases, the coherent work output is of the form $|\omega\>_A =  |\lambda, n\>_A$, where $\lambda = \sqrt{|\alpha|^2 - |\alpha'|^2}$ and $n,k'$ are any integers that obey $n+k' = k$.
\item In the classical limit of large displacements $|\alpha| \gg 1$, from coherent processes on $\mathcal{C}$ we recover all classical work processes on the system $S$ under a conservative force.
\end{enumerate}
\end{theorem}
\begin{proof}
 
(Proof of 1.) Let $|\psi \>_S = e^{i L_S \theta} |\lambda ,m\>_{S} $ be a state in $\mathcal{C}$ defined with respect to a Hamiltonian $H_S = h \nu a^\dagger a$, with $[L_S,H_S ]=0$, $L_S = L_S^\dagger$ and $\theta \in \mathbb{R}$. Denote by $\hat{S}$ the classical random variable obtained from an energy measurement of $S$, which takes values in $\{ s_n = n h \nu :  n \in \mathbb{Z}\} $, and with distribution $\p = (p_n = |\< \psi | n \>|^2 )$. We can expand the state $|\psi \>_S$ in the energy basis as
\begin{equation}
|\psi \> = e^{-\lambda/2} \sum_{n=0}^\infty e^{i\theta_n} \sqrt{\frac{\lambda^n}{n!}}  |n +m \>,
\end{equation}
for some arbitrary phases $\{\theta_n \}$ and some $m \in \mathbb{N}$. Denote by $\sigma = h \nu$ the energy spacing of $H_S$. Then the energy measurement distribution is Poissonian and given by 
\begin{equation}
p_n=P \left(\hat{S} = (n+m)h \nu \right) = e^{-\lambda} \frac{\lambda^n}{n!},
\end{equation}
for $n \ge 0$ and $p_n =0$ for $n <0$. The associated cumulative distribution function is given by
\begin{equation}
F_{\hat{S}}(x) =F \left( \frac{x - m h \nu}{\sigma} , \lambda \right) .
\end{equation}

The set $\mathcal{C}$ is clearly closed under trivial coherent work processes. Now suppose that $|\psi\>_S$ admits some non-trivial coherent work process given by
\begin{equation}
|\psi\>_S \stackrel{\omega}{\longrightarrow} |\phi\>_{S'}.
\end{equation}
By Theorem \ref{thm:Decomposable} this implies that $\hat{S} = \hat{S}' + \hat{A}$, where $\hat{S}'$ and $\hat{A}$ are non-trivial, independent classical random variables. From these we can construct tuples $(S', |\phi\>_{S'}. H_{S'})$ and $(A, |\omega\>_A, H_A)$ that realise the associated coherent work process, as detailed in the proof of Theorem \ref{thm:Decomposable}. For any $|\psi\>_S \in \mathcal{C}$ we have that $\hat{S}=\hat{S}' + \hat{A}$ is a decomposable Poisson random variable, and thus by Raikov's theorem the cumulative distribution functions for $\hat{S}'$ and $\hat{A}$ are also Poissonian, with
\begin{align}\label{Poisson eqns}
F_{\hat{S}'}(x) &= F \left( \frac{x - \alpha_1}{\sigma} , \mu \right), \\
F_{\hat{A}}(x) &= F \left( \frac{x - \alpha_2}{\sigma} , \kappa \right),
\end{align}
where $\nu, \kappa \ge 0$ and constrained as $\mu + \kappa = \lambda$, $\alpha_1 + \alpha_2 = m h \nu$ and $\alpha_1,\alpha_2 \in \mathbb{R}$.

As $\sigma$ appears identically in all cumulative distribution functions this implies that both $H_{S'}$ and $H_A$ must have within their spectral decompositions, countably infinite discrete energy eigenstates with equal spacings $h \nu$, corresponding to the outcomes of the energy measurements for the constructed states $|\phi\>_{S'}$ and $|\omega\>_A$. We write the corresponding distributions for $\hat{S}'$ and $\hat{A}$ as
\begin{align}
q_n := P [\hat{S}' = n h \nu + c_{S'}] &=  e^{-\mu}\frac{\mu^n}{n!} \nonumber \\
r_a := P [\hat{A} = a h \nu + c_A] &= e^{-\kappa}\frac{\kappa^a}{a!},
\end{align}
for integers $n,a \ge 0$, and where $c_{S'}$ and $c_A$ are any constants that obey
\begin{equation}
c_{S'} + c_A = m h \nu.
\end{equation}
By suitable constant shifts of the Hamiltonians $H_{S'}$ and $H_A$, that shift the zero energy level, we can ensure that $c_{S'} = m' h \nu$ for some integer $m'$, and thus $c_A = (m-m')h\nu$. This implies that both $(q_n)$ and $(r_a)$ are both Poissonian distributions on a discrete set of energy eigenvalues with constant separation $h \nu$.

Let $\Pi_n^{S'}$ and $\Pi^A_a$ denote the projectors for the discrete energy eigenvalues in the energy measurements of $S'$ and $A$ with probabilities $q_n$ and $r_a$ on $|\phi\>_{S'}$ and $|\omega\>_A$ respectively. We now define the states $|n+m'\>_{S'}$ and $|a+m-m'\>_A$ via
\begin{align}
\sqrt{q_n} |n+m'\>_{S'}& := \Pi_n^{S'} |\phi\>_{S'} \\
\sqrt{r_a} |a+ m-m'\>_A & := \Pi_a^A |\omega\>_A ,
\end{align}
with $n, a \ge 0$ and where we have ensured trivial phases via the definition of the states.
Therefore we can write the two output states as
\begin{align}
|\phi\>_{S'} &= e^{-\mu/2}\sum_{n=0}^\infty \sqrt{\frac{\mu^n}{n!}} |n+m'\>_{S'} \nonumber \\
|\omega\>_A &=  e^{-\kappa/2}\sum_{a=0}^\infty \sqrt{\frac{\kappa^a}{a!}} |a+m-m'\>_A.
\end{align}
Recall that the translated coherent states are given by
\begin{align}
D(\alpha) |0\> = |\alpha\> &= e^{-\frac{|\alpha|^2}{2}} \sum_{n=0}^\infty \frac{\alpha^n}{n!}|n\> \\
 |\alpha, m\> &= e^{-\frac{|\alpha|^2}{2}} \sum_{n=0}^\infty \frac{\alpha^n}{n!}|n+m\>.
\end{align}
and so we can define real constants $\beta_1 = \sqrt{\mu}$ and $\beta_2 = \sqrt{\kappa}$ that allow us to write $|\phi\>_{S'} = |\beta_1, m'\>_{S'}$ and $|\omega\>_A = |\beta_2, m-m'\>_A$, in terms of the notation for $\mathcal{C}$ states.
Therefore $|\phi \>_{S'} , |\omega \>_A \in \mathcal{C}$ and so $\mathcal{C}$ is closed under coherent work processes.  

(Proof of 2.) Any state in $\mathcal{C}$ can be expressed in the form \begin{equation}
|\psi \>_S = e^{-\lambda/2} \sum_{n=0}^\infty e^{i\theta_n} \sqrt{\frac{\lambda^n}{n!}} |n+m\>,
\end{equation} 
where $m$ is some non-negative integer, $\lambda \geq 0$ and $\{\theta_n \}$ are phases.
We now define a unitary of the form $V = V_S \otimes \mathds{1}_A$ where 
\begin{equation}
V_S = \sum_{n = 0}^\infty e^{-i\theta_n} |n+m\>\<n+m|  + \sum_{k = 0}^{m-1} |k\>\< k|,
\end{equation}
which realises the transformation $V[|\psi \>_S \otimes |0\>_A ] = |\sqrt{\lambda}, m\>_S \otimes |0\>_A$ with $\lambda \ge 0$. It is clear that both $V$ and its inverse $V^\dagger$ are energy conserving and so we have a reversible coherent work process with $|\omega_c\>_A = |0\>_A$ as claimed.

As constructed in the proof of part 1, we saw that modulo arbitrary phases between the energy eigenstates, any coherent work process on a state $|\alpha,m\>_S$ takes the form
\begin{equation}
|\alpha , m\>_S \stackrel{\omega}\longrightarrow  |\beta_1 , m' \>_{S'} 
\end{equation}
and with $|\omega\>_A = |\beta_2, m - m' \>_A$. With a freedom to shift the energies of $S'$ and $A$ by a constant factor $\pm c$ respectively. Moreover, since the Poisson parameters are given, as above, by $\lambda = |\alpha|^2$, $\mu= |\beta_1|^2$ and $\kappa = |\beta_2|^2$, Raikov's theorem tells us that we must have $|\alpha|^2 = |\beta_1|^2 + |\beta_2|^2$, and therefore $|\beta_2| = \sqrt{|\alpha|^2 - |\beta_1|^2}$ as claimed.

(Proof of 3.) For a classical work process under a conservative force, we need two conditions to hold: (i) the position, momentum and energy are sharp and (ii) the process depends only on the initial and final energies, not the path. For the first, consider the position and momentum operators $q = \sqrt{\frac{\hbar}{2m\omega}} (a^\dagger + a)$ and $p = i \sqrt{\frac{\hbar m \omega}{2}}(a^\dagger - a)$ with the canonical commutation relation $[q,p]= i\hbar$, and $\hbar \omega = h \nu$. It is straightforward to verify that the variances satisfy $\sigma_q^2 =  \frac{\hbar}{2 m \omega}$ and $\sigma_p^2 = \frac{\hbar m \omega}{2}$ for any coherent state. 
Then coherent states are minimal uncertainty states with Heisenberg uncertainty relation $\sigma_q^2 \sigma_p^2 = \frac{\hbar^2}{4}$. For a coherent state $|\alpha \>$, the expectation value for the position and momentum are $\< q \>_\alpha =\sqrt{ \frac{\hbar}{2m\omega}}(\alpha^* + \alpha)$ and $\< p \>_\alpha = i \sqrt{\frac{\hbar m \omega}{2}} (\alpha^* - \alpha)$. Therefore
\begin{align}
\frac{\sigma_q}{\<q \>_\alpha} &= \frac{1}{(\alpha^* + \alpha)},  \\
\frac{\sigma_p}{\<p \>_\alpha} &= \frac{i}{(\alpha - \alpha^*)}.  
\end{align}
In the limit of large displacement $|\alpha| \gg 1$, the values for $\<q\>_\alpha$ and $\< p \>_\alpha$ become arbitrarily sharp with the ratio of standard deviation to value tending to zero. Similarly, for a Harmonic oscillator Hamiltonian $H_S = \hbar \omega a^\dagger a$, the variance of energy is $\sigma_H^2 = \hbar^2 \omega^2 |\alpha|^2$ leading to the ratio of the standard deviation to the mean 
\begin{equation}
\frac{\sigma_H}{\< H_S \>_\alpha} = \frac{1}{|\alpha|},
\end{equation}
which again becomes sharp in the limit in the same manner. Thus, all energy transfers, positions and momenta become sharp in the limit of large displacements and the classical mechanics of a harmonic oscillator is recovered.

For (ii), the energy transfer depends only on the initial and final states of the system $|\alpha \>_S$ and $|\alpha'\>_{S}$. As all coherent states are included in $\mathcal{C}$, then coherent work processes on $\mathcal{C}$ states with large displacement are classical work processes under a conservative force. Furthermore, for the transition $|\alpha \>_S \stackrel{\omega}{\longrightarrow} |\alpha' \>_S$, we have that $|\alpha | \geq |\alpha'|$ and therefore $\< H_S \>_\alpha \geq \< H_S \>_{\alpha'}$ -- meaning the process on $S$ is work extracting only for the coherence component of $|\alpha,k\>$. However variations in the $k$ label are always sharp and can involve either positive or negative work and it is clear that for large displacements the fractional energy separation of levels tends to zero and one has an effective continuum of energy values that recovers the classical case.

\end{proof}

\subsection{Proof of Theorem II.6 (coherent work processes on a ladder system)}\label{Proof of ladder theorem}

\begin{theorem} 
Let $S$ be a quantum system,  with energy eigenbasis $\{ |n\>: n \in \mathbb{Z}\}$ with fixed Hamiltonian $H_S = \sum_{n \in \mathbb{Z}} n |n\>\<n|$. Given an initial quantum state $|\psi\>_S = \sum_i \sqrt{p_i} |i\>_S$, a coherent work process
\begin{equation}
|\psi\>_S \stackrel{\omega}{\longrightarrow} |\phi\>_S ,
\end{equation}
is possible with
\begin{align}
|\phi \>_S &= \sum_j e^{i\theta_j} \sqrt{q_j} |j\>_S\\
|\omega \>_A &= \sum_k e^{i\varphi_k}\sqrt{r_k} |k\>_A
\end{align}
for arbitrary phases $\{\theta_j\}$ and $\{\varphi_k\}$, with output Hamiltonian $H_S$ as above and $H_A = \sum_{n \in \mathbb{Z}} n |n\>_A\<n|$,
if and only the distribution $(p_n)$ over $\mathbb{Z}$ can be written as
\begin{align}
p_n &= \sum_{j \in \mathbb{Z}} r_{j} (\Delta^{j}\mathbf{q})_n.
\end{align}
for distributions $(q_m)$ and $(r_j)$ over $\mathbb{Z}$. 
\end{theorem}

\begin{proof}
Assume $|\psi \>_S \stackrel{\omega}{\longrightarrow} |\phi \>_S$. Then there exists a unitary $V$ such that $V |\psi \>_S \otimes |0 \>_A = |\phi \>_S \otimes |\omega \>_A$. Furthermore, the unitary satisfies the global energy constraint $[V,H_S + H_A ] = 0$, where $H_S$ and $H_A$ are the Hamiltonians of $S$ and $A$ respectively. 

We expand $|\phi\>_S$ and $|\omega\>_A$ in their energy eigenbases as
\begin{align*}
|\phi \>_S &= \sum_j e^{i\theta_j} \sqrt{q_j} |j\>_S\\
|\omega \>_A &= \sum_k e^{i\varphi_k}\sqrt{r_k} |k\>_A.
\end{align*}
If $V$ provides a valid coherent work process then so too does the unitary $UV$ where $U = e^{i B_S \theta} \otimes e^{i C_A \varphi}$ for any Hermitian $B_S$ and $C_A$ such that $[B_S,H_S ] = [C_A ,H_A] = 0$. By choosing the operators $B_S$ and $C_A$ appropriately one can freely vary the phases $\theta_j$ and $\varphi_k$. We have that $A$ has the same Hamiltonian ladder structure as $S$. 
Energy conservation implies that for any $n \in \mathbb{Z}$, the action of $V$ on the state $|n\>_S \otimes |0\>_A$ gives
\begin{equation}
V |n \>_S \otimes |0 \>_A = \sum_{j \in \mathbb{Z}} e^{i\alpha_{n,j}}\sqrt{m_{j|n}} |n-j \>_S \otimes |j \>_A.
\end{equation}
for some distribution $(m_{j|n})_j$ over $j$, and phases $(\alpha_{n,j})$. This implies that 
\begin{align}
 |\phi\>_S \otimes |\omega\>_A &= \sum_{n,j \in \mathbb{Z}} e^{i\alpha_{n,j}} \sqrt{m_{j|n} p_n} |n-j \>_S \otimes |j \>_A,  \\
      &= \sum_{k,l \in \mathbb{Z}} e^{i (\theta_k + \varphi_l)} \sqrt{q_k r_l} |k \>_S \otimes |l \>_A.
\end{align}
Let $\Pi_n$ be the projector onto the energy eigenspace of total energy $n \in \mathbb{Z}$ for the joint system $SA$. We thus have that
\begin{align}
\< \phi |_S \otimes \< \omega |_A \Pi_n  |\phi \>_S \otimes |\omega \>_A  &= \sum_{j \in \mathbb{Z}} m_{j|n} p_n,  \\
    &= \sum_{j \in \mathbb{Z}} q_{n-j} r_j,
\end{align}
However $\sum_j m_{j|n} =1$ and so we have that the distribution $(p_n)$ decomposes as
\begin{equation}\label{trans-major2}
p_n = \sum_{j \in \mathbb{Z}} r_j q_{n-j} = \sum_{j \in \mathbb{Z}} r_j (\Delta^{j} \mathbf{q})_n,
\end{equation}
as required.
Conversely, suppose we have a distribution $(p_n)$ over $\mathbb{Z}$ which can be decomposed as
\begin{equation}
p_n = \sum_{j \in \mathbb{Z}} r_j q_{n-j} = \sum_{j \in \mathbb{Z}} r_j (\Delta^{j} \mathbf{q})_n,
\end{equation}
for two distributions $(r_j)$ and $(q_k)$. From these we define for ladder systems $S$ and $A$ the states
\begin{equation}
|\psi\>_S = \sum_i \sqrt{p_i}|i\>_S, 
\end{equation}
and
\begin{align}
|\phi \>_S &= \sum_j e^{i\theta_j} \sqrt{q_j} |j\>_S\\
|\omega \>_A &= \sum_k e^{i\varphi_k}\sqrt{r_k} |k\>_A
\end{align}
where $\theta_j$ and $\varphi_k$ are arbitrary phases. The joint system $SA$ therefore has the two states
\begin{align}
|\psi\>_S \otimes |0\>_A &=  \sum_i \sqrt{p_i}|i\>_S\otimes |0\>_A \nonumber \\
 |\phi\>_S \otimes |\omega\>_A &= \sum_{k,l \in \mathbb{Z}} e^{i (\theta_k + \varphi_l)} \sqrt{q_k r_l} |k \>_S \otimes |l \>_A. \nonumber \\
  &= U^\dagger \sum_{k,l \in \mathbb{Z}} \sqrt{q_k r_l} |k \>_S \otimes |l \>_A,
\end{align}
where $U^\dagger$ is a phase unitary that generates the phases $\theta_k$, $\varphi_j$ and which is diagonal in the energy eigenbasis.
Again, we let $\Pi_n$ denote the projector onto the $n$--eigenspace of the system $SA$ and so
\begin{align}
\Pi_n  |\psi\>_S &\otimes |0\>_A = \sqrt{p_n}|i\>_S\otimes |0\>_A \nonumber \\
 \!\!\!\Pi_n U |\phi\>_S \otimes |\omega\>_A &=\!\!\! \sum_{m \in \mathbb{Z}} \sqrt{q_{n-m} r_m} |n-m \>_S |m \>_A. 
\end{align}
Thus we have that
\begin{equation}
\<\phi|_S\otimes \<\omega|_A U^\dagger \Pi_n U |\phi\>_S \otimes |\omega\>_A = \sum_{m} q_{n-m} r_m =p_n,
\end{equation}
and thus the vectors $\Pi_n |\psi\>_S \otimes |0\>_A$ and $\Pi_n U |\phi\>_S \otimes |\omega\>_A$ both lie in the $n$--eigenspace of the total hamiltonian and have equal norms. Therefore, we may define a unitary $V_n$ within each energy eigenspace such that
\begin{equation}
V_n \Pi_n |\psi\>_S \otimes |0\>_A  = \Pi_n U |\phi\>_S \otimes |\omega\>_A.
\end{equation}
For those $n$--eigenspaces outside the support of $|\psi\>_S\otimes |0\>_A$ we simply let $V_n$ be the identity unitary on that subspace. We now write $V := \oplus_n V_n$ to obtain an energy conserving unitary on the full space, and using the fact that $U$ commutes with the energy projectors we have that
\begin{equation}
U^\dagger V |\psi\>_S \otimes |0\>_A = |\phi\>_S \otimes |\omega\>_A,
\end{equation}
where $U^\dagger V$ is an energy conserving unitary on the joint system. Therefore the required coherent work process exists.

\end{proof}

\subsection{Physical example of decomposable coherence}

In the main text, an example of infinitely divisible coherence was provided in the form of a coherent state. Here we provide a physical system in which one might encounter such examples of coherent state manipulations. Consider a beam splitter set up, with a coherent state of light injected as the input on one mode and the vacuum on the other mode, such that the initial state is $|\alpha \>_S \otimes |0 \>_A$. We use the beam splitter unitary 
\begin{equation}
B(\theta) = e^{\theta (ab^\dagger - a^\dagger b)}
\end{equation}
where $a$ and $b$ are the annihilation operators for $S$ and $A$ respectively \cite{BeamSplitter}.
Under the action of a beam splitting,  which conserves photon number, we want to achieve
\begin{align}
B(\theta) |\alpha \> \otimes |0 \> = |\beta \> \otimes |\omega\>,
\end{align}
for some tunable parameter $\theta$. It is well known that the two mode annihilation operators satisfy the algebraic property:
\begin{align}
B^\dagger (\theta) a B(\theta) &= a \cos \theta + b \sin \theta   \\
B^\dagger (\theta ) b B(\theta) &= b \cos \theta - a \sin \theta.
\end{align}
Since coherent states are displaced vacuum states, we have that

\begin{align}
B(\theta) |\alpha \> \otimes |0 \> &= B(\theta) e^{-\frac{1}{2} |\alpha|^2} e^{ \alpha a^\dagger} e^{-\alpha^* a} |0 \> \otimes |0 \>  \\
      &= e^{-\frac{1}{2} |\alpha|^2} B(\theta) e^{ \alpha a^\dagger}  B^\dagger (\theta) |0 \> \otimes |0 \> \\
      &= e^{-\frac{1}{2} |\alpha|^2} e^{ \alpha B(\theta) a^\dagger B^\dagger (\theta)} |0 \> \otimes |0 \> \\
      &= e^{-\frac{1}{2} |\alpha|^2}  e^{\alpha (a^\dagger \cos \theta - b^\dagger \sin \theta)} |0 \> \otimes |0 \> \\
      &= |\alpha \cos \theta \> \otimes |-\alpha \sin \theta \>
\end{align}
where we used the fact that $B(\theta) |0 \> \otimes |0 \> = |0 \> \otimes |0 \>$ and $B(\theta)$ is unitary.
Here, $r = \cos \theta$ and $t =  \sin \theta$ are the transmission and reflection coefficients. The consequence is that we have achieved the desired coherent state splitting with a photon number conserving unitary. The second mode that began in the vacuum therefore carries information of the coherent energy transfer. As expected, the condition $|\alpha |^2 = |\alpha \cos\theta |^2 + |\alpha \sin \theta |^2$ is satisfied.
\\
If one were to perform measurements on the ancilla system, they would extract information on the energetics and coherences. However, by performing the measurement, the state would no longer be available for manipulations that involve interference effects. For this reason, the coherent work output is the state $|\alpha \sin \theta \>$ rather than a post-processed state.

\section{Fluctuation relations and the effective potential}
Here we gather together the proofs for the fluctuation relation analysis.

\subsection{Basic properties of the effective potential.}

It is worth describing some properties of the function $\Lambda(\beta,\rho)$ to shed light on the fluctuation relations.

We first note that the effective potential has an invariance under the group generated by the observable $H$. More precisely, let $U(t)$ be the unitary representation of time-translations generated by $H$. Then $\Lambda(\beta, U(t) \rho U(t)^\dagger) = \Lambda(\beta, \rho)$ for any $t \in \mathbb{R}$, and so we have that $\Lambda(\beta,\rho(t_1)) = \Lambda(\beta,\rho(t_2))$ for any two times $t_1$, and $t_2$. Thus, $\Lambda$ is a constant of motion under the free Hamiltonian evolution. For a pure quantum state $|\psi\> = \sum_k e^{i\theta_k}\sqrt{p_k} |E_k\>$, $\Lambda$ is purely a function of the classical distribution $p_k$ over the energy eigenstates. For a projective measurement of energy in this eigenbasis we can denote by $\hat{X}$ the classical random variable for the energy obtained in the measurement, then it is readily seen that $\Lambda(\beta,\psi)$ is essentially the negative cumulant generating function for $\hat{X}$ in the inverse temperature $\beta$,
\begin{equation}
\Lambda(\beta, \psi) = - \sum_{n \geq 1} \frac{(-1)^n\beta^n}{n!} \kappa_n(\hat{X}).
\end{equation}

We shall make use of the following theorem for our analysis.

\begin{theorem} \label{Theorem:Legendre}
If $A$ is Hermitian and $Y$ is strictly positive, then 
\begin{equation*}
\begin{aligned}
\ln \tr(e^{A + \ln Y}) = \max \{ &\tr(XA) - S(X||Y) \ : \  X \ \text{is positive } \\ &\& \ \tr X = 1 \}.
\end{aligned}
\end{equation*}
Alternatively, if $X$ is positive with $\Tr(X) = 1$ and $B$ is Hermitian, then
\begin{equation*}
S(X||e^{B}) = \max \{ \tr(XA) - \log \tr(e^{A+B}) \ : \ A = A^\dagger \}.
\end{equation*}
\end{theorem}
The proof of this can be found in \cite{Petz1994}.

We now establish basic properties of the effective potential $\Lambda(\beta, \rho)$. As a function of the quantum state $\rho$, the effective potential has the following properties.
\begin{lemma} \label{lemma:EffectivePotential2} Given a quantum system $S$ with Hamiltonian $H$ , and inverse temperature $\beta >0$. We have that
\begin{enumerate}
\item (Invariance) Let $\rho(t) = e^{-iHt} \rho e^{+iHt}$ for all $t \in \mathbb{R}$. For any state $\rho$ and any $t\in \mathbb{R}$ we have $\Lambda(\beta, \rho(t)) = \Lambda (\beta, \rho(0))$. Moreover, we have that $\Lambda(\beta, \rho) = \Lambda(\beta, \D(\rho))$ where $\D(\rho) = \sum_k \Pi_k \rho \Pi_k$ is the dephasing of $\rho$ in the energy basis, with $\Pi_k$ being the projector onto the $k$'th energy eigenspace of $H$.
\item  For any state $\rho$, we have that
\begin{equation}
 \hspace{-0.5cm}\beta \< H \>_\rho-\frac{\beta^2}{8}(E^{\rm max} - E^{\rm min})^2 \le \Lambda(\beta, \rho) \le \beta \<H\>_\rho.
\end{equation}
\item Let  $E^{\rm min}$ and $E^{\rm max}$ be the smallest and largest eigenvalues of $H$ respectively in the support of $\rho$. Then
\begin{equation}
\beta E^{\rm min} \le \Lambda(\beta,\rho) \le \beta E^{\rm max}.
\end{equation}
\item (Variational form) For states $\rho$ of full rank and $[\rho,H]=0$, we have that
\begin{equation}
\Lambda(\beta,\rho) = \min_\sigma \left [ \beta \<H\>_\sigma + S(\sigma||\rho) \right ],
\end{equation}
where the minimization is taken over all quantum states $\sigma$ of the system $S$.
\item (Additivity) For a bipartite system $S_1S_2$ with total Hamiltonian $H_{12} = H_1\otimes \I_2 + \I_1 \otimes H_2$ and a bipartite product state $\rho_{12} = \rho_1 \otimes \rho_2$ we have $\Lambda(\beta,\rho_1\otimes \rho_2) = \Lambda_1(\beta,\rho_1)+ \Lambda_2(\beta,\rho_2)$, where the effective potentials for subsystems are defined solely with respect to the corresponding Hamiltonian on that system.
\end{enumerate}
\end{lemma}
\begin{proof}
1. For any $t \in \mathbb{R}$ we have $\Lambda(\beta,\rho(t)) = - \ln \tr(e^{-\beta H} e^{-iHt} \rho e^{+iHt})  = \Lambda (\beta,\rho)$, by cyclicity of the trace. Moreover, we can evaluate the trace in an energy eigenbasis of $H$, and so denoting by $\Pi_k$ the projector onto the $k$'th eigenspace of $H$ we have that
\begin{align}
-\ln \tr ( e^{-\beta H} \rho) &= -\ln \tr ( \sum_k \Pi_k \Pi_k (e^{-\beta H} \rho) ) \nonumber \\
&= -\ln \tr ( \sum_k  \Pi_k (e^{-\beta H} \rho) \Pi_k) \nonumber \\
&= -\ln \tr ( \sum_k  e^{-\beta H} \Pi_k \rho \Pi_k) \nonumber \\
&= -\ln\tr (e^{-\beta H}\D(\rho)) = \Lambda(\beta, \D(\rho)).
\end{align}

2. The proof of the upper bound is a direct result of Jensen's inequality, $\< e^{-\beta  H } \>_\rho \geq  e^{- \beta \< H \>_\rho} $, which implies
\begin{equation}
\Lambda (\beta, \rho)  \leq - \ln e^{-\beta \< H \>_\rho} = \beta  \< H \>_\rho,
\end{equation}
achieving the desired result.

 The lower bound is obtained as follows. Given a state $\rho$, let us define the zero mean observable $H_0 = H - \< H \>_\rho \mathds{1}$. It is readily seen that the effective potential satisfies $\Lambda (\beta , \rho) = \Lambda_0 (\beta, \rho) + \beta \< H \>_\rho$, where $\Lambda_0(\beta , \rho)$ is evaluated for the Hamiltonian $H_0$. Let us denote the random variable $\hat{X}_0$ with sample space $\text{eigs}(H_0)$, distributed by $\rho$. Then $\hat{X}_0$ has zero mean and $E^{\rm min} \leq \hat{X}_0 \leq E^{\rm max}$ almost surely. Thus we can apply Hoeffding's lemma \cite{HoeffdingLemma}, which states 
\begin{equation}
\mathbb{E} [e^{\lambda \hat{X}_0} ] \leq e^{\frac{\lambda^2}{8} (E^{\rm max} - E^{\rm min})^2}
\end{equation}
for any $\lambda \in \mathbb{R}$. Taking the logarithm and re-ordering the inequality, $\Lambda_0 (\beta , \rho) \geq -\frac{\beta^2}{8}(E^{\rm max} - E^{\rm min})^2$, which gives the lower bound. 

 3. From 2, we know that $\Lambda (\beta, \rho) \leq \beta  \< H \>_\rho  \leq \beta E^{\rm max}$ as the expectation value cannot exceed the  greatest eigenvalue. To prove the lower bound, consider a state $\rho$ with $\text{diag} (\rho) = \p$ and assume $E_k = E^{\rm min} \implies p_k \neq 0$  without loss of generality. Then
\begin{align}
\Lambda (\beta , \rho) &= - \ln \big( e^{-\beta E^{\rm min} } [ p_k + \sum_{j \neq k} p_j e^{-\beta(E_j - E^{\rm min} )} ] \big)    \\
       &= \beta E^{\rm min} - \ln \big( p_k  + \sum_{j \neq k} p_j e^{-\beta(E_j - E^{\rm min} )}\big)
\end{align} 
By assumption, $E^{\rm min }$ is the smallest eigenvalue of $H$ in the support of $\rho$ and therefore $e^{-\beta(E_n - E^{\rm min})} \leq 1$ for all $n$. From the property $\sum_i p_i = 1$,  the argument of the logarithm is smaller than one and therefore the latter term is positive. It therefore follows that $\beta E^{\rm min} \leq \Lambda (\beta , \rho) \leq \beta E^{\rm max}$. \\

\indent 4. With an appropriate choice of operators $A = - \beta H$ and $Y = \rho$, Theorem \ref{Theorem:Legendre} implies $\Lambda (\beta , \rho) = - \max_\sigma \{ -\beta \< H \>_\sigma - S(\sigma ||\rho) \} = \min_\sigma \{ \beta  \< H \>_\sigma + S(\sigma || \rho ) \}$ for states $\sigma$. However as $\ln \rho$ will diverge for non-full rank states, and so the assumption that $\rho$ must be full rank is required. \\
\indent 5. Since $H_{12} = H_1 \otimes \I_2 + \I_1 \otimes H_2$ we obtain 
\begin{align}
\Lambda (\beta, \rho_1 \otimes \rho_2) &= - \ln \tr(e^{-\beta  (H_1 \otimes \mathds{1} + \mathds{1} \otimes H_2)} \rho_1 \otimes \rho_2)  \\
     &= - \ln \Tr (e^{-\beta  H_1}\rho_1) \tr(e^{-\beta H_2 } \rho_2 )  \\
     &= \Lambda_1 (\beta ,\rho_1 ) + \Lambda_2 (\beta , \rho_2).
\end{align}
\end{proof}

The effective potential has the following properties in terms of the parameter $\beta$.
\begin{lemma} \label{lemma:EffectivePotential1} Given a quantum system $S$ with Hamiltonian $H$. For any state $\rho$ of $S$, the effective potential $\Lambda(\beta,\rho)$ obeys 
\begin{enumerate}
\item $\Lambda(\beta,\rho)$ is concave in $\beta$ for all $ \beta \ge 0$. 
\item $\Lambda(\beta,\rho)$ is monotone increasing in the variable $\beta$, for all states $\rho$, if and only if $H\ge 0$.
\item (High temperature regime) $\lim_{s\rightarrow 0} s^{-1} \Lambda(s \beta,\rho) = \beta \<H \>_\rho$.
\item  (Low temperature regime) $\lim_{s\rightarrow \infty} s^{-1} \Lambda(s \beta,\rho) = \beta E^{\rm  min}$, where $E^{\rm \tiny min}$ is the smallest eigenvalue of $H$ in the support of $\rho$.
\end{enumerate}
\end{lemma}
\begin{proof}
1. Let $a,b > 0$ satisfy $a + b =1$. For two different inverse temperatures $\beta_1$ and $\beta_2$ and Hamiltonian $H$, we have:
\begin{align}
\Lambda (a \beta_1 + b \beta_2 , \rho) &= - \ln \tr (\rho e^{-a \beta_1 H} e^{- b \beta_2 H } ).
\end{align}  
Let us define the operators $X = e^{- a \beta_1 H}$ and $Y  = e^{- b \beta_2 H}$, where $X,Y \geq 0$. Then H\"older's inequality tells us expectation values satisfy $\mathbb{E} [|XY|] \leq (\mathbb{E} [|X|^p])^{1/p} (\mathbb{E} [|Y|^q ])^{1/q}$, where $\frac{1}{p} + \frac{1}{q} = 1$. Choosing $p = \frac{1}{a}$ and $q = \frac{1}{b}$, we get:
\begin{align}
\Lambda (a \beta_1 + b \beta_2 , \rho)  \geq a \Lambda (\beta_1 ,\rho) + b \Lambda (\beta_2 , \rho).
\end{align}
Therefore $\Lambda$ is concave in the inverse temperature.
\\
\indent 2.  First suppose that $H \geq 0$, then 
\begin{equation}
\frac{\partial \Lambda (\beta , \rho)}{ \partial \beta} = \frac{\tr (H e^{-\beta H}\rho) }{\tr (e^{-\beta H}\rho)} 
\end{equation}
For $\beta \in \mathbb{R}$, the denominator is always strictly positive. The nominator can be expanded using the spectral decomposition $H = \sum_k E_k \Pi_k$ and so
\begin{equation}
\tr (H e^{-\beta H}\rho) = \sum_k E_k e^{-\beta E_k} \tr (\Pi_k \rho \Pi_k).
\end{equation}
Since $H \ge 0$ all terms in this expression are non-negative for all $\beta$ and thus $\partial_\beta (\Lambda(\beta, \rho) \ge 0$ which implies the function is monotonely increasing in $\beta$.

Conversely, assume $\partial_\beta \Lambda(\beta , \rho)  \geq 0$ for all states $\rho$. This implies that $\tr(H e^{-\beta H }\rho) \ge 0 $ for all states $\rho$. Again, using the spectral decomposition of $H$ and the fact that we have that
\begin{equation}
 \sum_k E_k e^{-\beta E_k}p_k \ge 0
\end{equation}
for all distributions $p_k$ over the eigenvalues $E_k$ of $H$. For $\p=(p_k)$ sharp on an eigenvalue $E_m$ the above equation implies that $E_m \ge 0$, and since this holds for all $m$ this in turn implies that $H \ge 0$ as required. 
 \\
\indent 3. To prove this, we Taylor expand in orders of $s$. We find:
\begin{align}
\lim_{s \to 0} s^{-1} \Lambda ( s \beta, \rho) &= \lim_{s \to 0} (-s^{-1}) \ln \< e^{-s \beta H } \>_\rho \\
     &= \lim_{s \to 0} (-s^{-1}) \ln [1 - s \beta \< H \>_\rho + \O(s^2) ] \\
     &= \lim_{s \to 0} (-s^{-1}) [-s \beta  \< H \>_\rho + \mathcal{O} (s^2)] \\
     &= \beta  \<H \>_\rho
\end{align}
as claimed. \\
\indent 4. Let $E^{\rm min}=E_k$ be the smallest energy eigenvalue in the support of $\rho$, and let $\rho_{i} := \tr (\Pi_i \rho)$ for any $i$. Then:
\begin{align}
\lim_{s \to \infty} s^{-1} \Lambda (s \beta , \rho ) = \lim_{s \to \infty} (-s^{-1} ) \ln \left( \sum_i e^{- s \beta E_i } \rho_i \right) \\
   =  \lim_{s \to \infty} (-s^{-1}) \ln \big( e^{-s \beta E^{\rm min}} [\rho_k  + \sum_{i \neq k} e^{-s \beta (E_i - E^{\rm min})} \rho_i ] \big) \nonumber \\
   = \lim_{s \to \infty} (-s^{-1}) [ -s \beta E^{\rm min} + \ln (\rho_k  + \sum_{i \neq k} e^{-s \beta (E_i - E^{\rm min})} \rho_i)]\nonumber \\
    =   \beta E^{\rm min} +\lim_{s \to \infty} (-s^{-1}) \ln (\rho_k  + \sum_{i \neq k} e^{-s \beta (E_i - E^{\rm min})} \rho_i).
\end{align}
Now since $E_i -E_k>0$ we have that $e^{-s \beta( E_i - E_{\rm min})} <1$, for all $i\ne k$ and for all $s >0, \beta >0$ and this term decreases to zero as $s \rightarrow \infty$. Moreover $\rho_k \ne 0$ by assumption, and therefore 
\begin{align}
&\lim_{s \to \infty} (-s^{-1}) \ln (\rho_k  + \sum_{i \neq k} e^{-s \beta (E_i - E^{\rm min})} \rho_i)\nonumber \\
&= \lim_{s \to \infty} (-s^{-1}) \ln (\rho_k) =0.
\end{align}
which implies that $\lim_{s \to \infty} s^{-1} \Lambda (s \beta , \rho ) = \beta E^{\rm min}$ as required.

\end{proof}

\newpage

\subsection{Mean coherence representation}

As in the main text, we have that $\chi_m$, the \emph{mean coherence at inverse temperature} $\beta$, is given by the equation
\begin{equation}\label{mean-coh}
\Lambda(\beta,\rho) = \beta \<H\>_\rho - \frac{1}{8\pi} \beta^2 \chi_m(\tilde{\rho}),
\end{equation}
where $\chi_m$ is evaluated at the inverse temperature $\beta_m \le \beta$ determined from the Mean Value Theorem.

\begin{theorem}\label{thm:mean-coherence} Given a quantum system $S$ with Hamiltonian $H_S$ with spectral decomposition $H_S = \sum_k E_k \Pi_k$, the mean coherence at inverse temperature $\beta$ of a quantum state $|\psi\>$ is given by
\begin{equation}
\frac{\beta^2}{8 \pi}\chi_m (\tilde{\psi}) =  \beta(\<H_S\>_{\p} - \<H_S\>_{\tilde{\p}}) - S(\tilde{\p} ||\p).
\end{equation}
where $\p$ is the distribution $p_k := \tr[\Pi_k \psi]$ over energy of $\psi$ and $\tilde{\p}$ is given by $ \tilde{p}_k:=\tr [ \Pi_k \tilde{\psi}]$, the distribution over energy for $|\tilde{\psi}\>\<\tilde{\psi}| = \Gamma(|\psi\>\<\psi|)$.
\end{theorem}
\begin{proof}
From Lemma \ref{lemma:EffectivePotential2} we have that $\Lambda(\beta, \rho) = \Lambda(\beta, \D(\rho))$ where $\D(\rho)$ is the dephased state across the energy eigenspaces, and so because of this it will suffice to restrict to this dephased state, as we now show.

Let $\sigma = \D(\rho)$ for compactness. Noting that $[\sigma, H]=0$ we now expand $S(\tilde{\sigma})$ as
\begin{align}
S(\tilde{\sigma}) &= - \tr [ \tilde{\sigma} \ln \tilde{\sigma}] \nonumber \\
&= - \tr \left [ \tilde{\sigma} \ln \left \{ e^{-\beta H/2} \sigma e^{-\beta H/2}\frac{1}{\tr(e^{-\beta H}\sigma)}\right \} \right] \nonumber \\ 
&= - \tr \left [ \tilde{\sigma} \left \{ \ln (e^{-\beta H/2} \sigma e^{-\beta H/2}) +\Lambda(\beta, \sigma) \right \} \right] \nonumber \\ 
&= - \tr \left [ \tilde{\sigma} \ln (e^{-\beta H/2} \sigma e^{-\beta H/2})\right ] -\Lambda(\beta, \sigma) \nonumber \\ 
&= - \tr \left [ \tilde{\sigma}(-\beta H+  \ln  \sigma) \right ] -\Lambda(\beta, \sigma) \nonumber \\ 
&=\beta \<H\>_{\tilde{\sigma}} - \tr [\tilde{\sigma}  \ln  \sigma]  -\Lambda(\beta, \sigma) \nonumber \\ 
&=\beta \<H\>_{\tilde{\sigma}} +S(\tilde{\sigma}) + S(\tilde{\sigma}||\sigma) -\Lambda(\beta, \sigma),
\end{align}
which implies that
\begin{equation}\label{effective-divergence}
\Lambda(\beta,\sigma) = \beta \<H\>_{\tilde{\sigma}} + S(\tilde{\sigma}||\sigma).
\end{equation}
Now since $\sigma = \D(\rho)$ we have that $\tr [ \Pi_k \sigma] = \tr [ \Pi_k \rho]$ and likewise it can be checked that $\tr [ \Pi_k \tilde{\sigma}] = \tr [ \Pi_k \tilde{\rho}]$, where we use that $\D(\Gamma(\rho)) = \Gamma(\D(\rho))$ for any $\rho$. Thus the preceding equation can be written
\begin{equation}
\Lambda(\beta, \D(\rho)) = \Lambda(\beta,\rho) = \beta \<H\>_{\tilde{\rho}} + S(\tilde{\sigma} ||\sigma).
\end{equation}
Restricting to $\rho = \psi$, a pure state, and using $p_k = \tr[\Pi_k \psi]$ and $\tilde{p}_k=\tr [ \Pi_k \tilde{\psi}]$ we have that
\begin{equation}
\Lambda(\beta,\psi) = \beta \<H\>_{\tilde{\p}} + S(\tilde{\p} ||\p),
\end{equation}
where we use that $\tilde{\sigma}$ commutes with $\sigma$ and so $S(\tilde{\sigma} ||\sigma)$ can be replaced with the classical divergence between the respective probability distributions over energy $S(\tilde{\p} ||\p)$. Finally, using equation (\ref{mean-coh}) we have that
\begin{equation}
\frac{1}{8\pi} \beta^2 \chi_m(\tilde{\psi}) =  \beta(\<H_S\>_{\p} - \<H_S\>_{\tilde{\p}}) - S(\tilde{\p} ||\p),
\end{equation}
as required.
\end{proof}
From this the fluctuation relation in Theorem \ref{thm:Fluctuation} can be stated in the following way.

\begin{theorem}\label{thm:RelativeEntropy}
Given the assumptions to Theorem \ref{thm:Fluctuation}, let $\p_k$ and $\tilde{\p}_k$ denote the probability distribution for an energy measurement of the state $|\psi_k \>$ and $|\tilde{\psi}_{k} \>$ respectively. Then
\begin{equation}
 \frac{P[\psi_1 | \tilde{\psi}_0]}{P[\psi_0^* | \tilde{\psi}_1^*]}  = e^{- \beta \Delta F - \beta  \tilde{W}_{S} }e^{S(\mathbf{\tilde{p}}_{0} ||\mathbf{p}_{0}) - S(\mathbf{\tilde{p}}_{1} || \mathbf{p}_{1})} ,
\end{equation}
where $ \tilde{W}_{S} := \< H_S \>_{\tilde{\psi}_{1}} - \< H_S \>_{\tilde{\psi}_{0}}  $ and $S(\mathbf{p}||\mathbf{q})$ is the Kullback-Leibler divergence.
\end{theorem}

\subsection{Variational expression for the effective potential}\label{Proof of Lemma III.3}

The following is the same as Lemma \ref{free-energy-like} in the main text.

\begin{lemma} Given a quantum system $S$ with Hamiltonian $H_S$ and $\{\Pi_k\}$ the projectors onto the energy eigenspaces of $H_S$. Let us denote the de-phased form of state in the energy basis  by $\D(\rho) = \sum_k \Pi_k \rho \Pi_k$. Then, for any full rank quantum state $\rho$, the effective potential $\Lambda(\beta, \rho)$ is given by
\begin{equation}
\Lambda(\beta, \rho) = \min_\tau \{   \beta \<H_S\>_\tau + S(\tau||\D(\rho)) \},
\end{equation}
where the minimization is over all quantum states $\tau$ of $S$, and $S(\tau||\rho) = \tr [ \tau \log \tau - \tau \log \rho]$ is the relative entropy function. Moreover, the minimization is attained for the state $\tau = \Gamma(\D(\rho))$. 
\end{lemma}
\begin{proof}
From Lemma \ref{lemma:EffectivePotential2}, we see that for any full rank $\rho$ we have
\begin{align}
\hspace{-0.5cm}\Lambda(\beta, \rho) &= \Lambda(\beta, \D(\rho) ) \\
 &=  \min_\tau [  \beta \< H_S\>_\tau   + S(\tau||\D(\rho))].\nonumber
\end{align}
However equation (\ref{effective-divergence}) says
\begin{equation}
\Lambda(\beta, \rho) =\Lambda (\beta , \sigma) = \beta \<H_S\>_{\tilde{\sigma}} + S(\tilde{\sigma}||\sigma),
\end{equation}
where $\sigma = \D(\rho)$. Comparing the two equations we see that the minimization is realised for $\tau = \tilde{\sigma}= \Gamma (\sigma) = \Gamma(\D(\rho))$ as claimed.

\end{proof}

\subsection{Coherent work and the quantum fluctuation relation.}\label{app:coherent work}

We now prove the coherent work form of the coherent fluctuation theorem and provide an account of how it naturally connects with the resource theory of asymmetry, inducing a natural measure for coherent work. However, first we require the conditions in which the moments of a probability distribution uniquely determine the distribution.

\begin{theorem}[\cite{BillingsleyProbability}]
Let $\mu$ be a probability measure on the line having finite moments $m_k = \int_{-\infty}^\infty x^k \mu(dx)$ of all orders. If the power series $\sum_{k = 1}^\infty \frac{r^k }{k!} m_k$ has positive radius of convergence, then $\mu$ is the unique probability measure with the moments $m_1,m_2,\dots$.
\end{theorem}
When the full set of moments uniquely define the probability distribution, then so too do the cumulants. This follows since in such cases the moments can be expressed in terms of the cumulants. Furthermore, finite moments correspond to finite cumulants since once can obtain the cumulants from a recursive formula that is polynomial in the moments. Thus the full set of cumulants uniquely determines a distribution when the moment generating function of the distribution has a positive radius of convergence. When the moment generating function is finite and non-zero, then the cumulant generating function is finite and non-zero as it is obtained by taking the logarithm.

In the result that follows, we assume the moment generating function and by extension the effective potential, have infinite radius of convergence and therefore always exist and uniquely specify the distribution. This choice is for simplicity in the result, though one could generalise it to hold for any positive radius of convergence. Requiring an infinite radius of convergence is a restriction on the states permitted. For example, the negative binomial distribution $\text{NB}(r,p)$ has finite radius of convergence and therefore any state with such a distribution would not be included in the statement. This can be seen from the moment generating function $M(t)$ for a random variable $\hat{X} \sim \text{NB}(r,p)$, with \cite{Lukacs}
\begin{equation}
M(t) = \left(\frac{p}{1 - (1-p)e^t}\right)^r,
\end{equation}
which is valid for $t <  - \ln(1-p)$.

\begin{theorem}
Let $S$ be a quantum system with Hamiltonian $H_S$ and let $|\psi_0\>_S$ and $|\psi_1\>_S$ be two pure states of $S$ that have energy statistics with finite cumulants of all orders. Then the following are equivalent:
\begin{enumerate}
\item The pure states $|\psi_0\>_S$ and $|\psi_1\>_S$ are coherently connected with coherent work output/input $\omega$ on an auxiliary system $A$.
\item Within the fluctuation relation context with a thermal system $B$ and quantum system $S$ we have that
\begin{equation}
\frac{P[\psi_1 | \tilde{\psi}_0]}{P[\psi^*_0 | \tilde{\psi}^*_1]}   = e^{-\beta(\Delta F \pm kT\Lambda(\beta,\omega))},
\end{equation} 
for all inverse temperatures $\beta \ge 0$, for some state $|\omega\>_A$ with finite cumulants on an auxiliary system $A$ and some choice of sign before $\Lambda(\beta, \omega)$.
\end{enumerate}
 
\end{theorem}

\begin{proof}
First assume that $|\psi_0\>_S$ and $|\psi_1\>_S$ are coherently connected states. Suppose there is a coherent work process transforming $|\psi_0 \>_S \stackrel{\omega} \longrightarrow |\psi_1 \>_S$. Without loss of generality we can choose the initial state $|0\>_A$ of the auxiliary system $A$ to have zero energy.

By unitary invariance of the effective potential we have that
\begin{align}
\Lambda(\beta, |\psi_0\>_S \otimes |0\>_A ) &=\Lambda(\beta, V(|\psi_0\>_A \otimes |0\>_A)) \nonumber \\
&=\Lambda (\beta, |\psi_1\>_A \otimes |\omega\>_A),
\end{align}
where $V$ is a unitary that realises the process. This conserves energy and so commutes with $e^{-\beta (H_S \otimes \mathds{1}_A + \mathds{1}_S \otimes H_A)}$. Since the effective potential is additive for independent Hamiltonians and product states, we have that
\begin{equation}\label{coherent work effective potential}
\Lambda(\beta, \psi_0) = \Lambda(\beta,\psi_1) + \Lambda ( \beta, \omega).
\end{equation}
Therefore the effective potential change in the fluctuation theorem can be expressed in terms of the coherent work output $|\omega\>$ as
\begin{equation}
e^{\Lambda(\beta, \psi_0) - \Lambda(\beta, \psi_1)} = e^{\Lambda(\beta, \omega)}.
\end{equation}
and so the fluctuation theorem can be expressed as
\begin{equation}
\frac{P[\psi_1|\tilde{\psi}_0]}{P[\psi_0^*|\tilde{\psi}_1^*]}   = e^{\Lambda(\beta, \omega) - \beta \Delta F} = \frac{e^{-\beta \Delta F}}{\< \omega| e^{-\beta H_A} |\omega\>}.
\end{equation}
If instead $|\psi_1 \>_S \stackrel{\omega} \longrightarrow |\psi_0 \>_S$, then the same argument applies with the only difference being the sign in front of $\Lambda(\beta, \omega)$.

Conversely, suppose the above statement 2 holds for all inverse temperatures $\beta \ge 0$, for some pure state $|\omega\>_A$ on a system $A$, and choice of sign in the exponent. The exponent in Theorem \ref{thm:Fluctuation} implies that we have 
\begin{equation}
\Lambda(\beta, \psi_0) = \Lambda(\beta,\psi_1) \mp \Lambda ( \beta, \omega),
\end{equation}
 for all inverse temperatures $\beta \in [0,\infty)$. In the event of the minus sign, we write this as
 \begin{align}
\Lambda(\beta, \psi_0) + \Lambda ( \beta, \omega) &= \Lambda(\beta, |\psi_0\>_A \otimes |\omega\>_A) \nonumber \\
&= \Lambda(\beta,|\psi_1\>_S \otimes |0\>_A),
\end{align}
where we can assume that $A$ has a zero energy eigenstate $|0\>_A$. However, since all the above cumulant generating functions have infinite radius of convergence and are finite, the energy distributions of measuring $|\psi_0\>_S$, $|\psi_1\>_S$ and $|\omega\>_A$ are uniquely determined. Moreover the distributions of a joint energy measurement of composite system $SA$ in the state  $|\psi_0\>_S \otimes |\omega\>_A$ and $|\psi_1\>_S \otimes |0\>_A$ are identical. This implies that
\begin{equation}
\hat{S}_0 + \hat{A} = \hat{S}_1,
\end{equation}
where $\hat{S}_0, \hat{S}_1, \hat{A}$ are the associated random variables for measuring $|\psi_0\>_S, |\psi_1\>_S, |\omega\>_A$ respectively in their energy eigenbases. By Theorem \ref{thm:Decomposable} this implies there exists a non-trivial coherent work process 
\begin{equation}
|\psi_1\> \stackrel{\omega}{\longrightarrow} |\psi_0\>,
\end{equation}
and so the states are coherently connected. In the event of the plus sign in the exponent we instead write
\begin{equation}
\Lambda(\beta, \psi_0) = \Lambda(\beta,\psi_1) + \Lambda ( \beta, \omega),
\end{equation}
and then express the argument of the left-hand side using $|\psi_0\>_S \otimes |0\>_A$ as before, using the additivity of the effective potential for the left-hand side. The same argument as before then implies the existence of a coherent work process
\begin{equation}
|\psi_0\> \stackrel{\omega}{\longrightarrow} |\psi_1\>,
\end{equation}
and again the states are coherently connected, which completes the proof.

\end{proof}

\subsection{Proof of semi-classical fluctuation theorem}

\begin{theorem}[\textbf{Semi-classical relation }\cite{Zoe}]
Let the assumptions of Theorem \ref{thm:Fluctuation} hold, and let $S$ be a harmonic oscillator with Hamiltonian $H_S = h\nu (a^\dagger a + \frac{1}{2})$. Then for two coherent states of the system $|\alpha_0\>$ and $|\alpha_1\>$, the following holds
\begin{align}
\frac{P[\alpha_1|\tilde{\alpha}_0]}{P[\alpha_0^*| \tilde{\alpha}_1^*]}  &=  \exp \left [ -\frac{\Delta F }{k T}+ \frac{\bar{W}_B}{h \nu_{\rm \tiny{th}}}  \right],
\end{align}
where  $h \nu_{th} = \< H_S \>_\gamma$ is the average energy of a Gibbs state $\gamma$ of a quantum harmonic oscillator, related to the thermal de Broglie wavelength $\lambda_{\rm \tiny{dB}} (T)$ via
\begin{equation}
h \nu_{\rm \tiny{th}} = \frac{h^2}{m \lambda_{\rm \tiny{dB}}(T)^2} + \frac{1}{2} h\nu,
\end{equation}
and $\bar{W}_B := - \frac{1}{2}(W_S + \tilde{W}_S)$, $W_S = \<H_S\>_{\alpha_1} - \<H_S\>_{\alpha_0}$, $\tilde{W}_S = \<H_S\>_{\tilde{\alpha}_1} - \<H_S\>_{\tilde{\alpha}_0}$.
\end{theorem}
\begin{proof}

Consider a system with Hamiltonian $H_S = h \nu (a^\dagger a + \frac{1}{2})$. For simplicity, we take an arbitrary state $|\alpha , k \> \in \mathcal{C}$ and assume $k = 0$, as the extension is simple but provides little benefit. The effective potential for a coherent state $|\alpha\>$:
\begin{align}
\Lambda ( \beta, \alpha ) &= \frac{\beta h \nu }{2} + |\alpha|^2 (1- e^{- \beta h \nu } ) .
\end{align}
Now consider the average energy of a thermally distributed quantum harmonic oscillator
\begin{align}
\< H_S \>_\gamma &= \frac{\sum_{n \geq 0} h\nu  (n + \frac{1}{2}) e^{-\beta h \nu (n +\frac{1}{2})}}{\sum_{m \geq 0} e^{-\beta h \nu (m +\frac{1}{2})}}  \\
&= \frac{h \nu}{2} \left( \frac{1 + e^{-\beta h \nu}}{1 - e^{-\beta h \nu} } \right).
\end{align}
Let us define a thermal frequency $h \nu_{th} := \< H_S \>_\gamma$. Then the effective potential change can be expressed as
\begin{align}
\Lambda(\beta , \alpha_1) - \Lambda(\beta ,\alpha_0) &= (|\alpha_1|^2 - |\alpha_0|^2) (1 - e^{-\beta h \nu}) \\ 
    &= \frac{ \nu}{2 \nu_{th}} (|\alpha_1|^2 - |\alpha_0|^2)(1+ e^{-\beta h \nu}).
\end{align}
Since the Gibbs rescaled coherent state $|\tilde{\alpha} \> = |\alpha e^{-\beta h \nu/2} \>$, the change in effective potential can be expressed in terms of the unscaled and rescaled coherent state energy transfers,
\begin{align}
 W_S &= \< H_S \>_{\alpha_1} - \< H_S \>_{\alpha_0}  \\
\tilde{W}_S &= \< H_S \>_{\tilde{\alpha}_1} - \< H_S \>_{\tilde{\alpha}_0},
\end{align}  
as
\begin{equation}
\Lambda(\beta , \alpha_1) - \Lambda(\beta ,\alpha_0) = \frac{1}{h \nu_{th}}\left(\frac{W_S + \tilde{W}_S}{2} \right).
\end{equation}   
which gives the claimed result upon insertion into Theorem \ref{thm:Fluctuation}. For the breakdown of $h \nu_{th}$ in terms of $\lambda_{dB}(T)$, we refer to \cite{Zoe}.
\end{proof}

For completeness, we can also show the form of the fluctuation theorem using the coherent work representation. Consider an auxillary weight system $A$ with Hamiltonian $H_A = h \nu a^\dagger a$. Once again, this system exists as a hypothetical process in which $|\psi_0\> \stackrel{\omega}\longrightarrow |\psi_1\>$ under a coherent work process. We note that
\begin{equation}
e^{\Lambda(\beta,\omega)} = \frac{\< \alpha_1 |e^{-\beta H_S} |\alpha_1 \>}{\< \alpha_0 | e^{-\beta H_S} |\alpha_0 \>}.
\end{equation}
From the semi-classical result Theorem \ref{thm:Semi Classical}, we know that $|\alpha_0 \>_S \stackrel{\omega} \longrightarrow |\alpha_1 \>_S$ under a coherent work process if and only if $|\alpha_1| \leq |\alpha_0|$. Assuming this condition to be true, then we can simply write down:
\begin{equation}
\frac{P[\alpha_1|\tilde{\alpha}_0]}{P[\alpha_0^*|\tilde{\alpha}_1^*]}   = e^{|\alpha|^2 (1 - e^{-\beta h \nu}) - \beta \Delta F}
\end{equation}
where $|\alpha|^2 = |\alpha_0|^2 - |\alpha_1|^2$ and $||\alpha|\>$ is the coherent work output of the hypothetical process. Owing to the simple form of the effective potential for coherent states, this is fairly trivial to write down. However generally, the effective potential fails to  obtain such a simplified form, and simplifying the change in effective potential $\Lambda(\beta,\psi_1) - \Lambda(\beta,\psi_0)$ may be harder than finding the analytic form for $\Lambda (\beta , \omega)$.

\subsection{Macrosopic limit of the quantum fluctuation relation}

For simplicity, we choose to make the random variable dimensionless. We have that $\<\psi_i |H_i |\psi_i\> = \mu$, which is independent of $i$. We define the operator $X_i :=\frac{1}{\mu}H_i - \I$, which has mean zero in the state $ |\psi_i\>$. Note that the variance of $X_i$ is also dimensionless and given by
\begin{align}
\<X_i^2\> - \<X_i\>^2 &= \<X_i^2\> \\
&= \< \frac{1}{\mu^2} H_i^2 - \frac{2}{\mu} H_i + \I \>\\
&= \frac{1}{\mu^2}\<H_i^2\> - 2 + 1 \\
&=\frac{1}{\mu^2}( \<H_i^2\> - \mu^2) \\
\text{Var}(X_i)&= \frac{\sigma^2}{\mu^2} \equiv \tilde{\sigma}^2.
\end{align}

The cumulant generating function of the total state is
\begin{align}
e^{-\Lambda(\beta, \psi_n)} &= \tr [ e^{-\beta \sum_i H_i} \psi^{\otimes n}] \\
&= \tr [ e^{-\beta \mu \sum_i (X_i + \I)} \psi^{\otimes n}] \\
&= e^{-\beta n \mu} \tr [ e^{-\beta \mu \sum_i X_i } \psi^{\otimes n}] \\
&= e^{-\beta n \mu} \tr [ e^{-\beta \sqrt{n} \mu \sum_i \frac{1}{\sqrt{n}}X_i } \psi^{\otimes n}] \\
&= e^{-\beta n \mu} e^{-\tilde{\Lambda}(\beta \sqrt{n} \mu, \psi_n)}.
\end{align}
Here $\tilde{\Lambda}$ is the CGF corresponding to the dimensionless mean zero observable $\frac{1}{\sqrt{n}} (X_1 + \dots + X_n)$. Thus we have that
\begin{equation}
\Lambda(\beta, \psi_n) = \beta n \mu + \tilde{\Lambda}(\beta \sqrt{n} \mu, \psi_n),
\end{equation}
for all $\beta, n$.
\\
\indent
A subtlety comes in here -- if one wants to take a limit  and obtain a Central Limit Theorem (CLT) statement we need to scale $\beta$ also in order to get a non-trivial statement. This is a mathematical re-formulation so as to maintain the correct scaling to obtain a non-trivial statement the mean value of the sequence (in $n$) of random variables.
\\
\indent
Since $\beta$ has units of reciprocal energy we can set $\beta = \frac{\beta_0}{\sqrt{n} \lambda}$, for some arbitrary energy scale $\lambda$, where $\beta_0$ is dimensionless and we count in energy units of $\sqrt{n}\lambda$ when we have $n$ systems. It might look like the inverse temperature is decreasing as we increase the number of systems, but this is not what is intended by this -- in the description of the sequence of systems one scales the units of energy as we increase $n$ when we specify the temperature. Note also that the first and second cumulants scale differently with energy, and it is this that makes the scaling trick needed to obtain a statement containing both in the large $n$ limit.
\\
\indent
Expressed in these units we have
\begin{equation}
\Lambda(\frac{\beta_0}{\sqrt{n}\lambda}, \psi_n) = \beta_0 \mu_0 \sqrt{n}  + \tilde{\Lambda}(\beta _0\mu_0, \psi_n),
\end{equation}
where $\mu_0:=\frac{\mu}{\lambda}$ is dimensionless, and $\beta_0$ does not vary in $n$.

This holds for all $n$, and in the $n \rightarrow \infty$ limit we have that
\begin{equation}
\lim_{n\rightarrow \infty} \left (\Lambda(\frac{\beta_0}{\sqrt{n}\lambda}, \psi_n) - \beta_0\mu_0(\psi)\sqrt{n} \right)  = - \frac{1}{2} \tilde{\sigma}_\psi^2 \beta_0^2\mu_0(\psi)^2,
\end{equation}
and also
\begin{equation}
\lim_{n\rightarrow \infty} \left (\Lambda(\frac{\beta_0}{\sqrt{n}\lambda}, \phi_n) - \beta_0\mu_0(\phi)\sqrt{n} \right)  = - \frac{1}{2} \tilde{\sigma}_\phi^2 \beta_0^2\mu_0(\phi)^2,
\end{equation}
where we have included the state variables to distinguish the energies of the two states.

This implies that for any $\epsilon >0$ there is an integer $M$ such that for all $n >M$ we have the following two conditions hold
\begin{align}
\left | \Lambda(\frac{\beta_0}{\sqrt{n}\lambda}, \psi_n) - [\beta_0\mu_0(\psi)\sqrt{n}   - \frac{1}{2} \sigma_\psi^2 \beta_0^2 ]\right | &\le \epsilon \\
\left | \Lambda(\frac{\beta_0}{\sqrt{n}\lambda}, \phi_n) - [\beta_0\mu_0(\phi)\sqrt{n}   - \frac{1}{2} \sigma_\phi^2 \beta_0^2 ]\right | &\le \epsilon,
\end{align}
where we also used the definition of $\tilde{\sigma}^2$.
We therefore have that for any $\epsilon >0$ there is an $M>0$ such that for all $n > M$ we have
\begin{equation}
|\Lambda(\frac{\beta_0}{\sqrt{n}\lambda} ,\psi_n) - \Lambda(\frac{\beta_0}{\sqrt{n}\lambda}, \phi_n)  - [\beta_0 \sqrt{n} \Delta \mu_0 - \frac{1}{2}\beta_0^2\Delta \sigma_0^2] | \le 2\epsilon.
\end{equation}

Thus, for any $\epsilon >0$ there is an $M$ such that for all $n >M$ we have
\begin{align}
\left |  \frac{P[\psi_n | \phi_n]}{P[\phi^*_n | \psi^*_n]}  - e^{-\sqrt{n} (\beta_0 \Delta f_0 +\beta_0 \Delta \mu_0) + \frac{1}{2}\beta_0^2\Delta \sigma_0^2} \right | \le \epsilon
\end{align}
where $\Delta f_0 := \frac{\Delta F}{n\lambda}$ is a dimensionless number. We could put back in $\beta$, but only with the proviso that it is assumed to scale with $n$ in the way stated. This gives
\begin{align}
\left |  \frac{P[\psi_n | \tilde{\phi}_n]}{P[\phi^*_n | \tilde{\psi}^*_n]}  - e^{-n (\beta \Delta f +\beta \Delta \mu - \frac{1}{2}\beta^2\Delta \sigma^2)} \right | \le \epsilon
\end{align}
with the proviso that the scaling $\beta = \frac{\beta_0}{\sqrt{n} \lambda}$ is adopted. This captures the informal statement that for the IID case the states ``become more like Gaussians'' and thus the fluctuation theorem behaves as
\begin{equation}
 \frac{P[\psi_n | \tilde{\phi}_n]}{P[\phi^*_n | \tilde{\psi}^*_n]}  \stackrel{n \rightarrow \infty}\sim e^{-n (\beta \Delta f +\beta \Delta \mu - \frac{1}{2}\beta^2\Delta \sigma^2)}.
\end{equation}

\end{document}